\documentclass[runningheads, envcountsame]{llncs}

\usepackage{csquotes}
\usepackage{geometry}
\usepackage{bm,mathptmx,booktabs}
\usepackage{tcolorbox}
\usepackage[usestackEOL]{stackengine}
\usepackage{amsmath,amssymb,amsfonts}
\usepackage{thm-restate}

\usepackage{ifthen}
\usepackage{comment}
\makeatletter
\usepackage{xcolor,tikz,environ}
\usetikzlibrary{decorations.markings, shapes.geometric}

\usepackage{etoolbox}

\usepackage{orcidlink}

\usepackage{amsmath,amsfonts,amssymb}
\usepackage{mathtools}
\usepackage{mathrsfs}
\usepackage{wrapfig}
\usepackage{xspace}

\usepackage[shortlabels]{enumitem}
\setlist[enumerate]{nosep}
\setlist[itemize]{nosep}

\usepackage[protrusion=true]{microtype}

\spn@wtheorem{observation}{Observation}{\bfseries}{\itshape}

\newcommand{\ignore}[1]{}

\makeatother

\usepackage{framed}

\usepackage{color,tikz,environ}
\usetikzlibrary{decorations.markings,arrows,plotmarks}
\tikzset{nomorepostaction/.code={\let\tikz@postactions\pgfutil@empty}}

\tikzset{middlearrow/.style={
    decoration={markings,
    mark= at position 0.5 with {\arrow{#1}} ,
    },
    postaction={decorate}
}
}

\tikzset{onethirdarrow/.style={
    decoration={markings,
    mark= at position 0.33 with {\arrow{#1}} ,
    },
    postaction={decorate}
}
}

\tikzset{twothirdarrow/.style={
    decoration={markings,
    mark= at position 0.67 with {\arrow{#1}} ,
    },
    postaction={decorate}
}
}

\tikzset{endarrow/.style={
    decoration={markings,
    mark= at position 0.9 with {\arrow{#1}} ,
    },
    postaction={decorate}
}
}

\tikzset{startarrow/.style={
    decoration={markings,
    mark= at position 0.1 with {\arrow{#1}} ,
    },
    postaction={decorate}
}
}

\usepackage{manyfoot}
\RequirePackage{multirow}
\usepackage{makecell}
\usepackage{graphicx}

\newcommand{\gt}{\textsc{Grid-Tiling}\xspace}
\newcommand{\gtleq}{\textsc{Grid-Tiling-}$\leq$\xspace}
\newcommand{\edp}{\textsc{Edge-Disjoint-Paths}\xspace}
\newcommand{\kclique}{\textsc{\texorpdfstring{$k$}{k}-Clique}\xspace}

\newcommand{\vdp}{\textsc{Vertex-Disjoint-Paths}\xspace}

\newcommand{\Le}{\texttt{Left}}
\newcommand{\Ri}{\texttt{Right}}
\newcommand{\To}{\texttt{Top}}
\newcommand{\Bo}{\texttt{Bottom}}
\newcommand{\LB}{\text{LB}}
\newcommand{\TR}{\text{TR}}
\newcommand{\lefty}{\texttt{west}}
\newcommand{\righty}{\texttt{east}}
\newcommand{\topy}{\texttt{north}}
\newcommand{\bottomy}{\texttt{south}}
\newcommand{\poly}{\text{poly}}
\newcommand{\w}{\textbf{w}}
\newcommand{\x}{\textbf{x}}

\newcommand{\Matching}{\texttt{Matching}}
\newcommand{\Sink}{\texttt{Sink}}
\newcommand{\Source}{\texttt{Source}}
\newcommand{\Row}{\texttt{RowPath}}
\newcommand{\Column}{\texttt{ColumnPath}}
\newcommand{\Columna}{\texttt{Column}}
\newcommand{\z}{\textbf{z}}
\newcommand{\edgesplitt}{\texttt{Edge}}
\newcommand{\vertsplitt}{\texttt{Vertex}}

\newcommand{\dsp}{\textsc{Disjoint-Shortest-Paths}\xspace}
\newcommand{\DSP}{\text{Disjoint-Shortest-Paths}\xspace}
\newcommand{\disjp}{\textsc{Disjoint-Paths}\xspace}
\newcommand{\edsp}{\textsc{Edge-Disjoint-Shortest-Paths}\xspace}
\newcommand{\vdsp}{\textsc{Vertex-Disjoint-Shortest-Paths}\xspace}
\newcommand{\EDSP}{\text{EDSP}\xspace}
\newcommand{\EDP}{\textsc{EDP}\xspace}
\newcommand{\VDSP}{\text{VDSP}\xspace}

\newcommand{\edge}{\text{edge}}
\newcommand{\vertex}{\text{vertex}}
\newcommand{\inter}{\text{int}}
\newcommand{\Mid}{\text{Mid}}
\newcommand{\Hor}{\text{Hor}}
\newcommand{\Ver}{\text{Ver}}
\newcommand{\Clique}{\textsc{Clique}\xspace}
\newcommand{\onesplit}{\texttt{one-split}\xspace}
\newcommand{\twosplit}{\texttt{two-split}\xspace}
\newcommand{\vertsplit}{\texttt{vertex-split}\xspace}
\newcommand{\notsplit}{\texttt{not-split}\xspace}

\newcommand{\UEDSP}{Undirected-$k$-\textsc{EDSP}\xspace}
\newcommand{\dedsp}{\textsc{Directed-$k$-Edge-Disjoint-Shortest-Paths}\xspace}
\newcommand{\DEDSP}{Directed-\texorpdfstring{$k$}{k}-\textsc{EDSP}\xspace}

\newcommand{\UVDSP}{Undirected-$k$-\textsc{VDSP}\xspace}
\newcommand{\dvdsp}{\textsc{Directed-$k$-Vertex-Disjoint-Shortest-Paths}\xspace}
\newcommand{\DVDSP}{Directed-$k$-\textsc{VDSP}\xspace}

\newcommand{\CanInter}{\textsc{Canonical}\textsubscript{int}\textsuperscript{D}}
\newcommand{\CanEdge}{\textsc{Canonical}\textsubscript{edge}\textsuperscript{D}}
\newcommand{\CanVertex}{\textsc{Canonical}\textsubscript{vertex}\textsuperscript{D}}
\newcommand{\HorizontalInter}{\textsc{Horizontal}\textsubscript{int}\textsuperscript{D}}
\newcommand{\HorizontalEdge}{\textsc{Horizontal}\textsubscript{edge}\textsuperscript{D}}
\newcommand{\HorizontalVertex}{\textsc{Horizontal}\textsubscript{vertex}\textsuperscript{D}}
\newcommand{\VerticalInter}{\textsc{Vertical}\textsubscript{int}\textsuperscript{D}}
\newcommand{\VerticalEdge}{\textsc{Vertical}\textsubscript{edge}\textsuperscript{D}}
\newcommand{\VerticalVertex}{\textsc{Vertical}\textsubscript{vertex}\textsuperscript{D}}

\newcommand{\Gint}{\texorpdfstring{$D_{\inter}$\xspace}{Dint}}
\newcommand{\GintU}{\texorpdfstring{$U_{\inter}$\xspace}{Uint}}
\newcommand{\Gedge}{\texorpdfstring{$D_{\edge}$\xspace}{Dedge}}
\newcommand{\GedgeU}{\texorpdfstring{$U_{\edge}$\xspace}{Uedge}}
\newcommand{\Gvert}{\texorpdfstring{$D_{\vertex}$\xspace}{Dvert}}
\newcommand{\GvertU}{\texorpdfstring{$U_{\vertex}$\xspace}{Uvert}}

\newcommand{\CanInterU}{\textsc{Canonical}\textsubscript{int}\textsuperscript{U}}
\newcommand{\CanEdgeU}{\textsc{Canonical}\textsubscript{edge}\textsuperscript{U}}
\newcommand{\CanVertexU}{\textsc{Canonical}\textsubscript{vertex}\textsuperscript{U}}

\AtBeginDocument{}
\AtBeginDocument{}
\AtBeginDocument{}

\definecolor{darkblue}{rgb}{0,0,1}
\definecolor{darkred}{rgb}{0.6,0,0}
\definecolor{darkgreen}{rgb}{0,0.6,0}
\definecolor{darkyellow}{rgb}{1,0.75,0.03}

\definecolor{red}{rgb}{1,0,0}
\newcommand{\red}[1]{{\textcolor{red} { #1}}}
\definecolor{green}{HTML}{89CE00}
\newcommand{\green}[1]{{\textcolor{green} { #1}}}
\definecolor{magenta}{HTML}{e6308a}
\newcommand{\magenta}[1]{{\textcolor{magenta} { #1}}}
\newcommand{\black}[1]{{\textbf{#1}}}

\usepackage{hyperref}

\usepackage{float}

\begin{document}

	\title{Lower Bounds for Approximate (\& Exact) \texorpdfstring{$k$}{k}-\dsp}

	\titlerunning{Lower Bounds for Approximate (\& Exact) \texorpdfstring{$k$}{k}-\dsp}

	\author{Rajesh Chitnis\inst{1}\textsuperscript{\orcidlink{0000-0002-6098-7770}} \and
	Samuel Thomas\inst{1,2}\textsuperscript{\orcidlink{0000-0002-6278-6833}} \and
	Anthony Wirth\inst{2,3}\textsuperscript{\orcidlink{0000-0003-3746-6704}}\\
	\email{rajeshchitnis@gmail.com}\\
	\email{samuelthomascs@gmail.com}\\
	\email{anthony.wirth@sydney.edu.au}}

	\authorrunning{R. Chitnis, S. Thomas and A. Wirth}

	\institute{School of Computer Science, The University of Birmingham, United Kingdom \and
	School of Computing and Information Systems, The University of Melbourne, Australia \and School of Computer Science, The University of Sydney, Australia}

	\maketitle

\begin{abstract}
    Given a graph $G=(V,E)$ and a set $\mathcal{T}=\big\{ (s_i, t_i) : 1\leq i\leq k \big\}\subseteq V\times V$ of $k$ pairs, the $k$-\vdp (resp.\ $k$-\edp) problem asks to determine whether there exist~$k$ pairwise vertex-disjoint (resp.\ edge-disjoint) paths $P_1, P_2, \ldots, P_k$ in $G$ such that, for each $1\leq i\leq k$, $P_i$ connects $s_i$ to $t_i$.
    Both the edge-disjoint and vertex-disjoint versions in undirected graphs are famously known to be FPT (parameterized by $k$) due to the Graph Minor Theory of Robertson and Seymour.

    Eilam-Tzoreff [DAM `98] introduced a variant, known as the $k$-\dsp problem, where each individual path is further required to be a shortest path connecting its pair. They showed that the
    $k$-\dsp problem is NP-complete on both directed and undirected graphs; this holds even if the graphs are planar and have unit edge lengths.
    We focus on four versions of the problem, corresponding to considering edge/vertex disjointness, and to considering directed/undirected graphs.
    Building on the reduction of Chitnis [SIDMA `23] for $k$-\edp on planar DAGs, we obtain the following \emph{inapproximability lower bound} for each of the four versions of $k$-\dsp on $n$-vertex graphs:
    \begin{itemize}
        \item Under the gap version of the Exponential Time Hypothesis (Gap-ETH), there exists a constant $\delta>0$ such that for any constant $0<\epsilon\leq \frac{1}{2}$ and any computable function $f$, there is no  $(\frac{1}{2}+\epsilon)$-approximation\footnote{An $\alpha$-approximation for $k$-\dsp distinguishes between these two cases: either (i) all $k$ pairs can be satisfied; or (ii) the maximum number of pairs that can be satisfied is less than $\alpha\cdot k$.} in $f(k)\cdot n^{\delta\cdot k}$ time.
    \end{itemize}
    We provide a single, \textbf{unified framework} to obtain lower bounds for \emph{each of the four versions} of $k$-\dsp.
    We are able to further strengthen our results by restricting the structure of the input graphs in the lower bound constructions as follows:
    \begin{itemize}
        \item \underline{Directed}:
        The inapproximability lower bound for edge-disjoint (resp. vertex-disjoint) paths holds even if the input graph is a planar (resp.\ 1-planar) DAG with max in-degree and max out-degree at most $2$.
        \item \underline{Undirected}:
        The inapproximability lower bound for edge-disjoint (resp. vertex-disjoint) paths hold even if the input graph is planar (resp.\ 1-planar) and has max degree $4$.
    \end{itemize}
    The reductions outlined in this paper produce graphs in which half of the terminal pairs are trivially satisfiable, so any improvement of our $(\frac{1}{2} + \epsilon)$ inapproximability factor requires a different approach.

    As a byproduct of our reductions, we also show that the exact versions of each problem is $W[1]$-hard and give a $f(k)\cdot n^{o(k)}$-time lower bound for them under ETH.
    This exact lower bound shows that the $n^{O(k)}$-time algorithms of Bérczi and Kobayashi [ESA `17] for \DEDSP and \DVDSP are tight.

\end{abstract}

\section{Introduction}
\label{sec:introduction-app}

The $k$-\disjp problem is one of the oldest and best-studied in graph theory: given a graph on $n$ vertices and a set of $k$ terminal pairs, the question is to determine whether there exists a collection of $k$ pairwise-disjoint paths where each path connects one of the given terminal pairs. There are four versions of the $k$-\disjp problem depending on whether the underlying graph is undirected or directed, and whether the paths are required to be pairwise edge-disjoint or vertex-disjoint. The undirected \disjp problem is a fundamental ingredient in the algorithmic Graph Minor Theory of Robertson and Seymour: they designed an algorithm~\cite{DBLP:journals/jct/RobertsonS95b} for $k$-\disjp which runs in $f(k)\cdot n^{3}$ time for some function $f$, i.e., an FPT algorithm parameterized by the number $k$ of terminal pairs. The dependence on~$n$ was improved from cubic to quadratic by Kawarabayashi et al., who designed an algorithm running in $g(k)\cdot n^2$ time for some function $g$~\cite{DBLP:journals/jct/KawarabayashiKR12}. However, functions $f$ and $g$ are rapidly growing and this led to the development of faster algorithms (with explicit bounds) for the special case of planar graphs~\cite{DBLP:journals/jct/AdlerKKLST17,DBLP:conf/stoc/LokshtanovMP0Z20}.

In this paper, we focus on a variant of the $k$-\disjp problem, called the $k$-\dsp problem, where there is an additional requirement that each of the paths must be a shortest path for the terminal pair that it connects. This problem was introduced by Eilam-Tzoreff~\cite{eilam}. There are four versions of the $k$-\dsp problem, depending on whether we require edge-disjointness or vertex-disjointness and whether the input graph is directed or undirected. The $k$-\dsp problem arises in several real-world scenarios, such as effective packet switching~\cite{ogierDistributedAlgorithmsComputing1993,srinivasFindingMinimumEnergy2005} and integrated circuit design~\cite{frankPackingPaths1990,robertsonOutlineDisjointPaths}.

\subsection{Organization of the paper}
\label{sec:organization}

We first briefly survey some of the known results for $k$-\dsp on directed graphs (\autoref{sec:prior-directed}) and undirected graphs (\autoref{sec:prior-undirected}) before stating our results in~\autoref{sec:our-result}.

Our results in this paper are all obtained by reductions from known hardness results for the \kclique problem (\autoref{sec:kclique-results}). For each of the four versions of the $k$-\dsp problem, a similar template (see~\autoref{fig:flowchart} for a visual depiction) is followed that entails firstly obtaining an intermediate graph from an instance of \kclique before then applying an operation to vertices of that graph.
\autoref{sec:our-result} gives details of our theorems, whilst~\autoref{sec:app-notation} explains our graph-theoretic notation. Organization of the later sections is as follows:
\begin{itemize}
    \item \underline{Directed graphs}: The reductions from $k$-Clique to edge-disjoint and vertex-disjoint versions of $k$-\dsp on digraphs have a common step which is the construction of an intermediate graph \Gint described in~\autoref{sec:setting-up-the-redn-dir}. Then, the graphs in the reduction for the edge-disjoint version (\autoref{sec:fpt-inapprox-edsp-planar-dags}) and vertex-disjoint version (\autoref{sec:fpt-inapprox-vdsp-1-planar-dags}) are obtained by problem-specific splitting operations from the digraph~\Gint.
    \item \underline{Undirected graphs}: The reductions from $k$-Clique to edge-disjoint and vertex-disjoint versions of $k$-\dsp on undirected graphs have a common step being the construction of an intermediate graph \GintU (\autoref{sec:setting-up-the-U}). Then, the graphs in the reduction for the edge-disjoint version (\autoref{sec:fpt-inapprox-edsp-planar-undir}) and vertex-disjoint version (\autoref{sec:fpt-inapprox-vdsp-undir}) are obtained by problem-specific splitting operations from the undirected graph \GintU.
\end{itemize}

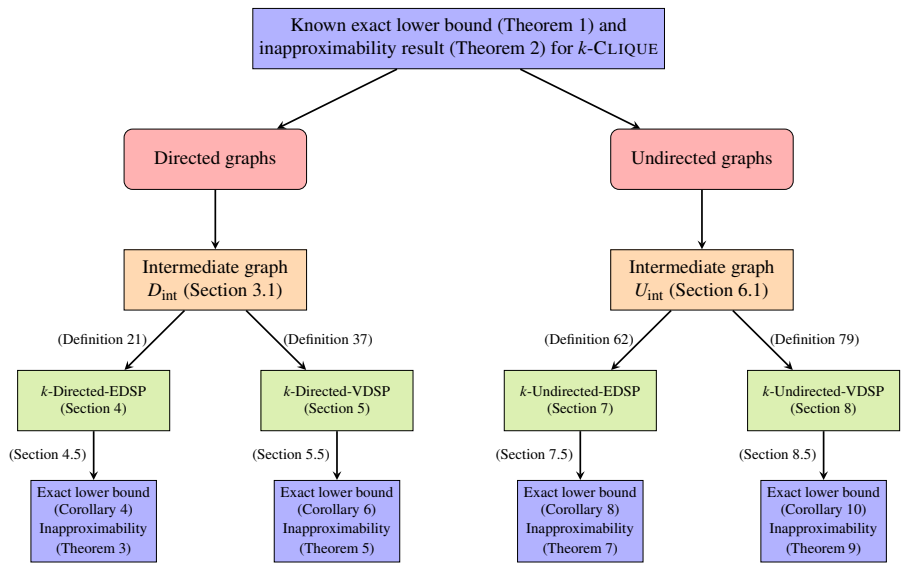
\begin{figure}
    \centering
    \scalebox{0.8}{
        \begin{tikzpicture}[ node distance = 2cm]
            \tikzstyle{startstop} = [rectangle, rounded corners, minimum width=3cm, minimum height=1cm,text centered, draw=black, fill=red!30]

            \tikzstyle{process} = [rectangle, minimum width=3cm, minimum height=1cm, text centered, draw=black, fill=orange!30]
            \tikzstyle{decision} = [rectangle, minimum width=2.5cm, minimum height=1cm, text centered, draw=black, fill=green!30]
            \tikzstyle{bluestar} = [rectangle, minimum width=3cm, minimum height=1cm, text centered, draw=black, fill=blue!30]

            \tikzstyle{bluesmall} = [rectangle, minimum width=2cm, minimum height=1cm, text centered, draw=black, fill=blue!30]

            \tikzstyle{arrow} = [thick,->,>=stealth]

            \node (dir) [startstop] {Directed graphs};
            \node (dirint) [process, below of=dir] {\shortstack{Intermediate graph \\ \Gint (\autoref{sec:construction-of-Gint})}  };

            \node (undir) [startstop, right of=dir, xshift=6cm] {Undirected graphs};
            \node (undirint) [process, below of=undir] {\shortstack{Intermediate graph \\ \GintU (\autoref{sec:construction-of-GintU})}  };

            \node (clique) [bluestar, above of=dir, xshift=4cm] {\shortstack{Known exact lower bound (\autoref{thm:chen-clique-exact-lb}) and \\ inapproximability result (\autoref{thm:cli_inapprox}) for $k$-\textsc{Clique}} };

            \draw [arrow] (dir) -- (dirint);
            \draw [arrow] (undir) -- (undirint);

            \draw [arrow] (clique) -- (dir);
            \draw [arrow] (clique) -- (undir);

            \node (diredge) [decision, below of=dirint, xshift=-2cm] {\shortstack{{\scriptsize $k$-Directed-\EDSP} \\ {\scriptsize (\autoref{sec:fpt-inapprox-edsp-planar-dags})}} };
            \node (dirvert) [decision, below of=dirint, xshift=2cm] {\shortstack{{\scriptsize $k$-Directed-\VDSP} \\ {\scriptsize (\autoref{sec:fpt-inapprox-vdsp-1-planar-dags})}} };

            \node (undiredge) [decision, below of=undirint, xshift=-2cm] {\shortstack{{\scriptsize $k$-Undirected-\EDSP} \\ {\scriptsize (\autoref{sec:fpt-inapprox-edsp-planar-undir})} } };
            \node (undirvert) [decision, below of=undirint, xshift=2cm] {\shortstack{{\scriptsize $k$-Undirected-\VDSP} \\ {\scriptsize (\autoref{sec:fpt-inapprox-vdsp-undir})}} };

            \draw [arrow] (dirint) -- node[left] {\scriptsize (\autoref{def:splitting-operation-edge})} (diredge);
            \draw [arrow] (dirint) -- node[right] {\scriptsize (\autoref{def:splitting-operation-vertex})} (dirvert);

            \draw [arrow] (undirint) -- node[left] {\scriptsize (\autoref{def:splitting-operation-edge-U})} (undiredge);
            \draw [arrow] (undirint) -- node[right] {\scriptsize (\autoref{def:splitting-operation-vertex-U})} (undirvert);

            \node (diredgethm) [bluesmall, below of=diredge] {\shortstack{{\scriptsize Exact lower bound} \\ {\scriptsize (\autoref{thm:hardness-edge-result})} \\ {\scriptsize Inapproximability} \\ {\scriptsize (\autoref{thm:inapprox-edge-result})}} };
            \node (dirvertthm) [bluesmall, below of=dirvert] {\shortstack{{\scriptsize Exact lower bound} \\ {\scriptsize (\autoref{thm:hardness-vertex-result})} \\ {\scriptsize Inapproximability} \\ {\scriptsize (\autoref{thm:inapprox-vertex-result})}} };

            \node (undiredgethm) [bluesmall, below of=undiredge] {\shortstack{{\scriptsize Exact lower bound} \\ {\scriptsize (\autoref{thm:hardness-edge-result-U})} \\ {\scriptsize Inapproximability} \\ {\scriptsize (\autoref{thm:inapprox-edge-result-U})}} };
            \node (undirvertthm) [bluesmall, below of=undirvert] {\shortstack{{\scriptsize Exact lower bound} \\ {\scriptsize (\autoref{thm:hardness-vertex-result-U})} \\ {\scriptsize Inapproximability} \\ {\scriptsize (\autoref{thm:inapprox-vertex-result-U})}} };

            \draw [arrow] (diredge) -- node[left] {\scriptsize (\autoref{sec:proof-of-main-theorem-edge})} (diredgethm);
            \draw [arrow] (dirvert) -- node[left] {\scriptsize (\autoref{sec:proof-of-main-theorem-vertex})} (dirvertthm);

            \draw [arrow] (undiredge) -- node[left] {\scriptsize (\autoref{sec:proof-of-main-theorem-edge-U})} (undiredgethm);
            \draw [arrow] (undirvert) -- node[left] {\scriptsize (\autoref{sec:proof-of-main-theorem-vertex-U})} (undirvertthm);

        \end{tikzpicture}

    }
    \caption{Our lower bounds originate from one of two known results for \kclique; namely~\autoref{thm:chen-clique-exact-lb} for our exactness results and~\autoref{thm:cli_inapprox} for inapproximability. This flowchart demonstrates the symmetry of the processes for obtaining each result by first defining an intermediate graph and making adjustments for the specific \DSP instance.
    \label{fig:flowchart}
    }
\end{figure}

\subsection{Prior Work on \texorpdfstring{$k$}{k}-\dsp on directed graphs}
\label{sec:prior-directed}

The two versions of $k$-\dsp on directed graphs are defined as follows:\\

\begin{center}
    \noindent\framebox{\begin{minipage}{0.95\linewidth}
                           \dedsp \textbf{(\DEDSP)}\\
                           \emph{\underline{Input}}: An integer $k$, a directed graph $G=(V,E)$ with non-negative edge-lengths, and a set $\mathcal{T}=\big\{ (s_i, t_i)\ : 1\leq i\leq k \big\}\subseteq V\times V$ of $k$ terminal pairs.\\
                           \emph{\underline{Question}}: Does there exist a collection of $k$ paths $P_1, P_2, \ldots, P_k$ in $G$ such that
                           \begin{itemize}
                               \item $P_i$ is a shortest $s_i \leadsto t_i$ path in $G$ for each $1\leq i\leq k$, and
                               \item for each $1\leq i\neq j\leq k$, the paths $P_i$ and $P_j$ are edge-disjoint?
                           \end{itemize}
    \end{minipage}}
\end{center}

\vspace{2mm}

\begin{center}
    \noindent\framebox{\begin{minipage}{0.95\linewidth}
                           \dvdsp \textbf{(\DVDSP)}\\
                           \emph{\underline{Input}}: An integer $k$, a directed graph $G=(V,E)$ with non-negative edge-lengths, and a set $\mathcal{T}=\big\{ (s_i, t_i)\ : 1\leq i\leq k \big\}\subseteq V\times V$ of $k$ terminal pairs.\\
                           \emph{\underline{Question}}: Does there exist a collection of $k$ paths $P_1, P_2, \ldots, P_k$ in $G$ such that
                           \begin{itemize}
                               \item $P_i$ is a shortest $s_i \leadsto t_i$ path in $G$ for each $1\leq i\leq k$, and
                               \item for each $1\leq i\neq j\leq k$, the paths $P_i$ and $P_j$ are vertex-disjoint?
                           \end{itemize}
    \end{minipage}}
\end{center}

\medskip

By setting all edge-lengths to be $0$, the hardness for $k$-\dsp on digraphs follows from that of the $k$-\disjp problem on digraphs. Eilam-Tzoreff~\cite{eilam} showed that both \DVDSP and \DEDSP are NP-hard when $k$ is part of the input, even when the input digraph is planar and all edge-lengths are $1$. Bérczi and Kobayashi~\cite{berczi-kobayashi} designed $n^{O(k)}$-time algorithms for \DVDSP on planar digraphs, and for \DVDSP and \DEDSP on DAGs by modifying an earlier algorithm of Fortune et al.\ for the $k$-Disjoint-Paths problem on DAGs~\cite{DBLP:journals/tcs/FortuneHW80}.
When each edge-length is positive, Bérczi and Kobayashi~\cite{berczi-kobayashi} also showed that Directed-$2$-VDSP and Directed-$2$-EDSP can be solved in $n^{O(1)}$ time.
Amiri and Wargalla~\cite{amiri-arxiv} showed a tight lower bound for \DEDSP on planar DAGs: under the Exponential Time Hypothesis (ETH)\footnote{The Exponential Time Hypothesis (ETH) states that $n$-variable $m$-clause 3-SAT cannot be solved in $2^{o(n)}\cdot (n+m)^{O(1)}$ time~\cite{eth,eth-2}.}, there is no computable function $f$, such that \DEDSP on planar DAGs admits an $f(k)\cdot n^{o(k)}$-time algorithm.
Our results are reached by advancing Chitnis's technique for obtaining an exact lower bound for \edp on planar DAGs~\cite{rajesh-ciac-21}.
Although not explicitly mentioned in their paper, the hardness reduction for \UVDSP on general graphs by Bentert et al. also holds for $1$-planar graphs and for DAGs if one were to orient all edges from either left-to-right or bottom-to-top, and can also be adapted to hold for a bounded max degree of 4~\cite[Proposition~3]{nichterlein}.

\subsection{Prior work on \texorpdfstring{$k$}{k}-\dsp on undirected graphs}
\label{sec:prior-undirected}

The two versions of the $k$-\dsp problem on undirected graphs, being \UEDSP and \UVDSP, are defined analogously to their directed counterparts. Eilam-Tzoreff~\cite{eilam} designed an $O(n^{8})$-time algorithm for Undirected-$2$-VDSP and~Undirected-$2$-EDSP in the case when all edge costs are guaranteed to be positive. Akhmedov~\cite{akhmedov} improved this to $O(n^7)$ when the costs are positive and further to $O(n^6)$ when all costs are $1$. Gottschau et al.~\cite{gottschauUndirectedTwoDisjoint2019} and Kobayashi and Sako~\cite{sako} independently gave $n^{O(1)}$-time algorithms for Undirected-$2$-VDSP and Undirected-$2$-EDSP when edge costs are non-negative.

The complexity of Undirected-$k$-VDSP and Undirected-$k$-EDSP for $k\geq 3$ was a long-standing open problem until Lochet~\cite{lochet} designed an XP algorithm running in $n^{O(k^{5^k})}$ time for general $k$ on \VDSP. Bentert et al. improved the running time of this algorithm to $n^{O(k!k)}$ using some geometric ideas~\cite[Lemma~26]{nichterlein}, and also showed that there is no $f(k)\cdot n^{o(k)}$-time algorithm (for any computable function $f$) under ETH~\cite[Proposition~3]{nichterlein}.
Bérczi and Kobayashi~\cite{berczi-kobayashi} showed that both the \UEDSP and \UVDSP problems on planar undirected graphs can be solved with a $n^{O(k)}$ time algorithm.

\subsection{Known Exact \& Inapproximate Lower Bounds for \kclique}
\label{sec:kclique-results}

There are four versions of the $k$-\dsp problem, depending on whether we require edge-disjointness or vertex-disjointness and if the input graph is directed or undirected. We obtain two lower bounds, one exact and one approximate, for each of these four versions. All eight of our lower bounds are obtained using reductions from $k$-Clique.

\medskip

\begin{center}
    \noindent\framebox{\begin{minipage}{0.95\linewidth}
                           \textbf{\kclique}\\
                           \emph{\underline{Input}}: Integer $k$, and an undirected graph $G=(V,E)$, where $V=\{v_1,v_2,\dots,v_N\}$.\\
                           \emph{\underline{Question}}: Does there exist a set $Z\subseteq V$ of size $k$ such that for all $x\neq y\in Z$ we have $x-y \in E$?
    \end{minipage}}
\end{center}

It is known that the \kclique problem is W[1]-hard~\cite{DBLP:journals/tcs/DowneyF95}. Chen et al.~\cite{chen-hardness} showed the following asymptotically tight lower bound for \kclique:

\begin{theorem}
    \label{thm:chen-clique-exact-lb}
    ~\cite{chen-hardness} Under the Exponential Time Hypothesis (ETH), the $k$-Clique problem on graphs with $N$ vertices cannot be solved in $f(k)\cdot N^{o(k)}$ time for any computable function~$f$.
\end{theorem}

We use~\autoref{thm:chen-clique-exact-lb} to show our lower bounds on the running times of exact algorithms for edge-disjoint version and vertex-disjoint version of $k$-\dsp on undirected and directed graphs.

To show our lower bounds on the running times of approximate algorithms, we need a stronger assumption known as the Gap-ETH
under which hardness of approximating \kclique is known. Formally, we use the following result:

\medskip
\begin{theorem}[Theorem~18,
    \cite{DBLP:journals/siamcomp/ChalermsookCKLM20}]
    Assuming Gap-ETH, there exist constants $\delta, r_0 > 0$ such that, for any
    computable function $g$ and for any positive integers $q \ge r \ge r_0$, there is no
    algorithm that, given a graph $G'$, can distinguish between the following cases
    in $g(q,r)\cdot N^{\delta r}$ time, where $N=|V(G')|$:
    \begin{description}
        \item[Case 1:] $\textsc{Clique}(G') \ge q$; and
        \item[Case 2:] $\textsc{Clique}(G') < r$;
    \end{description}
    where $\textsc{Clique}(G')$ denotes the maximum size of a clique in $G'$.
    \label{thm:cli_inapprox}
\end{theorem}
\medskip

\subsection{Notation}
\label{sec:app-notation}

All graphs considered in this paper are simple, i.e., do not have self-loops or multiple edges. We use (mostly) standard graph-theory notation~\cite{diestel-book}. The set $\{1,2,3,\ldots, M\}$ is denoted by $[M]$ for each $M\in \mathbb{N}$.
A directed edge (resp. path) from $s$ to $t$ is denoted by $s\to t$ (resp. $s - t$).
An undirected edge between $s$ and $t$ is denoted by $s-t$: we also use the same notation for an undirected path, but the context is made clear by saying $s-t$ path or edge $s-t$.

We use the \textbf{non-standard} notation (to avoid having to consider different cases in our proofs): $s \leadsto s$ or $s\to s$ \textbf{does not} represent a self-loop but rather is to be viewed as \emph{``just staying put"} at the vertex $s$. A similar notation is used for the undirected case: $s-s$ \textbf{does not} represent a self-loop but rather is to be viewed as \emph{``just staying put"} at the vertex $s$.

If $A,B\subseteq V(G)$ then we say that there is an $A - B$ path (resp. $A\to B$ path in digraphs) if and only if there exists two vertices $a\in A, b\in B$ such that there is an $a - b$ path (resp. $a\leadsto b$ path in digraphs). For $A\subseteq V(G)$ we define $N_{G}^{+}(A) = \big\{ x\notin A\ : \exists\ y\in A\ \text{such that } (y,x)\in E(G) \big\}$ and $N_{G}^{-}(A) = \big\{ x\notin A\ : \exists\ y\in A\ \text{such that } (x,y)\in E(G) \big\}$. For $A\subseteq V(G)$ we define $G[A]$ to be the graph induced on the vertex set $A$, i.e., $G[A]:= (A,E_A)$ where $E_A:=E(G)\cap (A\times A)$.

Given a directed graph $G=(V,E)$ and a set $\mathcal{T}\subseteq V\times V$ of $k$ terminal pairs given by $\big\{(s_i, t_i) : 1\leq i\leq k\big\}$ which form an instance of $k$-\EDSP (resp. $k$-\VDSP), we say that a subset $\mathcal{T}^{\prime} \subseteq \mathcal{T}$ of the terminal pairs can be \emph{satisfied} if and only if there exists a set $\mathcal{P} = \{P_1,P_2,\ldots,P_{|\mathcal{T}^{\prime}|}\}$ of paths such that
\begin{itemize}
    \item $\mathcal{P}$ contains a shortest $s \leadsto t$ path for each $(s,t)\in \mathcal{T}'$
    \item Every pair of paths from $\mathcal{P}$ is pairwise edge-disjoint (resp. vertex-disjoint)
\end{itemize}

\noindent Finally, note that we use a constant edge-length of $1$ in all graphs and reductions throughout this paper. This allows us to measure lengths of paths equivalently either by counting the number of edges or the number of vertices. We choose the latter option as it helps simplify some of the arguments.

\section{Our Results}
\label{sec:our-result}

In this paper, we obtain lower bounds on the running time of exact and approximate\footnote{An $\alpha$-approximation for $k$-\dsp distinguishes between these two cases: either (i) all $k$ pairs can be satisfied; or (ii) the maximum number of pairs that can be satisfied is less than $\alpha\cdot k$.}~algorithms for the edge-disjoint and vertex-disjoint versions of $k$-\dsp on undirected and directed graphs. For the notion of shortest paths, there are two possible choices with allowing either vertex costs or edge costs. This does not matter in our lower bounds since we use a uniform cost of 1 for each vertex (hence paths lengths could also be counted equivalently in number of unit-cost edges). Note that, by considering each vertex to have non-zero cost, we cannot exploit the known hardness results for $k$-\disjp (a special case of $k$-\dsp with all vertex costs~$0$).

Our exact and approximate lower bounds are based on assuming the Exponential Time Hypothesis (ETH) and Gap Exponential Time Hypothesis (Gap-ETH) respectively:
\begin{itemize}
    \item \underline{ETH}: The Exponential Time Hypothesis (ETH) states that $n$-variable $m$-clause 3-SAT cannot be solved in $2^{o(n)}\cdot (n+m)^{O(1)}$ time~\cite{eth,eth-2}.
    \item \underline{Gap-ETH}: The gap version of the ETH~\cite{pasin-gap-eth,irit-gap-eth} states that there exists a constant $\delta>0$ such that there is no $2^{o(n)}$ time algorithm which given instances of $3$-SAT on $n$ variables can distinguish between the case when all clauses are satisfiable versus the case when every assignment to the variables leaves at least $\delta$-fraction of the clauses unsatisfied. We refer the interested reader to~\cite{irit-gap-eth,DBLP:journals/siamcomp/ChalermsookCKLM20} for discussions about the plausibility of Gap-ETH.
\end{itemize}

\subsection{Exact and approximate lower bounds for directed graphs}\label{sec:exact-and-approximate-lower-bounds-for-directed-graphs}

\noindent We now state our exact and approximate lower bounds for the edge-disjoint and vertex-disjoint versions of $k$-\dsp on directed graphs, which all hold even if both the max in-degree and max out-degree of the input digraph are at most $2$. The exact and approximate lower bounds for the edge-disjoint version are:

\begin{restatable}{theorem}{apxdirectededgethm}
    \textbf{(inapproximability)} Assuming Gap-ETH,  for each $0<\varepsilon\le \frac{1}{2}$ there exists a constant~$\zeta > 0$ such that no $f(k)\cdot n^{\zeta k}$ time algorithm can distinguish between the following two cases of \DEDSP
    \begin{itemize}
        \item All $k$ pairs can be satisfied
        \item At most $(\frac{1}{2} + \epsilon) \cdot k$ pairs can be satisfied
    \end{itemize}
    Here $f$ is any computable function, $n$ is the number of vertices and $k$ is the number of terminal pairs. Our lower bound also holds if the input graph is a planar DAG and has both max in-degree and max out-degree at most 2.
    \label{thm:inapprox-edge-result}
\end{restatable}

\begin{restatable}{corollary}{exactdirectededgethm}
    \textbf{(exact lower bound)} The \DEDSP problem on planar DAGs is W[1]-hard parameterized by the number of terminal pairs $k$, even if the max in-degree and max out-degree is at most $2$. Moreover, under the ETH, there is no computable function $f$ which solves this problem in $f(k)\cdot n^{o(k)}$ time.

    \label{thm:hardness-edge-result}
\end{restatable}

\noindent
The exact and approximate lower bounds for the vertex-disjoint version are:

\medskip

\begin{restatable}{theorem}{apxdirectedvertexthm}
    \textbf{(inapproximability)} Assuming Gap-ETH,  for each $0<\varepsilon\le \frac{1}{2}$ there exists a constant~$\zeta > 0$ such that no $f(k)\cdot n^{\zeta k}$ time algorithm can distinguish between the following two cases of \DVDSP
    \begin{itemize}
        \item All $k$ pairs can be satisfied
        \item At most $(\frac{1}{2} + \epsilon) \cdot k$ pairs can be satisfied
    \end{itemize}
    Here $f$ is any computable function, $n$ is the number of vertices and $k$ is the number of terminal pairs. Our lower bound also holds if the input graph is a $1$-planar DAG and has both max in-degree and max out-degree at most 2.

    \label{thm:inapprox-vertex-result}
\end{restatable}

\begin{restatable}{corollary}{exactdirectedvertexthm}
    \textbf{(exact lower bound)} The \DVDSP problem on $1$-planar DAGs is W[1]-hard parameterized by the number of terminal pairs $k$, even if the max in-degree and max out-degree is at most $2$. Moreover, under the ETH, there is no computable function $f$ which solves this problem in $f(k)\cdot n^{o(k)}$ time.

    \label{thm:hardness-vertex-result}
\end{restatable}

\noindent We note that the W[1]-hardness of \DEDSP on DAGs was also obtained by Amiri et al.~\cite{amiri-arxiv}, and our~\autoref{thm:hardness-edge-result} strengthens this by showing that the hardness holds even if max in-degree and max out-degree is 2. Likewise, by orienting all of the edges in Bentert et al.'s reduction from~\cite{nichterlein} from either left-to-right or bottom-to-top, one appears to obtain a W[1]-hardness of the \DVDSP problem on DAGs. Our~\autoref{thm:hardness-vertex-result} strengthens this by showing that the hardness holds even if the graph is $1$-planar and all vertices have a maximum in and out degree of 2.

\subsection{Exact and approximate lower bounds for undirected graphs}\label{sec:exact-and-approximate-lower-bounds-for-undirected-graphs}

\noindent We now state our exact and approximate lower bounds for the edge-disjoint and vertex-disjoint versions of $k$-\dsp on undirected graphs, which all hold even if the max degree of the input graph is at most $4$.. The exact and approximate lower bounds for the edge-disjoint version are:

\begin{restatable}{theorem}{apxundirectededgethm}
    \textbf{(inapproximability)} Assuming Gap-ETH,  for each $0<\varepsilon\le \frac{1}{2}$ there exists a constant~$\zeta > 0$ such that no $f(k)\cdot n^{\zeta k}$ time algorithm can distinguish between the following two cases of \UEDSP
    \begin{itemize}
        \item All $k$ pairs can be satisfied
        \item At most $(\frac{1}{2} + \epsilon) \cdot k$ pairs can be satisfied
    \end{itemize}
    Here $f$ is any computable function, $n$ is the number of vertices and $k$ is the number of terminal pairs. Our lower bound also holds if the input graph is planar and has max degree at most $4$.

    \label{thm:inapprox-edge-result-U}
\end{restatable}

\begin{restatable}{corollary}{exactundirectededgethm}
    \textbf{(exact lower bound)} The \UEDSP problem on planar graphs is W[1]-hard parameterized by the number of terminal pairs $k$, even if the max degree is at most $4$. Moreover, under the ETH, there is no computable function $f$ which solves this problem in $f(k)\cdot n^{o(k)}$ time.

    \label{thm:hardness-edge-result-U}
\end{restatable}

\medskip

\noindent
The exact and approximate lower bounds for the vertex-disjoint version are:

\begin{restatable}{theorem}{apxundirectedvertexthm}
    \textbf{(inapproximability)} Assuming Gap-ETH,  for each $0<\varepsilon\le \frac{1}{2}$ there exists a constant~$\zeta > 0$ such that no $f(k)\cdot n^{\zeta k}$ time algorithm can distinguish between the following two cases of \UVDSP
    \begin{itemize}
        \item All $k$ pairs can be satisfied
        \item At most $(\frac{1}{2} + \epsilon) \cdot k$ pairs can be satisfied
    \end{itemize}
    Here $f$ is any computable function, $n$ is the number of vertices and $k$ is the number of terminal pairs. Our lower bound also holds if the input graph is $1$-planar and has max degree at most $4$.

    \label{thm:inapprox-vertex-result-U}
\end{restatable}

\begin{restatable}{corollary}{exactundirectedvertexthm}
    \textbf{(exact lower bound)} The \UVDSP problem on $1$-planar graphs is W[1]-hard parameterized by the number of terminal pairs $k$, even if the max degree is at most $4$. Moreover, under the ETH, there is no computable function $f$ which solves this problem in $f(k)\cdot n^{o(k)}$ time.
    \label{thm:hardness-vertex-result-U}
\end{restatable}

\medskip

\noindent
We again note that the W[1]-hardness of \UVDSP was obtained by Bentert et al.~\cite{nichterlein}, and our~\autoref{thm:hardness-vertex-result-U} strengthens this by showing that the hardness holds even if the graph is a $1$-planar DAG and all vertices have a maximum degree of 4.

\begin{table}[ht]
    \label{tab:theorems}
    \centering
    \begin{tabular}{|| c | c | c | c | c ||}
        \hline
        Problem & \thead{Inapproximability \\ Factor} & Hypothesis & Lower Bound & Reference \\ [0.5ex]
        \hline\hline
        \DEDSP & \makecell{$\left(\frac{1}{2} + \epsilon\right)$                                                                       \\ for each $0 < \epsilon \leq \frac{1}{2}$} & Gap-ETH & \makecell{$f(k)\cdot n^{\zeta k}$ \\ for some $\zeta > 0$} & \autoref{thm:inapprox-edge-result} \\
        \hline
        \DEDSP & Exact & ETH & $f(k)\cdot n^{o(k)}$ & \autoref{thm:hardness-edge-result} \\
        \hline
        \DVDSP & \makecell{$\left(\frac{1}{2} + \epsilon\right)$                                                                       \\ for each $0 < \epsilon \leq \frac{1}{2}$} & Gap-ETH & \makecell{$f(k)\cdot n^{\zeta k}$ \\ for some $\zeta > 0$} & \autoref{thm:inapprox-vertex-result} \\
        \hline
        \DVDSP & Exact & ETH & $f(k)\cdot n^{o(k)}$ & \autoref{thm:hardness-vertex-result} \\
        \hline
        \hline
        \UEDSP & \makecell{$\left(\frac{1}{2} + \epsilon\right)$                                                                       \\ for each $0 < \epsilon \leq \frac{1}{2}$} & Gap-ETH & \makecell{$f(k)\cdot n^{\zeta k}$ \\ for some $\zeta > 0$} & \autoref{thm:inapprox-edge-result-U} \\
        \hline
        \UEDSP & Exact & ETH & $f(k)\cdot n^{o(k)}$ & \autoref{thm:hardness-edge-result-U} \\
        \hline
        \UVDSP & \makecell{$\left(\frac{1}{2} + \epsilon\right)$                                                                       \\ for each $0 < \epsilon \leq \frac{1}{2}$} & Gap-ETH & \makecell{$f(k)\cdot n^{\zeta k}$ \\ for some $\zeta > 0$} & \autoref{thm:inapprox-vertex-result-U} \\
        \hline
        \UVDSP & Exact & ETH & $f(k)\cdot n^{o(k)}$ & \autoref{thm:hardness-vertex-result-U} \\
        \hline
    \end{tabular}

    \caption[]{A compendium of results in this paper. Throughout the table, $f$ represents any computable function. Note that all our \EDSP and \VDSP results hold even for planar and $1$-planar\footnotemark~graphs respectively. Furthermore, the directed results hold if the input graph is a DAG and has both max in-degree and max-degree upper bounded by~$2$. The undirected results hold even if the max degree of the input graph is upper bounded by~$4$.}
\end{table}

\footnotetext{A graph is $1$-planar if it can be drawn in the plane with each edge crossed by at most one other edge.}

\paragraph*{Placing our lower bounds in the context of prior work:}

Our inapproximability results are \emph{tight} for our specific reductions because in all of our reductions of \textsc{Disjoint-Shortest-Paths} it is trivially possible to satisfy half the pairs\footnote{For example, in \autoref{fig:main}, selecting a shortest $c_i \leadsto d_i$ for every $1 \leq i \leq k$ or a shortest $a_j \leadsto b_j$ for every $1 \leq j \leq k$ provides $k$ pairs that are necessarily pairwise edge-disjoint.}.
To obtain stronger inapproximability results, one therefore needs ideas quite different from those introduced in this paper such as those given by Bentert et al.~\cite{bentertTightApproximationKernelization2024}.
Their paper provides an $o(k)$-factor inapproximability lower bound in $f(k)\cdot \text{poly}(n)$ time under Gap-ETH for \vdsp and \edsp graphs for which the terminal vertices have a degree of at most 2 and every other vertex has degree at most 3.
Our inapproximability results for \edsp (\autoref{thm:inapprox-edge-result} \autoref{thm:inapprox-edge-result-U}) and \vdsp (\autoref{thm:inapprox-vertex-result} \autoref{thm:inapprox-vertex-result-U}), however, hold even if the input graph is planar or $1$-planar respectively.

The framework presented in this paper provides one reduction technique for achieving the aforementioned inapproximability bounds, in addition to a number of exact hardness results (\autoref{thm:hardness-edge-result} \autoref{thm:hardness-vertex-result}, \autoref{thm:hardness-edge-result-U} and \autoref{thm:hardness-vertex-result-U}).
Chitnis focused on planar graphs in their hardness proof for $k$-\edp, and we take this further in analysing how we can obtain analogous lower bounds on other specific graph classes.
Chitnis,~\cite{rajesh-ciac-21}, used a reduction from \gtleq and we note that \autoref{thm:hardness-edge-result} could also be obtained by reducing from \gt, although we present a reduction from \kclique here. \autoref{thm:hardness-edge-result} was obtained independently by Amiri \& Wargalla~\cite{amiri-arxiv}, although without the added restriction of bounded in-degree and out-degree.
Although we obtain our results by developing Chitnis's technique for obtaining lower bounds for \edp on DAGs~\cite{rajesh-ciac-21}, we note that one can obtain~\autoref{thm:hardness-vertex-result-U} by analysing Bentert et al.'s graph for hardness on general undirected graphs to be $1$-planar and subsequently~\autoref{thm:hardness-vertex-result} by orienting all edges in their graph either left-to-right or bottom-to-top~\cite[Proposition~3]{nichterlein}.
To the best of our knowledge, \autoref{thm:hardness-edge-result-U} is the first result showing a tight lower bound for Bérczi and Kobayashi's~\cite{berczi-kobayashi} $n^{O(k)}$ time algorithm for \UEDSP on planar graphs.

	\section{Setting up the reductions for \texorpdfstring{$k$}{k}-\dsp on directed graphs}
\label{sec:setting-up-the-redn-dir}

This section describes the common part of the reductions from \kclique to \DEDSP and \DVDSP, which corresponds to the top of the left-hand branch in~\autoref{fig:flowchart}. First, in~\autoref{sec:construction-of-Gint} we construct the intermediate directed graph \Gint which is later used to obtain the graphs \Gedge (\autoref{sec:fpt-inapprox-edsp-planar-dags}) and \Gvert (\autoref{sec:fpt-inapprox-vdsp-1-planar-dags}) used to obtain lower bounds for \DEDSP and \DVDSP respectively.
In~\autoref{sec:characterizing-shortest-in-G}, we then characterize shortest paths (between terminal pairs) in this intermediate graph \Gint.

We note that the intermediate graph \Gint graph is essentially the same as the graph that was constructed for the W[1]-hardness reduction of $k$-Directed-\EDP from \gtleq by Chitnis~\cite{rajesh-ciac-21}.

\subsection{Construction of the intermediate graph \Gint}
\label{sec:construction-of-Gint}
Given an instance $G=(V,E)$ of \kclique with $V=\{v_1, v_2, \ldots, v_N\}$, we now build an instance of an intermediate digraph \Gint (\autoref{fig:main}). This graph, \Gint, is later modified to obtain the final graphs \Gedge (\autoref{sec:construction-of-Gedge}) and \Gvert (\autoref{sec:construction-of-Gvertex}) which are used to obtain exact and approximate lower bounds
for the \dedsp and \dvdsp problems, respectively.

\begin{figure}[!p]
\centering
\begin{tikzpicture}[scale=0.55]
\foreach \i in {0,1,2}
    \foreach \j in {0,1,2}
{
\begin{scope}[shift={(6*\i,6*\j)}]

        \foreach \x in {1,2,...,5}
        \foreach \y in {1,2,...,5}
    {
        \draw [black] plot [only marks, mark size=3, mark=*] coordinates {(\x,\y)};
    }

        \foreach \x in {1,2,...,5}
    \foreach \y in {1,2,3,4}
    {
        \path (\x,\y) node(a) {} (\x,\y+1) node(b) {};
        \draw[thick,->] (a) -- (b);
    }

        \foreach \y in {1,2,...,5}
        \foreach \x in {1,2,3,4}
    {
        \path (\x,\y) node(a) {} (\x+1,\y) node(b) {};
        \draw[thick,->] (a) -- (b);
    }

\end{scope}
}
\foreach \i in {0,1,2}
\foreach \j in {1,2}
{
\begin{scope}[shift={(6*\i,6*\j-6)}]

\foreach \x in {1,2,...,5}
{

    \path (\x,5) node(a) {} (\x,7) node(b) {};
        \draw[red,very thick,->] (a) -- (b);

}

\end{scope}
}

\foreach \j in {0,1,2}
\foreach \i in {0,1}
{
\begin{scope}[shift={(6*\i,6*\j)}]

\foreach \y in {1,2,...,5}
{
        \path (5,\y) node(a) {} (7,\y) node(b) {};
        \draw[red,very thick,->] (a) -- (b);
}

\end{scope}
}

\draw [green] plot [only marks, mark size=3, mark=*] coordinates {(-1,3)}
node[label={[xshift=-3mm,yshift=-4mm] $c_{1}$}] {} ;

\draw [green] plot [only marks, mark size=3, mark=*] coordinates {(-1,9)}
node[label={[xshift=-3mm,yshift=-4mm] $c_{2}$}] {} ;

\draw [green] plot [only marks, mark size=3, mark=*] coordinates {(-1,15)}
node[label={[xshift=-3mm,yshift=-4mm] $c_{3}$}] {} ;

\draw [green] plot [only marks, mark size=3, mark=*] coordinates {(19,3)}
node[label={[xshift=3mm,yshift=-4mm] $d_{1}$}] {} ;

\draw [green] plot [only marks, mark size=3, mark=*] coordinates {(19,9)}
node[label={[xshift=3mm,yshift=-4mm] $d_{2}$}] {} ;

\draw [green] plot [only marks, mark size=3, mark=*] coordinates {(19,15)}
node[label={[xshift=3mm,yshift=-4mm] $d_{3}$}] {} ;

\foreach \k in {0,1,2}
{
\begin{scope}[shift={(0,6*\k)}]
\foreach \y in {1,2,...,5}
    {
        \path (-1,3) node(a) {} (1,\y) node(b) {};
        \draw[magenta,very thick,middlearrow={>}] (a) -- (b);

        \path (17,\y) node(a) {} (19,3) node(b) {};
        \draw[magenta,very thick,middlearrow={>}] (a) -- (b);
    }
\end{scope}
}

\draw [green] plot [only marks, mark size=3, mark=*] coordinates {(3,-1)}
node[label={[xshift=0mm,yshift=-7mm] $a_{1}$}] {} ;

\draw [green] plot [only marks, mark size=3, mark=*] coordinates {(9,-1)}
node[label={[xshift=0mm,yshift=-7mm] $a_{2}$}] {} ;

\draw [green] plot [only marks, mark size=3, mark=*] coordinates {(15,-1)}
node[label={[xshift=0mm,yshift=-7mm] $a_{3}$}] {} ;

\draw [green] plot [only marks, mark size=3, mark=*] coordinates {(3,19)}
node[label={[xshift=0mm,yshift=0mm] $b_{1}$}] {} ;

\draw [green] plot [only marks, mark size=3, mark=*] coordinates {(9,19)}
node[label={[xshift=0mm,yshift=0mm] $b_{2}$}] {} ;

\draw [green] plot [only marks, mark size=3, mark=*] coordinates {(15,19)}
node[label={[xshift=0mm,yshift=0mm] $b_{3}$}] {} ;

\foreach \k in {0,1,2}
{
\begin{scope}[shift={(6*\k,0)}]
\foreach \x in {1,2,...,5}
    {
        \path (3,-1) node(a) {} (\x,1) node(b) {};
        \draw[magenta,very thick,middlearrow={>}] (a) -- (b);

        \path (3,19) node(a) {} (\x,17) node(b) {};
        \draw[magenta,very thick,middlearrow={>}] (b) -- (a);
    }
\end{scope}
}

\draw [rotate=45,black] plot [only marks, mark size=0, mark=*] coordinates
{(0,0)}
node[label={[rotate=45,xshift=0mm,yshift=-2mm] Origin}] {} ;

\end{tikzpicture}
\caption{The intermediate directed graph \Gint constructed from an instance $(G,k)$ of \kclique (with $k=3$ and $N=5$) via the construction described in~\autoref{sec:construction-of-Gint}.
\label{fig:main}
}
\end{figure}
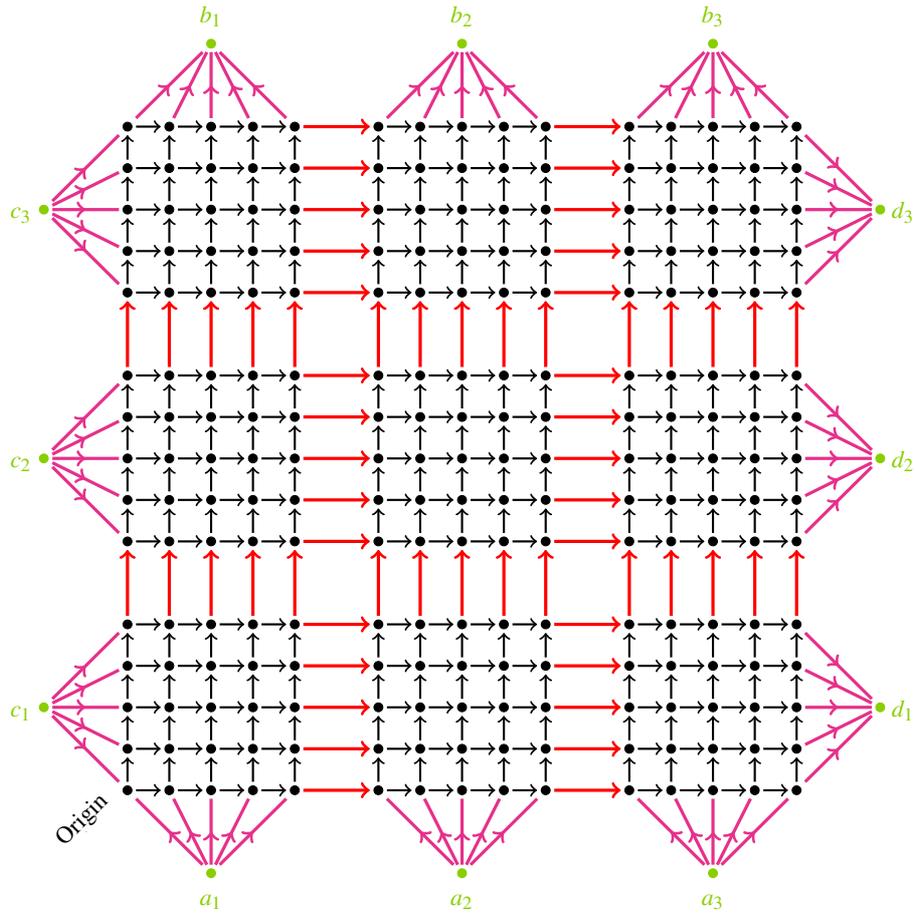

Before constructing the graph \Gint, we first define the following sets for a given instance $G$ of \kclique:
\begin{equation}
    \label{eqn:clique-to-gt-reduction-D}
    \begin{aligned}
        \text{For each}\ i\in [k],\ \text{let}\ S_{i,i}: = \{(a,a)\ :\ 1\leq a\leq N\} \\
        \text{For each}\ 1\leq i\neq j\leq k,\ \text{let}\ S_{i,j}:= \{ (a,b)\ :\ v_{a}-v_{b}\in E(G) \}
    \end{aligned}
\end{equation}

\noindent
We construct the digraph \Gint via the following steps (refer to~\autoref{fig:main}):
\begin{enumerate}
    \item \textbf{Origin}: The origin is marked at the bottom left corner of \Gint (\autoref{fig:main}). This is defined just so we can view the naming of the vertices as per the usual $X-Y$ coordinate system: increasing horizontally towards the right, and vertically towards the top.

    \item \textbf{Grid (black) vertices and edges}: For each $1\leq i,j\leq k$, introduce a (directed) $N\times N$ grid $D_{i,j}$ where the column numbers increase from $1$ to $N$ as we go from left to right, and the row numbers increase from $1$ to $N$ as we go from bottom to top. For each $1\leq q,\ell\leq N$ the unique vertex which is the intersection of the $q^{\text{th}}$ column and $\ell^{\text{th}}$ row of $D_{i,j}$ is denoted by $\w_{i,j}^{q,\ell}$. The vertex set and edge set of $D_{i,j}$ is defined formally as:
    \begin{itemize}
        \item $V(D_{i,j})= \big\{ \w_{i,j}^{q,\ell} : 1\leq q,\ell\leq N \big\}$
        \item $E(D_{i,j}) = \left(\bigcup_{(q,\ell)\in [N]\times [N-1]} \w_{i,j}^{q,\ell} \to \w_{i,j}^{q,\ell+1} \right) \cup \left( \bigcup_{(q,\ell)\in [N-1]\times [N]} \w_{i,j}^{q,\ell} \to \w_{i,j}^{q+1,\ell} \right)$
    \end{itemize}

    All vertices and edges of $D_{i,j}$ are shown in~\autoref{fig:main} using black colour. Note that each horizontal edge of the grid $D_{i,j}$ is oriented to the right, and each vertical edge is oriented towards the top. We later (\autoref{def:splitting-operation-edge} and \autoref{def:splitting-operation-vertex}) modify the grid $D_{i,j}$ (in a problem-specific way) to \emph{represent} the set $S_{i,j}$ defined in~\autoref{eqn:clique-to-gt-reduction-D}.

    For each $1\leq i,j\leq k$ we define the set of \emph{boundary} vertices of the grid $D_{i,j}$ as follows:
    \begin{equation}
        \label{eqn:left-right-top-bottom-G1}
        \begin{aligned}
            \Le(D_{i,j}) := \big\{ \w_{i,j}^{1,\ell}\ :\ \ell\in [N]  \big\}\ ;\
            \Ri(D_{i,j}) := \big\{ \w_{i,j}^{N,\ell}\ :\ \ell\in [N]  \big\}\ ; \\
            \To(D_{i,j}) := \big\{ \w_{i,j}^{\ell,N}\ :\ \ell\in [N]  \big\}\ ;\
            \Bo(D_{i,j}) := \big\{ \w_{i,j}^{\ell,1}\ :\ \ell\in [N]  \big\}\ .
        \end{aligned}
    \end{equation}

    \item \textbf{Arranging the $k^2$ different $N\times N$ grids $\{D_{i,j}\}_{1\leq i,j\leq k}$ into a large $k\times k$ grid}: Place the $k^2$ grids $\Big\{ D_{i,j} :\ (i,j)\in [k]\times [k] \Big\}$ into a big $k\times k$ grid of grids left to right according to growing $i$ and from bottom to top according to growing $j$. In particular, the grid $D_{1,1}$ is at bottom left corner of the construction, the grid $D_{k,k}$ at the top right corner, and so on.

    \item \textbf{\red{Red} edges for horizontal connections}: For each $(i,j)\in [k-1]\times [k]$, add a set of $N$ edges which form a directed perfect matching from $\Ri(D_{i,j})$ to $\Le(D_{i+1,j})$ given by $\Matching\left( D_{i,j}, D_{i+1,j}  \right):= \big\{ \w_{i,j}^{N,\ell} \to \w_{i+1,j}^{1,\ell}\ :\ \ell\in [N] \big\}$.

    \item \textbf{\red{Red} edges for vertical connections}: For each $(i,j)\in [k]\times [k-1]$, add a set of $N$ edges which form a directed perfect matching from $\To(D_{i,j})$ to $\Bo(D_{i,j+1})$ given by $\Matching\left( D_{i,j}, D_{i,j+1}  \right):= \big\{ \w_{i,j}^{\ell,N}  \to \w_{i,j+1}^{\ell,1}\ :\ \ell\in [N] \big\}$.

    \item \textbf{\green{Green} (terminal) vertices and magenta edges}: For each $i\in [k]$, add the following four sets of (terminal) vertices (shown in~\autoref{fig:main} using \green{green} colour)
    \begin{equation}
        \label{eqn:A-B-C-D-D}
        \begin{aligned}
            A := \big\{ a_i\ :\ i\in [k] \big\}\quad ;\quad B := \big\{ b_i\ :\ i\in [k] \big\}\quad ; \\
            C := \big\{ c_i\ :\ i\in [k] \big\}\quad ;\quad D := \big\{ d_i\ :\ i\in [k] \big\}\quad .
        \end{aligned}
    \end{equation}
    \noindent
    For each $i\in [k]$ we add the edges (shown in~\autoref{fig:main} using \magenta{magenta} colour)
    \begin{equation}
        \label{eqn:source-sink-A-B}
        \begin{aligned}
            \Source(A) := \big\{ a_i \to \w_{i,1}^{\ell,1}\ :\ \ell\in [N] \big\}\ ;\
            \Sink(B) := \big\{ \w_{i,N}^{\ell,N} \to b_{i}\ :\ \ell\in [N] \big\}
        \end{aligned}
    \end{equation}
    \noindent
    For each $j\in [k]$ we add the edges (shown in~\autoref{fig:main} using \magenta{magenta} colour)
    \begin{equation}
        \label{eqn:source-sink-C-D}
        \begin{aligned}
            \Source(C) := \big\{ c_j \to \w_{1,j}^{1,\ell}\ :\ \ell\in [N] \big\}\ ;\
            \Sink(D) := \big\{ \w_{N,j}^{N,\ell} \to d_{j}\ :\ \ell\in [N] \big\}
        \end{aligned}
    \end{equation}

\end{enumerate}
\medskip

\begin{definition}
    \label{def:lefty-right-topy-bottomy-dir}
    \textbf{(four neighbors of each grid vertex in \Gint)} Consider the drawing of \GintU from \autoref{fig:main}. This gives the natural notion of four neighbors for every black grid vertex: one to the left, right, bottom and top of each. For each (black) grid vertex $\z\in$ \Gint we define these as follows
    \begin{itemize}
        \item $\lefty(\z)$ is the vertex to the left of $\z$ (as seen by the reader) which has an edge incoming into $\z$
        \item $\bottomy(\z)$ is the vertex below $\z$ (as seen by the reader) which has an edge incoming into $\z$
        \item $\righty(\z)$ is the vertex to the right of $\z$ (as seen by the reader) which has an edge outgoing from $\z$
        \item $\topy(\z)$ is the vertex above $\z$ (as seen by the reader) which has an edge outgoing from $\z$
    \end{itemize}
    Note that in the case that $\z$ lies on the edge of the grid in \autoref{fig:main}, up to $2$ of its neighbours are in fact \green{green} terminal vertices.
\end{definition}
\medskip

This completes the construction of the graph \Gint (\autoref{fig:main}). The next two claims analyze the structure and size of this graph:

\medskip

\begin{claim}
    \label{clm:G1-is-planar-and-dag}
    \normalfont
    \Gint is a planar DAG.
\end{claim}
\begin{proof}
    ~\autoref{fig:main} gives a planar embedding of \Gint. It is easy to verify from the construction of \Gint described at the start of~\autoref{sec:construction-of-Gint} (see also~\autoref{fig:main}) that \Gint is a DAG.
\end{proof}
\medskip

\begin{claim}
    \normalfont
    The number of vertices in \Gint is $O(N^{2}k^{2})$
    \label{clm:size-of-G-int}
\end{claim}
\begin{proof}
    \Gint has $k^2$ different $N\times N$ grids viz.\ $\{D_{i,j}\}_{1\leq i,j\leq k}$. Hence, \Gint has $N^{2}k^2$ black vertices. Adding the $4k$ \green{green} vertices from $A\cup B\cup C\cup D$, it follows that number of vertices in \Gint is $N^{2}k^2+ 4k = O(N^2 k^2)$.
\end{proof}
\medskip

\subsection{Characterizing shortest paths in \Gint}
\label{sec:characterizing-shortest-in-G}

The goal of this section is to characterize the structure of shortest paths between terminal pairs in \Gint. In order to do this, we need to define the set of terminal pairs $\mathcal{T}$ and also assign vertex costs in \Gint.
\begin{equation}
    \label{eqn:definition-of-mathcal-T}
    \text{The set of terminal pairs is}\ \mathcal{T}:= \big\{(a_i, b_i) : i\in [k] \big\}\cup \big\{(c_j, d_j) : j\in [k] \big\}.
\end{equation}

\medskip

\begin{definition}
    \label{def:weights-in-directed} (\textbf{costs of vertices in \Gint}) Each black vertex in \Gint has a cost of 1.
\end{definition}
\medskip

\autoref{def:weights-in-directed} gives a cost to each vertex of \Gint which then naturally leads to the notion of cost of a path as the sum of costs of the vertices on it. With all costs being $1$, we can equivalently quantify paths either by measuring the number of edges or the number of vertices on them. Thus our choice to measure the cost in terms of the number of vertices has no bearing on the results that we obtain. We now define the \emph{canonical paths} within the graph.

\medskip

\begin{definition}
    \normalfont
    \textbf{(row-paths and column-paths in \Gint)} For each $(i,j)\in [k]\times [k]$ and $\ell\in [N]$ we define
    \begin{itemize}
        \item $\Row_{\ell}(D_{i,j})$ to be the $\w_{i,j}^{1,\ell}\leadsto \w_{i,j}^{N,\ell}$ path in $D_{\inter}[D_{i,j}]$ consisting of the following edges (in order): for each $r\in [N-1]$ take the black edge $\w_{i,j}^{r,\ell}\to \w_{i,j}^{r+1,\ell}$.
        \item $\Column_{\ell}(D_{i,j})$ to be the $\w_{i,j}^{\ell,1}\leadsto \w_{i,j}^{\ell,N}$ path in $D_{\inter}[D_{i,j}]$  consisting of the following edges (in order): for each $r\in [N-1]$ take the black edge $\w_{i,j}^{\ell,r}\to \w_{i,j}^{\ell,r+1}$.
    \end{itemize}

    \label{def:Column-Row-Gint}
\end{definition}
\medskip

It is easy to observe that each row-path and each column-path in \Gint contains exactly $N$ (black) vertices. We are now ready to define horizontal canonical paths and vertical canonical paths in \Gint:
\medskip

\begin{definition}
    \textbf{(horizontal canonical paths in \Gint)}
    Fix any $j\in [k]$. For each $r\in [N]$, we define $\CanInter(r\ ;\ c_j \leadsto d_j)$ to be the $c_j\leadsto d_j$ path in \Gint given by the following edges (in order):
    \begin{itemize}
        \item Start with the \magenta{magenta} edge $c_j \to \w_{1,j}^{1,r}$
        \item For each $i\in [k-1]$ use the $\w_{i,j}^{1,r} \leadsto \w_{i+1,j}^{1,r}$ path obtained by concatenating the $\w_{i,j}^{1,r}\leadsto \w_{i,j}^{N,r}$ path $\Row_{r}(D_{i,j})$ from~\autoref{def:Column-Row-Gint} with the \red{red} edge $\w_{i,j}^{N,r} \to \w_{i+1,j}^{1,r}$.
        \item Now, we have reached the vertex $\w_{k,j}^{1,r}$.
        Use the $\w_{k,j}^{1,r} \leadsto \w_{k,j}^{N,r}$ path $\Row_{r}(D_{k,j})$ from~\autoref{def:Column-Row-Gint} to reach the vertex $\w_{k,j}^{N,r}$.
        \item Finally, use the \magenta{magenta} edge $\w_{k,j}^{N,r} \to d_j$ to reach $d_j$.
    \end{itemize}
    \label{def:hori-canonical-G1}
\end{definition}
\medskip

\begin{definition}
    \textbf{(vertical canonical paths in \Gint)}
    Fix any $i\in [k]$. For each $r\in [N]$, we define $\CanInter(r\ ;\ a_i \leadsto b_i)$ to be the $a_i\leadsto b_i$ path in \Gint given by the following edges (in order):
    \begin{itemize}
        \item Start with the \magenta{magenta} edge $a_i \to \w_{i,1}^{r,1}$
        \item For each $j\in [k-1]$ use the $\w_{i,j}^{r,1} \leadsto \w_{i,j+1}^{r,1}$ path obtained by concatenating the $\w_{i,j}^{r,1}\leadsto \w_{i,j}^{r,N}$ path $\Column_{r}(D_{i,j})$ from~\autoref{def:Column-Row-Gint} with the \red{red} edge $\w_{i,j}^{r,N} \to \w_{i,j+1}^{r,1}$.
        \item Now, we have reached the vertex $\w_{i,k}^{r,1}$.
        Use the $\w_{i,k}^{r,1} \leadsto \w_{i,k}^{r,N}$ path $\Column_{r}(D_{i,k})$ from~\autoref{def:Column-Row-Gint} to reach the vertex $\w_{i,k}^{r,N}$.
        \item Finally, use the \magenta{magenta} edge $\w_{j,k}^{r,N} \to b_i$ to reach $b_i$.
    \end{itemize}

    \label{def:vert-canonical-G1}
\end{definition}
\medskip

The following observation measures the length (by counting the number of vertices) of every horizontal canonical path and vertical canonical path in \Gint.
\medskip

\begin{observation}
    \normalfont
    From~\autoref{def:hori-canonical-G1}, every horizontal canonical path in \Gint starts and ends with a \green{green} vertex. In the middle, this horizontal canonical path contains (all) the vertices from $k$ row-paths (\autoref{def:Column-Row-Gint}) which have $N$ (black) vertices each. Hence, each horizontal canonical path in \Gint contains exactly $kN+2$ vertices. A similar argument (using column-paths instead of row-paths) shows that each vertical canonical path in \Gint also contains exactly $kN+2$ vertices.
    \label{obs:Canonical-length}
\end{observation}
\medskip

We now set up notation for some special sets of vertices in \Gint, which helps to streamline some of the subsequent proofs.
\medskip

\begin{definition}
    \label{def:horizontal-vertical-sets-in-G1}
    \textbf{(horizontal \& vertical levels)}\\
    \begin{gather*}
        \text{For each}\ j\in [k],\ \text{set}\ \HorizontalInter(j):= \{ c_j, d_j \} \cup \left( \bigcup_{i=1}^{k} V(D_{i,j})\right)\\
        \text{For each}\ i\in [k],\ \text{set}\ \VerticalInter(i):= \{ a_i, b_i \} \cup \left( \bigcup_{j=1}^{k} V(D_{i,j})\right)\\
    \end{gather*}
    We also define the following ``border cases'':
    \begin{gather*}
        \HorizontalInter(0):=A \quad \text{and}\quad  \HorizontalInter(k+1)=B\\
        \VerticalInter(0):=C \quad \text{and}\quad  \VerticalInter(k+1):=D\\
    \end{gather*}
\end{definition}
\medskip

The next claim about the structure of $c_j\leadsto d_j$ paths in \Gint is used later in the proof of~\autoref{lem:horizontal-canonical-is-shortest-G1}.
\medskip

\begin{claim}
    \normalfont
    If $j\in [k]$, then every $c_j \leadsto d_j$ path in \Gint is contained in  $D_{\inter}\big[ \HorizontalInter(j) \big ]$.
    \label{clm:cj-dj-paths-G1}
\end{claim}
\begin{proof}
    The structure of \Gint (\autoref{fig:main}) allows us to make some simple observations about edges in \Gint:
    \begin{itemize}
        \item $N^{+}_{D_{\inter}}(c_j) \subseteq D_{1,j}$ and $N^{-}_{D_{\inter}}(d_j) \subseteq D_{k,j}$
        \item For each $0\leq j\leq k$, we have $N^{+}_{D_{\inter}}\left(\HorizontalInter(j)\right)\subseteq \left(\HorizontalInter(j+1)\right)$
        \item $N^{-}_{D_{\inter}}(A)=\emptyset=N^{+}_{D_{\inter}}(B)$
    \end{itemize}
    These three observations imply that if $j\in [k]$ and any $c_j\leadsto d_j$ path leaves $\HorizontalInter(j)$, then it could never return back to $\HorizontalInter(j)$. Since $c_j, d_j\in \HorizontalInter(j)$, every $c_j\leadsto d_j$ path in \Gint begins and ends at vertices of $\HorizontalInter(j)$. Therefore, we can conclude that every $c_j \leadsto d_j$ path in \Gint is contained in the induced subgraph $D_{\inter}\big[ \HorizontalInter(j) \big ]$.
\end{proof}
\medskip

The next lemma shows that if $j\in [k]$ then any shortest $c_j \leadsto d_j$ path in \Gint must be a horizontal canonical path and vice versa.

\medskip

\begin{lemma}
    \normalfont
    Let $j\in [k]$. The horizontal canonical paths in \Gint satisfy the following two properties:
    \begin{itemize}
        \item[(i)] For each $r\in [N]$, the path $\CanInter(r\ ;\ c_j \leadsto d_j)$ is a shortest $c_j \leadsto d_j$ path in \Gint.
        \item[(ii)] If $P$ is a shortest $c_j\leadsto d_j$ path in \Gint, then $P$ must be $\CanInter(\ell\ ;\ c_j \leadsto d_j)$ for some $\ell\in [N]$.
    \end{itemize}
    \label{lem:horizontal-canonical-is-shortest-G1}
\end{lemma}
\begin{proof}

    Consider any $c_j \leadsto d_j$ path, say $P$, in \Gint. By~\autoref{clm:cj-dj-paths-G1}, the path $P$ is completely contained in $D_{\inter}\big[ \HorizontalInter(j) \big ]$. Since $N^{+}_{D_{\inter}}(c_j)=\Le(D_{1,j})$ and $N^{-}_{D_{\inter}}(d_j)=\Ri(D_{k,j})$, it follows that the second vertex of $P$ must be from $\Le(D_{1,j})$ and the second-last vertex of $P$ must be from $\Ri(D_{k,j})$. Therefore, let the second and second-last vertices of $P$ be $\w_{1,j}^{1,\alpha}$ and $\w_{k,j}^{N,\beta}$ for some $1\leq \alpha, \beta\leq N$. We now make the following two observations:
    \begin{itemize}
        \item Since each horizontal black/\red{red} edge is oriented towards the right and each vertical black/\red{red} edge is oriented towards the top in $D_{\inter}\big[ \HorizontalInter(j) \big ]$, it follows that $\beta\geq \alpha$.
        \item For each $i\in [k]$ and each $\ell\in [N]$, let $\Columna_{\ell}(D_{i,j}):= \big\{ \w_{i,j}^{\ell,r}\ :\ 1\leq r\leq N \big\}$. From the structure of $D_{\inter}\big[ \HorizontalInter(j) \big ]$ it follows that $P$ contains at least one vertex from $\Columna_{\ell}(D_{i,j})$ for each $i\in [k]$ and each $\ell\in [N]$.
    \end{itemize}
    Therefore, the number of black vertices on $P$ is exactly $kN + (\beta-\alpha)\geq kN$. Remembering to add the first \green{green} vertex $c_j$ and last \green{green} vertex $d_j$, it follows that $P$ contains at least $kN+2$ vertices. The first part of the lemma now follows since each horizontal canonical path contains exactly $kN+2$ vertices (\autoref{obs:Canonical-length}).
    For the second part of the lemma: observe that if $P$ has length exactly equal to the length of a shortest $c_j\leadsto d_j$ path, then we have $kN+2 = 2+ kN + (\beta-\alpha)$ which implies $\beta=\alpha$. From the orientation of the edges within $D_{\inter}\big[ \HorizontalInter(j) \big ]$, it follows that $P$ is the path $\CanInter(\alpha\ ;\ c_j\leadsto d_j)$.
\end{proof}
\medskip

The proof of the next lemma is very similar to that of \autoref{lem:horizontal-canonical-is-shortest-G1}, and we skip repeating the details.

\medskip

\begin{lemma}
    \normalfont
    Let $i\in [k]$. The vertical canonical paths in \Gint satisfy the following two properties:
    \begin{itemize}
        \item For each $r\in [N]$, the path $\CanInter(r\ ;\ a_i \leadsto b_i)$ is a shortest $a_i \leadsto b_i$ path in \Gint.
        \item If $P$ is a shortest $a_i\leadsto b_i$ path in \Gint, then $P$ must be $\CanInter(\ell\ ;\ a_i \leadsto b_i)$ for some $\ell\in [N]$.
    \end{itemize}
    \label{lem:vertical-canonical-is-shortest-G1}
\end{lemma}
\medskip

\begin{remark}
    \label{rmk:hori-verti-sets-G}
    \normalfont
    \textbf{(reducing the in-degree and out-degree of \Gint)}
    The only vertices in \Gint which have out-degree greater than two are in $A\cup C$ and using Chitnis's technique from~\cite{rajesh-ciac-21}, we can reduce the out-degree of vertices from $C$, as follows: the argument for vertices from $A$ is analogous. Fix $j\in [k]$. The out-degree of $c_j$ is $N$ and $N^{+}_{D_{\inter}} = \big\{\w_{\LB}\ : \w\in \Le(D_{1,j})\big\}$. Replace the directed star, each of whose edges is from $c_j$ to a vertex of $\Le(D_{1,j})$, with a directed binary tree. This tree~$B$, whose root is $c_j$,  has leaves $\Le(D_{1,j})$ and each edge is directed away from the root. All non-leaf vertices of this binary tree are denoted by \green{green} color and all edges have \magenta{magenta} color. For simplicity, we assume that $N=2^{\ell}$ for some $\ell\in \mathbb{N}$. The only change for each terminal pair is that the path that started and ended with a \magenta{magenta} edge (equivalently, a \green{green} vertex) now starts and end with $\log_{2} N = \ell$ \magenta{magenta} edges (equivalently, $\log_{2} N = \ell$ \green{green} vertices). Hence, in the resulting graph, the out-degree and in-degree of every non-leaf vertex of $B$ is at most two, while the in-degree and out-degree of every leaf vertex of $B$ is unchanged (and hence exactly two). A similar argument also shows that we can reduce the in-degree of every vertex from $B\cup D$ to be at most two while preserving the correctness of the reduction from~\autoref{sec:construction-of-Gint}.

    It is easy to see that this editing of \Gint in~\autoref{rmk:hori-verti-sets-G}  adds $O(k\cdot N)$ new vertices and takes $\poly(N)$ time, and therefore it is still true (from~\autoref{clm:size-of-G-int}) that $n=|V\left(D_{\inter}\right)|=O(N^{2}k^{2})$ and \Gint can be constructed in $\poly(N,k)$ time.
\end{remark}
\medskip

	\section{Lower bounds for exact \& approximate \DEDSP on Planar DAGs}
\label{sec:fpt-inapprox-edsp-planar-dags}

The goal of this section is to prove lower bounds on the running time of exact (\autoref{thm:hardness-edge-result}) and approximate (\autoref{thm:inapprox-edge-result}) algorithms for the \DEDSP problem. We have already seen the first part of the reduction (\autoref{sec:construction-of-Gint}) from \kclique resulting in the construction of the intermediate graph \Gint.~\autoref{sec:construction-of-Gedge} describes the next part of the reduction which edits the intermediate \Gint to obtain the final graph \Gedge. This corresponds to the ancestry of the first leaf in~\autoref{fig:flowchart}. The characterization of shortest paths between terminal pairs in \Gedge is given in~\autoref{sec:characterizing-shortest-in-G-edge}. The completeness and soundness of the reduction from \kclique to Directed-$2k$-EDSP are proven in~\autoref{sec:clique-to-2kedsp} and~\autoref{sec:2kedsp-to-clique} respectively. Finally, everything is tied together in~\autoref{sec:proof-of-main-theorem-edge} allowing us to prove~\autoref{thm:hardness-edge-result} and~\autoref{thm:inapprox-edge-result}.

\subsection{Obtaining the graph \Gedge from \Gint\ via the splitting operation}
\label{sec:construction-of-Gedge}

We now define the splitting operation which allows us to obtain the graph \Gedge from the graph \Gint constructed in~\autoref{sec:construction-of-Gint}.
\medskip

\begin{definition}
    \textbf{(splitting operation to obtain \Gedge from \Gint)} For each $i,j\in [k]$ and each $q,\ell\in [N]$
    \begin{itemize}
        \item If $(q,\ell)\notin S_{i,j}$, then we \onesplit (\autoref{fig:split-edge-not}) the vertex $\w_{i,j}^{q,\ell}$ into \textbf{three distinct} vertices $\w_{i,j,\LB}^{q,\ell}, \w_{i,j,\Mid}^{q,\ell}$ and $\w_{i,j,\TR}^{q,\ell}$ and add the path $\w_{i,j,\LB}^{q,\ell}\to \w_{i,j,\Mid}^{q,\ell}\to \w_{i,j,\TR}^{q,\ell}$ (denoted by dotted edges in~\autoref{fig:split-edge-not}).
        \item Otherwise, if $(q,\ell)\in S_{i,j}$ then we \twosplit (\autoref{fig:split-edge-yes}) the vertex $\w_{i,j}^{q,\ell}$ into \textbf{four distinct} vertices $\w_{i,j,\LB}^{q,\ell}, \w_{i,j,\Hor}^{q,\ell}, \w_{i,j,\Ver}^{q,\ell}$ and $\w_{i,j,\TR}^{q,\ell}$ and add the two paths $\w_{i,j,\LB}^{q,\ell}\to \w_{i,j,\Hor}^{q,\ell}\to \w_{i,j,\TR}^{q,\ell}$ and $\w_{i,j,\LB}^{q,\ell}\to \w_{i,j,\Ver}^{q,\ell}\to \w_{i,j,\TR}^{q,\ell}$ (denoted by dotted edges in~\autoref{fig:split-edge-yes}).
    \end{itemize}
    The 4 edges (\autoref{def:lefty-right-topy-bottomy-dir}) incident on $\w_{i,j}^{q,\ell}$ are now changed as follows:
    \begin{itemize}
        \item Replace the edge $\lefty(\w_{i,j}^{q,\ell})\to \w_{i,j}^{q,\ell}$ by the edge $\lefty(\w_{i,j}^{q,\ell})\to \w_{i,j,\LB}^{q,\ell}$
        \item Replace the edge $\bottomy(\w_{i,j}^{q,\ell})\to \w_{i,j}^{q,\ell}$ by the edge $\bottomy(\w_{i,j}^{q,\ell})\to \w_{i,j,\LB}^{q,\ell}$
        \item Replace the edge $\w_{i,j}^{q,\ell}\to \righty(\w_{i,j}^{q,\ell})$ by the edge $\w_{i,j,\TR}^{q,\ell}\to \righty(\w_{i,j}^{q,\ell})$
        \item Replace the edge $\w_{i,j}^{q,\ell}\to \topy(\w_{i,j}^{q,\ell})$ by the edge $\w_{i,j,\TR}^{q,\ell}\to \topy(\w_{i,j}^{q,\ell})$
    \end{itemize}
    \label{def:splitting-operation-edge}
\end{definition}
\medskip

\begin{figure}[hbt!]
\centering
\begin{tikzpicture}[
vertex/.style={circle, draw=black, fill=black, text width=1.5mm, inner sep=0pt},
scale=0.65]
\node[vertex, label=above right:\footnotesize{$\w_{i,j}^{q,\ell}$}] (v) at (0,0) {} ;
\node[vertex, label=above:\footnotesize{$\lefty(\w_{i,j}^{q,\ell})$}] (l) at (-2,0) {};
\node[vertex, label=below:\footnotesize{$\righty(\w_{i,j}^{q,\ell})$}] (r) at (2,0) {};
\node[vertex, label=below:\footnotesize{$\bottomy(\w_{i,j}^{q,\ell})$}] (b) at (0,-2) {};
\node[vertex, label=above:\footnotesize{$\topy(\w_{i,j}^{q,\ell})$}] (t) at (0,2) {};
\draw[ultra thick, middlearrow={>}] (v) -- (t);
\draw[ultra thick, middlearrow={>}] (v) -- (r);
\draw[ultra thick, middlearrow={>}] (l) -- (v);
\draw[ultra thick, middlearrow={>}] (b) -- (v);

\draw[orange,double, ultra thick,->] (3,0) -- node[above=3mm, draw=none, fill=none, rectangle] {\onesplit } (5,0);

\node[vertex, label=below:\footnotesize{$\w_{i,j,\TR}^{q,\ell}$}] (vtr) at (12,0) {} ;
\node[vertex, label=below:\footnotesize{$\w_{i,j,\Mid}^{q,\ell}$}] (vm) at (10,0) {} ;
\node[vertex, label=above:$\w_{i,j,\LB}^{q,\ell}$] (vlb) at (8,-0) {} ;

\node[vertex, label=below:\footnotesize{$\lefty(\w_{i,j}^{q,\ell})$}] (l) at (6,0) {};
\node[vertex, label=above:\footnotesize{$\righty(\w_{i,j}^{q,\ell})$}] (r) at (14,0) {};
\node[vertex, label=below:\footnotesize{$\bottomy(\w_{i,j}^{q,\ell})$}] (b) at (8,-2) {};
\node[vertex, label=above:\footnotesize{$\topy(\w_{i,j}^{q,\ell})$}] (t) at (12,2) {};
\draw[ultra thick, middlearrow={>}] (vtr) -- (t);
\draw[ultra thick, middlearrow={>}] (vtr) -- (r);
\draw[ultra thick, middlearrow={>}] (l) -- (vlb);
\draw[ultra thick, middlearrow={>}] (b) -- (vlb);
\draw[dotted,ultra thick,middlearrow={>}] (vlb) -- (vm);
\draw[dotted,ultra thick,middlearrow={>}] (vm) -- (vtr);
\end{tikzpicture}

\caption{The \onesplit operation for the vertex $\w_{i,j}^{q,\ell}$ when
$(q,\ell)\notin S_{i,j}$. The idea behind this splitting is that the horizontal path $\lefty(w_{i,j}^{q,\ell})\to w_{i,j}^{q,\ell}\to \righty(w_{i,j}^{q,\ell})$ and vertical path $\bottomy(w_{i,j}^{q,\ell})\to w_{i,j}^{q,\ell}\to \topy(w_{i,j}^{q,\ell})$  are no longer edge-disjoint after the \onesplit operation as they must share the path $w_{i,j,\LB}^{q,\ell}\to w_{i,j,\Mid}^{q,\ell}\to w_{i,j,\TR}^{q,\ell}$.}
\label{fig:split-edge-not}
\end{figure}
\begin{figure}[hbt!]
\centering
\begin{tikzpicture}[
vertex/.style={circle, draw=black, fill=black, text width=1.5mm, inner sep=0pt},
scale=0.65]
\node[vertex, label=above right:\footnotesize{$\w_{i,j}^{q,\ell}$}] (v) at (0,0) {} ;
\node[vertex, label=above:\footnotesize{$\lefty(\w_{i,j}^{q,\ell})$}] (l) at (-2,0) {};
\node[vertex, label=below:\footnotesize{$\righty(\w_{i,j}^{q,\ell})$}] (r) at (2,0) {};
\node[vertex, label=below:\footnotesize{$\bottomy(\w_{i,j}^{q,\ell})$}] (b) at (0,-2) {};
\node[vertex, label=above:\footnotesize{$\topy(\w_{i,j}^{q,\ell})$}] (t) at (0,2) {};
\draw[ultra thick, middlearrow={>}] (v) -- (t);
\draw[ultra thick, middlearrow={>}] (v) -- (r);
\draw[ultra thick, middlearrow={>}] (l) -- (v);
\draw[ultra thick, middlearrow={>}] (b) -- (v);

\draw[orange,double, ultra thick,->] (3,0) -- node[above=3mm, draw=none, fill=none, rectangle] {\twosplit } (5,0);

\node[vertex, label=below:\footnotesize{$\w_{i,j,\TR}^{q,\ell}$}] (vtr) at (12,0) {} ;
\node[vertex, label=below:\footnotesize{$\w_{i,j,\Hor}^{q,\ell}$}] (vh) at (10,-1) {} ;
\node[vertex, label=above:\footnotesize{$\w_{i,j,\Ver}^{q,\ell}$}] (vv) at (10,1) {} ;
\node[vertex, label=above:$\w_{i,j,\LB}^{q,\ell}$] (vlb) at (8,-0) {} ;

\node[vertex, label=below:\footnotesize{$\lefty(\w_{i,j}^{q,\ell})$}] (l) at (6,0) {};
\node[vertex, label=above:\footnotesize{$\righty(\w_{i,j}^{q,\ell})$}] (r) at (14,0) {};
\node[vertex, label=below:\footnotesize{$\bottomy(\w_{i,j}^{q,\ell})$}] (b) at (8,-2) {};
\node[vertex, label=above:\footnotesize{$\topy(\w_{i,j}^{q,\ell})$}] (t) at (12,2) {};
\draw[ultra thick, middlearrow={>}] (vtr) -- (t);
\draw[ultra thick, middlearrow={>}] (vtr) -- (r);
\draw[ultra thick, middlearrow={>}] (l) -- (vlb);
\draw[ultra thick, middlearrow={>}] (b) -- (vlb);
\draw[dotted,ultra thick,middlearrow={>}] (vlb) -- (vh);
\draw[dotted,ultra thick,middlearrow={>}] (vh) -- (vtr);

\draw[dotted,ultra thick,middlearrow={>}] (vlb) -- (vv);
\draw[dotted,ultra thick,middlearrow={>}] (vv) -- (vtr);
\end{tikzpicture}

\caption{The \twosplit operation for the vertex $\w_{i,j}^{q,\ell}$ when
$(q,\ell)\in S_{i,j}$. The idea behind this splitting is that the horizontal path $\lefty(w_{i,j}^{q,\ell})\to w_{i,j}^{q,\ell}\to \righty(w_{i,j}^{q,\ell})$ and vertical path $\bottomy(w_{i,j}^{q,\ell})\to w_{i,j}^{q,\ell}\to \topy(w_{i,j}^{q,\ell})$  are still edge-disjoint after the \twosplit operation if we replace them with the paths $\lefty(w_{i,j}^{q,\ell})\to w_{i,j,\LB}^{q,\ell}\to w_{i,j,\Hor}^{q,\ell}\to w_{i,j,\TR}^{q,\ell}\to \righty(w_{i,j}^{q,\ell})$ and $\bottomy(w_{i,j}^{q,\ell})\to w_{i,j,\LB}^{q,\ell}\to w_{i,j,\Ver}^{q,\ell}\to w_{i,j,\TR}^{q,\ell}\to \topy(w_{i,j}^{q,\ell})$ respectively.
}
\label{fig:split-edge-yes}

\end{figure}

Finally, we are now ready to define the instance of Directed-$2k$-EDSP that we have built starting from an instance $G$ of \kclique.
\medskip

\begin{definition}
    \normalfont
    \textbf{(defining the $2k$-\EDSP instance)} The instance $(D_{\edge}, \mathcal{T})$ of Directed-$2k$-EDSP is defined as follows:
    \begin{itemize}
        \item The graph \Gedge is obtained by applying the splitting operation (\autoref{def:splitting-operation-edge}) to each (black) grid vertex of \Gint, i.e., the set of vertices given by $\bigcup_{1\leq i,j\leq k} V(D_{i,j})$.
        \item No \green{green} vertex is split in~\autoref{def:splitting-operation-edge}, and hence the set of terminal pairs remains the same as defined in~\autoref{eqn:definition-of-mathcal-T} and is given by $\mathcal{T}:= \big\{(a_i, b_i) : i\in [k] \big\}\cup \big\{(c_j, d_j) : j\in [k] \big\}$.
        \item We assign a cost of one to visit each of the vertices present after the splitting operation (\autoref{def:splitting-operation-edge}). Since each vertex in \Gint has a cost of one, it follows that each vertex in \Gedge also has a cost of one.

    \end{itemize}
    \label{def:G-edge}.
\end{definition}
\medskip

The next two claims analyze the structure and size of the graph \Gedge.
\medskip

\begin{claim}
    \label{clm:Gedge-is-planar-and-dag}
    \normalfont
    \Gedge is a planar DAG.
\end{claim}
\begin{proof}
    In~\autoref{clm:G1-is-planar-and-dag}, we have shown that \Gint is a planar DAG. The graph \Gedge is obtained from \Gint by applying the splitting operation (\autoref{def:splitting-operation-edge}) on every (black) grid vertex, i.e., every vertex from the set $\bigcup_{1\leq i,j\leq k} V(D_{i,j})$. By~\autoref{def:lefty-right-topy-bottomy-dir}, every vertex of \Gint that is split has exactly two in-neighbors and two out-neighbors in \Gint. Hence, one can observe (\autoref{fig:split-edge-not} and~\autoref{fig:split-edge-yes}) that the splitting operation (\autoref{def:splitting-operation-edge}) does not destroy planarity when we construct \Gedge from \Gint.

    Since \Gint is a DAG, it has a topological order say $\mathcal{X}$. The only changes done when going from \Gint to \Gedge are the addition of new vertices and edges when black grid vertices are split according to~\autoref{def:splitting-operation-edge}. We now explain how to modify $\mathcal{X}$ to obtain a topological order $\mathcal{X}'$ for \Gedge:
    \begin{itemize}
        \item If a black grid vertex $\w$ is \onesplit then we replace $\w$ by the following vertices (in order) $\w_{\LB}, \w_{\Mid}, \w_{\TR}$.
        \item If a black grid vertex $\w$ is \twosplit then we replace $\w$ by the following vertices (in order) $\w_{\LB}, \w_{\Hor}, \w_{\Ver}, \w_{\TR}$.
    \end{itemize}
    It is easy to see from~\autoref{fig:split-edge-not} and~\autoref{fig:split-edge-yes} that $\mathcal{X}'$ is a topological order for \Gedge.
\end{proof}
\medskip

\begin{claim}
    \normalfont
    The number of vertices in \Gedge is $O(N^{2}k^{2})$.
    \label{clm:size-of-G-edge}
\end{claim}
\begin{proof}
    The only change in going from \Gint to \Gedge is the splitting operation (\autoref{def:splitting-operation-edge}). If a black grid vertex $\w$ in \Gint is \onesplit (\autoref{fig:split-edge-not}) then we replace it by \textbf{three} vertices $\w_{\LB}, \w_{\Mid}, \w_{\TR}$ in \Gedge. If a black grid vertex $\w$ in \Gint is \twosplit (\autoref{fig:split-edge-yes}) then we replace it by \textbf{four} vertices $\w_{\LB}, \w_{\Hor}, \w_{\Ver}, \w_{\TR}$ in \Gedge. In both cases, the increase in number of vertices is only by a constant factor. The number of vertices in \Gint is $O(N^2 k^2)$ from~\autoref{clm:size-of-G-int}, and hence it follows that the number of vertices in \Gedge is $O(N^2 k^2)$.
\end{proof}
\medskip

\subsection{Characterizing shortest paths in \Gedge}
\label{sec:characterizing-shortest-in-G-edge}

The goal of this section is to characterize the structure of shortest paths between terminal pairs in \Gedge. Recall (\autoref{def:G-edge}) that the set of terminal pairs is given by $\mathcal{T}:= \big\{(a_i, b_i) : i\in [k] \big\}\cup \big\{(c_j, d_j) : j\in [k] \big\}$. Since each edge of \Gedge has length one (\autoref{def:G-edge}), we measure the length of paths in \Gedge by counting the number of vertices.

We now define canonical paths in \Gedge by adapting the definition of canonical paths (\autoref{def:hori-canonical-G1} and~\autoref{def:vert-canonical-G1}) in \Gint in accordance with the changes in going from \Gint to \Gedge.
\medskip

\begin{definition}
    \textbf{(horizontal canonical paths in \Gedge)}
    Fix any $j\in [k]$. For each $r\in [N]$, we define $\CanEdge(r\ ;\ c_j \leadsto d_j)$ to be the $c_j\leadsto d_j$ path in \Gedge obtained from the path $\CanInter(r\ ;\ c_j \leadsto d_j)$ in \Gint (recall~\autoref{def:hori-canonical-G1}) in the following way:
    \begin{itemize}
        \item The first and last \magenta{magenta} edges are unchanged
        \item If a black grid vertex $\w$ from $\CanInter(r\ ;\ c_j \leadsto d_j)$ is \onesplit (\autoref{fig:split-edge-not}), then
        \begin{itemize}
            \item The unique incoming edge into $\w$ is changed to be incoming into $\w_{\LB}$
            \item The unique outgoing edge from $\w$ is changed to be outgoing from $\w_{\TR}$
            \item The path $\w_{\LB} \to \w_{\Mid} \to \w_{\TR}$ is added
        \end{itemize}
        \item If a black grid vertex $\w$ from $\CanInter(r\ ;\ c_j \leadsto d_j)$ is \twosplit (\autoref{fig:split-edge-yes}), then
        \begin{itemize}
            \item The unique incoming edge into $\w$ is changed to be incoming into $\w_{\LB}$
            \item The unique outgoing edge from $\w$ is changed to be outgoing from $\w_{\TR}$
            \item The path $\w_{\LB} \to \w_{\Hor} \to \w_{\TR}$ is added
        \end{itemize}
    \end{itemize}
    \label{def:hori-canonical-Gedge}
\end{definition}
\medskip

\begin{definition}
    \textbf{(vertical canonical paths in \Gedge)}
    Fix any $i\in [k]$. For each $r\in [N]$, we define $\CanEdge(r\ ;\ a_i \leadsto b_i)$ to be the $a_i\leadsto b_i$ path in \Gedge obtained from the path $\CanInter(r\ ;\ a_i \leadsto b_i)$ in \Gint (recall \autoref{def:vert-canonical-G1}) in the following way.
    \begin{itemize}
        \item The first and last \magenta{magenta} edges are unchanged
        \item If a black grid vertex $\w$ from $\CanInter(r\ ;\ a_i \leadsto b_i)$ is \onesplit (\autoref{fig:split-edge-not}), then
        \begin{itemize}
            \item The unique incoming edge into $\w$ is changed to be incoming into $\w_{\LB}$
            \item The unique outgoing edge from $\w$ is changed to be outgoing from $\w_{\TR}$
            \item The path $\w_{\LB} \to \w_{\Mid} \to \w_{\TR}$ is added
        \end{itemize}
        \item If a black grid vertex $\w$ from $\CanInter(r\ ;\ a_i \leadsto b_i)$ is \twosplit (\autoref{fig:split-edge-yes}), then
        \begin{itemize}
            \item The unique incoming edge into $\w$ is changed to be incoming into $\w_{\LB}$
            \item The unique outgoing edge from $\w$ is changed to be outgoing from $\w_{\TR}$
            \item The path $\w_{\LB} \to \w_{\Ver} \to \w_{\TR}$ is added
        \end{itemize}
    \end{itemize}
    \label{def:verti-canonical-Gedge}
\end{definition}
\medskip

\begin{definition}
    \textbf{(Image of a horizontal canonical path from \Gint in \Gedge)}
    Fix a $j\in [k]$ and $r\in [N]$. For each $\CanInter(r\ ;\ c_j \leadsto d_j)$ path $R$ in \Gint, we define an image of R as follows
    \begin{itemize}
        \item The first and last \magenta{magenta} edges are unchanged.
        \item If a black grid vertex $\w$ from $\CanInter(r\ ;\ c_j \leadsto d_j)$ is \onesplit (\autoref{fig:split-edge-not}), then
        \begin{itemize}
            \item The unique edge $\lefty(\w) \to \w$ is replaced with the edge $\lefty(\w) \to \w_{\LB}$;
            \item The unique edge $\w \to \righty(\w)$ is replaced with the edge $\w_{\TR} \to \righty(\w)$;
            \item The path $\w_{\LB} \to \w_{\Mid} \to \w_{\TR}$ is added.
        \end{itemize}
        \item If a black grid vertex $\w$ from $\CanInter(r\ ;\ c_j \leadsto d_j)$ is \twosplit (\autoref{fig:split-edge-yes}), then
        \begin{itemize}
            \item The unique edge $\lefty(\w) \to \w$ is replaced with the edge $\lefty(\w) \to \w_{\LB}$;
            \item The unique edge $\w \to \righty(\w)$ is replaced with the edge $\w_{\TR} \to \righty(\w)$;
            \item Either the edges $\w_{\LB} \to \w_{\Hor} \to \w_{\TR}$ or $\w_{\LB} \to \w_{\Ver} \to \w_{\TR}$ are added.
        \end{itemize}
    \end{itemize}
    \label{def:hor-image-of-a-path-G-edge}
\end{definition}
\medskip

\begin{definition}
    \textbf{(Image of a vertical canonical path from \Gint in \Gedge)}
    Fix a $i\in [k]$ and $r\in [N]$. For each $\CanInter(r\ ;\ a_i \leadsto b_i)$ path $R$ in \Gint, we define an image of R as follows
    \begin{itemize}
        \item The first and last \magenta{magenta} edges are unchanged.
        \item If a black grid vertex $\w$ from $\CanInter(r\ ;\ a_i \leadsto b_i)$ is \onesplit (\autoref{fig:split-edge-not}), then
        \begin{itemize}
            \item The unique edge $\lefty(\w) \to \w$ is replaced with the edge $\lefty(\w) \to \w_{\LB}$;
            \item The unique edge $\w \to \righty(\w)$ is replaced with the edge $\w_{\TR} \to \righty(\w)$;
            \item The path $\w_{\LB} \to \w_{\Mid} \to \w_{\TR}$ is added.
        \end{itemize}
        \item If a black grid vertex $\w$ from $\CanInter(r\ ;\ a_i \leadsto b_i)$ is \twosplit (\autoref{fig:split-edge-yes}), then
        \begin{itemize}
            \item The unique edge $\lefty(\w) \to \w$ is replaced with the edge $\lefty(\w) \to \w_{\LB}$;
            \item The unique edge $\w \to \righty(\w)$ is replaced with the edge $\w_{\TR} \to \righty(\w)$;
            \item Either the edges $\w_{\LB} \to \w_{\Hor} \to \w_{\TR}$ or $\w_{\LB} \to \w_{\Ver} \to \w_{\TR}$ are added.
        \end{itemize}
    \end{itemize}
    \label{def:vert-image-of-a-path-G-edge}
\end{definition}
\medskip

Note that a single path, $R$, in \Gint can have several images in \Gedge. This is because for every black vertex on $R$ that is \twosplit there are two choices of sub-path to add: either the path $\w_{\LB} - \w_{\Hor} - \w_{\TR}$ or the path $\w_{\LB} - \w_{\Ver} - \w_{\TR}$.

The following two lemmas (\autoref{lem:horizontal-canonical-is-shortest-G-edge} and~\autoref{lem:vertical-canonical-is-shortest-G-edge}) analyze the structure of shortest paths between terminal pairs in \Gedge. First, we define the image of a path from \Gint in the graph \Gedge.
\medskip

\begin{lemma}
    \normalfont
    The shortest paths in \Gedge satisfy the following two properties:
    \begin{itemize}
        \item[(i)] For each $r\in [N]$, the horizontal canonical path $\CanEdge(r\ ;\ c_j \leadsto d_j)$ is a shortest $c_j \leadsto d_j$ path in \Gedge.
        \item[(ii)] If $P$ is a shortest $c_j\leadsto d_j$ path in \Gedge, then $P$ must be an image (\autoref{def:hor-image-of-a-path-G-edge}) of the path $\CanInter(\ell\ ;\ c_j \leadsto d_j)$ for some $\ell\in [N]$.
    \end{itemize}
    \label{lem:horizontal-canonical-is-shortest-G-edge}
\end{lemma}
\begin{proof}
    The proof of this lemma can be obtained in the same way as shown for \Gint in \autoref{lem:horizontal-canonical-is-shortest-G1} with some minor observational changes. Note that any path in \Gint contains only \green{green} and black vertices. The splitting operation (\autoref{def:splitting-operation-edge}) applied to each black vertex of \Gint has the following property: if a path $Q$ contains a black vertex $\w$ in \Gint, then in the corresponding path in \Gedge this vertex $\w$ is \textbf{always replaced by three vertices}:
    \begin{itemize}
        \item If $\w$ is \onesplit (\autoref{fig:split-edge-not}), then it is replaced in $Q$ the three vertices $\w_{\LB}, \w_{\Mid}, \w_{\TR}$.
        \item If $\w$ is \twosplit (\autoref{fig:split-edge-yes}), then it is replaced in $Q$ either by the three vertices $\w_{\LB}, \w_{\Hor}, \w_{\TR}$ or the three vertices $\w_{\LB}, \w_{\Ver}, \w_{\TR}$.
    \end{itemize}
    Therefore, if a path $Q$ contains $\alpha$ \green{green} vertices and $\beta$ black vertices in \Gint, then the corresponding path in \Gedge contains $\alpha$ \green{green} vertices and $3\beta$ black vertices. The proof of the first part of the lemma now follows from~\autoref{lem:horizontal-canonical-is-shortest-G1}(i),~\autoref{def:splitting-operation-edge} and~\autoref{def:hori-canonical-Gedge}. The proof of the second part of the lemma follows from~\autoref{lem:horizontal-canonical-is-shortest-G1}(ii),~\autoref{def:splitting-operation-edge} and~\autoref{def:hor-image-of-a-path-G-edge}.
\end{proof}
\medskip

The proof of the next lemma is very similar to that of \autoref{lem:horizontal-canonical-is-shortest-G-edge}, and we skip repeating the details.
\medskip

\begin{lemma}
    \normalfont
    The shortest paths in \Gedge satisfy the following two properties:
    \begin{itemize}
        \item[(i)] For each $r\in [N]$, the vertical canonical path $\CanEdge(r\ ;\ a_i \leadsto b_i)$ is a shortest $a_i \leadsto b_i$ path in \Gedge.
        \item[(ii)] If $P$ is a shortest $a_i\leadsto b_i$ path in \Gedge, then $P$ must be an image (\autoref{def:vert-image-of-a-path-G-edge}) of the path $\CanInter(\ell\ ;\ a_i \leadsto b_i)$ for some $\ell\in [N]$.
    \end{itemize}
    \label{lem:vertical-canonical-is-shortest-G-edge}
\end{lemma}
\medskip

\subsection{\texorpdfstring{\underline{Completeness}}{Completeness}: \texorpdfstring{$G$}{G} has a \texorpdfstring{$k$}{k}-clique \texorpdfstring{$\Rightarrow$}{->} All pairs in the instance \texorpdfstring{$(D_{\edge}, \mathcal{T})$}{(Dedge, T)} of Directed-\texorpdfstring{$2k$}{2k}-EDSP can be satisfied}
\label{sec:clique-to-2kedsp}

In this section, we show that if the instance $G$ of \kclique has a solution then the instance $(D_{\edge}, \mathcal{T})$ of Directed-$2k$-\EDSP also has a solution.
Suppose the instance $G=(V,E)$ of \kclique has a clique $X=\{v_{\gamma_1}, v_{\gamma_2}, \ldots, v_{\gamma_k} \}$ of size $k$. Let $Y=\{\gamma_1, \gamma_2, \ldots, \gamma_k\}\in [N]$. Now for each $i\in [k]$ we choose the path as follows:
\begin{itemize}
    \item The path $R_i$ to satisfy $a_i\leadsto b_i$ is chosen to be the horizontal canonical path $\CanEdge(\gamma_i \ ;\ a_i \leadsto b_i)$ described in~\autoref{def:hori-canonical-Gedge}.
    \item The path $T_i$ to satisfy $c_i\leadsto d_i$ is chosen to be vertical canonical path $\CanEdge(\gamma_i \ ;\ c_i \leadsto d_i)$ described in~\autoref{def:verti-canonical-Gedge}.
\end{itemize}

Now we show that the collection of paths given by $\mathcal{Q}:=\{R_1, R_2, \ldots, R_k, T_1, T_2, \ldots, T_K\}$ forms a solution for the instance  $(D_{\edge}, \mathcal{T})$ of Directed-$2k$-\EDSP via the following two lemmas which argue being shortest for each terminal pair and pairwise edge-disjointness respectively:
\medskip

\begin{lemma}
    \normalfont
    For each $i \in [k]$, the path $R_i$ (resp. $T_i$) is a shortest $a_i \leadsto b_i$ (resp. $c_i \leadsto d_i$) path in \Gedge.
    \label{lem:completeness-edsp-shortest}
\end{lemma}
\begin{proof}
    Fix any $i\in [k]$.~\autoref{lem:horizontal-canonical-is-shortest-G-edge}(i) implies that $T_i$ is shortest $c_i \leadsto d_i$ path in \Gedge.~\autoref{lem:vertical-canonical-is-shortest-G-edge}(i) implies that $R_i$ is shortest $a_i \leadsto b_i$ path in \Gedge.
\end{proof}
\medskip

Before proving~\autoref{lem:completeness-edsp-disjoint}, we first set up notation for some special sets of vertices in \Gedge which helps to streamline some of the subsequent proofs.

\medskip

\begin{definition}
    \label{def:horizontal-vertical-sets-in-Gedge}
    \textbf{(horizontal \& vertical levels in \Gedge)}
    For each $(i,j)\in [k]\times [k]$, let $D_{i,j}^{\edgesplitt}$ to be the graph obtained by applying the splitting operation (\autoref{def:splitting-operation-edge}) to each vertex of $D_{i,j}$. For each $j\in [k]$, we define the following set of vertices:
    \begin{equation}
        \begin{aligned}
            \HorizontalEdge(j)     & = \{ c_j, d_j \} \cup \left( \bigcup_{i=1}^{k} V(D_{i,j}^{\edgesplitt})\right)\ \quad \\
            \quad \VerticalEdge(j) & = \{ a_j, b_j \} \cup \left( \bigcup_{i=1}^{k} V(D_{j,i}^{\edgesplitt})\right)
        \end{aligned}\label{eqn:hori-verti-sets-Gedge}
    \end{equation}
\end{definition}
\medskip

The next lemma shows that any two paths from $\mathcal{Q}$ are edge-disjoint.
\medskip

\begin{lemma}
    \normalfont
    Let $P\neq P'$ be any pair of paths from the collection  $\mathcal{Q}=\{R_1, R_2, \ldots, R_k, T_1, T_2, \ldots, T_K\}$. Then $P$ and $P'$ are edge-disjoint.
    \label{lem:completeness-edsp-disjoint}
\end{lemma}
\begin{proof}

    By~\autoref{def:horizontal-vertical-sets-in-Gedge}, it follows that every edge of the path $R_i$ has both endpoints in $\VerticalEdge(i)$ for every $i\in [k]$. Since $\VerticalEdge(i) \cap \VerticalEdge(i')=\emptyset$ for every $1\leq i\neq i'\leq k$, it follows that the collection of paths $\{R_1, R_2, \ldots, R_k\}$ are pairwise edge-disjoint.

    By~\autoref{def:horizontal-vertical-sets-in-Gedge}, it follows that every edge of the path $T_j$ has both endpoints in $\HorizontalEdge(j)$ for every $j\in [k]$. Since  $\HorizontalEdge(j) \cap \HorizontalEdge(j')=\emptyset$ for every $1\leq j\neq j'\leq k$, it follows that the collection of paths $\{T_1, T_2, \ldots, T_k\}$ are pairwise edge-disjoint.

    It remains to show that every pair of paths which contains one path from $\{R_1, R_2, \ldots, R_k\}$ and other path from $\{T_1, T_2, \ldots, T_k\}$ are edge-disjoint.
    \begin{claim}
        \label{clm:reduction-edge-disjoint}
        \normalfont
        For each $(i,j)\in [k]\times [k]$, the paths $R_i$ and $T_j$ are edge-disjoint in \Gedge.
    \end{claim}

    \begin{proof}
        Fix any $(i,j)\in [k]\times [k]$. First we argue that the vertex $\w_{i,j}^{\gamma_i, \gamma_j}$ is \twosplit, i.e., $(\gamma_i, \gamma_j)\in S_{i,j}$:
        \begin{itemize}
            \item If $i=j$ then $\gamma_i = \gamma_j$ and hence by~\autoref{eqn:clique-to-gt-reduction-D} we have $(\gamma_i, \gamma_j)\in S_{i,j}$
            \item If $i\neq j$, then $v_{\gamma_i} - v_{\gamma_j}\in E(G)$ since $X$ is a clique. Again, by~\autoref{eqn:clique-to-gt-reduction-D} we have $(\gamma_i, \gamma_j)\in S_{i,j}$.
        \end{itemize}
        Hence, by~\autoref{def:splitting-operation-edge}, it follows that the vertex $\w_{i,j}^{\gamma_i, \gamma_j}$ is \twosplit.

        By the construction of \Gint (\autoref{fig:main}) and definitions of canonical paths (\autoref{def:hori-canonical-G1} and~\autoref{def:vert-canonical-G1}), it is easy to verify that any pair of horizontal canonical path and vertical canonical path in \Gint are edge-disjoint and have only one vertex in common.

        By the splitting operation (\autoref{def:splitting-operation-edge}) and definitions of the paths $R_i$ (\autoref{def:verti-canonical-Gedge}) and $T_j$ (\autoref{def:hori-canonical-Gedge}), it follows that the only common edges between $R_i$ and $T_j$ must be from paths in \Gedge that start at $\w_{i,j,\LB}^{\gamma_i,\gamma_j}$ and end at $\w_{i,j,\TR}^{\gamma_i,\gamma_j}$. Since $\w_{i,j}^{\gamma_i, \gamma_j}$ is \twosplit, we have
        \begin{itemize}
            \item By~\autoref{def:verti-canonical-Gedge}, the unique $\w_{i,j,\LB}^{\gamma_i,\gamma_j}\leadsto \w_{i,j,\TR}^{\gamma_i,\gamma_j}$ sub-path of $R_i$ is $\w_{i,j,\LB}^{\gamma_i,\gamma_j}\to \w_{i,j,\Ver}^{\gamma_i,\gamma_j} \to  \w_{i,j,\TR}^{\gamma_i,\gamma_j}$.
            \item By~\autoref{def:hori-canonical-Gedge}, the unique $\w_{i,j,\LB}^{\gamma_i,\gamma_j}\leadsto \w_{i,j,\TR}^{\gamma_i,\gamma_j}$ sub-path of $T_i$ is $\w_{i,j,\LB}^{\gamma_i,\gamma_j}\to \w_{i,j,\Hor}^{\gamma_i,\gamma_j} \to  \w_{i,j,\TR}^{\gamma_i,\gamma_j}$.
        \end{itemize}
        Hence, it follows that $R_i$ and $T_j$ are edge-disjoint.
    \end{proof}
    This concludes the proof of~\autoref{lem:completeness-edsp-disjoint}.
\end{proof}
\medskip

\noindent From~\autoref{lem:completeness-edsp-shortest} and~\autoref{lem:completeness-edsp-disjoint}, it follows that the collection of paths given by $\mathcal{Q}=\{R_1, R_2, \ldots, R_k,$ $T_1, T_2, \ldots, T_K\}$ forms a solution for the instance $(D_{\edge}, \mathcal{T})$ of Directed-$2k$-\EDSP.

\subsection{\texorpdfstring{\underline{Soundness}}{Soundness}: \texorpdfstring{$(\frac{1}{2} +\epsilon)$}{(1/2 + theta)}-fraction of the pairs in the instance \texorpdfstring{$(D_{\edge}, \mathcal{T})$}{(Dedge,T)} of Directed-\texorpdfstring{$2k$}{k}-EDSP can be satisfied \texorpdfstring{$\Rightarrow$ $G$}{-> G} has a clique of size \texorpdfstring{$\geq 2\epsilon \cdot k$}{>= theta x k}}
\label{sec:2kedsp-to-clique}

In this section we show that if at least $(\frac{1}{2} +\epsilon)$-fraction of the $2k$ pairs from the instance $(D_{\edge}, \mathcal{T})$ of Directed-$2k$-\EDSP can be satisfied then the graph $G$ has a clique of size $2\epsilon \cdot k$.
Let $\mathcal{P}$ be a collection of paths in \Gedge which satisfies at least $(\frac{1}{2} +\epsilon)$-fraction of the $2k$ terminal pairs from the instance $(D_{\edge}, \mathcal{T})$ of Directed-$2k$-\EDSP.
\medskip

\begin{definition}
    An index $i \in [k]$ is called \emph{good} if both the terminal pairs $a_i \leadsto b_i$ and $c_i \leadsto d_i$ are satisfied by $\mathcal{P}$.
    \label{def:good-for-edge}
\end{definition}
\medskip

The next lemma gives a lower bound on the number of good indices.
\medskip

\begin{lemma}
    \label{lem:good_size}
    \normalfont
    Let $Y \subseteq [k]$ be the set of good indices. Then $|Y| \ge 2\epsilon \cdot k$.
\end{lemma}
\begin{proof}
    If $i\in [k]$ is good then both the pairs $a_i\leadsto b_i$ and $c_i\leadsto d_i$ are satisfied by $\mathcal{P}$. Otherwise, at most one of these pairs $a_i\leadsto b_i$ and $c_i\leadsto d_i$ is satisfied. Hence, the total number of satisfied pairs is at
    most $2\cdot |Y|+1\cdot (k-|Y|)= k+|Y|$. However, we know that $\mathcal{P}$ satisfies at least
    $(\frac{1}{2} +\epsilon)\cdot |\mathcal{T}|=\left(\frac{1}{2}+\epsilon\right)\cdot 2k=k+2\epsilon\cdot k$ pairs. Hence, it follows that~$|Y| \geq
    2\epsilon \cdot k$.
\end{proof}
\medskip

\begin{lemma}
    \label{lem:good-equal-edge}
    \normalfont
    If $i\in [k]$ is good, then there exists $\delta_i \in [N]$ such that the two paths in $\mathcal{P}$ satisfying $a_i\leadsto b_i$ and $c_i\leadsto d_i$ in \Gedge are images of the paths $\CanInter(\delta_i\ ;\ a_i \leadsto b_i)$ and $\CanInter(\delta_i\ ;\ c_i \leadsto d_i)$ from \Gint respectively.
\end{lemma}
\begin{proof}
    If $i$ is good, then by~\autoref{def:good-for-edge} both the pairs $a_i \leadsto b_i$ and $c_i \leadsto d_i$ are satisfied by $\mathcal{P}$. Let $P_1, P_2\in \mathcal{P}$ be the paths that satisfy the terminal pairs $(a_i, b_i)$ and $(c_i, d_i)$ respectively.
    Since $P_1$ is a shortest $a_i \leadsto b_i$ path in \Gedge, by~\autoref{lem:vertical-canonical-is-shortest-G-edge}(ii) it follows that $P_1$ is an image of the vertical canonical path $\CanInter(\alpha\ ;\ a_i \leadsto b_i)$ from \Gint for some $\alpha\in [N]$. Since $P_2$ is a shortest $c_i \leadsto d_i$ path in \Gedge, by~\autoref{lem:horizontal-canonical-is-shortest-G-edge}(ii) it follows that $P_2$ is an image of the horizontal canonical path $\CanInter(\beta\ ;\ c_i \leadsto d_i)$ from \Gint for some $\beta\in [N]$.

    Using the fact that $P_1$ and $P_2$ are edge-disjoint in \Gedge, we now claim that $\w_{i,i}^{\alpha,\beta}$ is \twosplit:
    \begin{claim}
        \normalfont
        The vertex $\w_{i,i}^{\alpha,\beta}$ is \twosplit by the splitting operation of~\autoref{def:splitting-operation-edge}.
        \label{clm:must-be-two-split-edge-i-i}
    \end{claim}
    \begin{proof}
        By~\autoref{def:splitting-operation-edge}, every black vertex of \Gint is either \onesplit or \twosplit. If $\w_{i,i}^{\alpha,\beta}$ was \onesplit (\autoref{fig:split-edge-not}), then by~\autoref{def:hor-image-of-a-path-G-edge} and \autoref{def:vert-image-of-a-path-G-edge} the path $\w_{i,i,\LB}^{\alpha,\beta}\to w_{i,i,\Mid}^{\alpha,\beta}\to w_{i,i,\TR}^{\alpha,\beta}$ belongs to both the paths $P_1$ and $P_2$ contradicting the fact that they are edge-disjoint.
    \end{proof}
    By~\autoref{clm:must-be-two-split-edge-i-i}, we know that the vertex $\w_{i,i}^{\alpha,\beta}$ is \twosplit. Hence, from~\autoref{eqn:clique-to-gt-reduction-D} and~\autoref{def:splitting-operation-edge}, it follows that $\alpha=\beta$ which concludes the proof of the lemma.
\end{proof}
\medskip

\begin{lemma}
    \label{lem:good_edges}
    \normalfont
    If both $i,j \in [k]$ are good and $i \neq j$, then $v_{\delta_i}-v_{\delta_j} \in E(G)$.
\end{lemma}
\begin{proof}
    Since $i$ and $j$ are good, by~\autoref{def:good-for-edge}, there are paths $Q_1, Q_2 \in \mathcal{P}$ satisfying the pairs $(a_i, b_i), (c_j, d_j)$ respectively. By~\autoref{lem:good-equal-edge}, it follows that
    \begin{itemize}
        \item $Q_1$ is an image of the path $\CanInter(\delta_i\ ;\ a_i\leadsto b_i)$ from \Gint.
        \item $Q_2$ is an image of the path $\CanInter(\delta_j\ ;\ c_j\leadsto d_j)$ from \Gint.
    \end{itemize}

    Using the fact that $Q_1$ and $Q_2$ are edge-disjoint in \Gedge, we now claim that $\w_{i,j}^{\delta_i,\delta_j}$ is \twosplit:
    \begin{claim}
        \normalfont
        The vertex $\w_{i,j}^{\delta_i,\delta_j}$ is \twosplit by the splitting operation of~\autoref{def:splitting-operation-edge}.
        \label{clm:must-be-two-split-edge-i-j}
    \end{claim}
    \begin{proof}
        By~\autoref{def:splitting-operation-edge}, every black vertex of \Gint is either \onesplit or \twosplit. If $\w_{i,j}^{\delta_j,\delta_j}$ was \onesplit (\autoref{fig:split-edge-not}), then by~\autoref{def:hor-image-of-a-path-G-edge} and \autoref{def:vert-image-of-a-path-G-edge} the path $\w_{i,j,\LB}^{\delta_i,\delta_j}\to w_{i,j,\Mid}^{\delta_i,\delta_j}\to w_{i,j,\TR}^{\delta_i,\delta_j}$ belongs to both the paths $Q_1$ and $Q_2$ contradicting the fact that they are edge-disjoint.
    \end{proof}
    By~\autoref{clm:must-be-two-split-edge-i-j}, we know that the vertex $\w_{i,j}^{\delta_i,\delta_j}$ is \twosplit. Since $i\neq j$, from~\autoref{eqn:clique-to-gt-reduction-D} and~\autoref{def:splitting-operation-edge}, it follows that $v_{\delta_i}-v_{\delta_j}\in E(G)$ which concludes the proof of the lemma.
\end{proof}
\medskip

From~\autoref{lem:good_size} and~\autoref{lem:good_edges}, it follows that the set $X:=\{v_{\delta_i}\ : i\in Y\}$ is a clique of size $\geq (2\epsilon)k$ in $G$.

\subsection{Proof of \autoref{thm:inapprox-edge-result} and \autoref{thm:hardness-edge-result}}
\label{sec:proof-of-main-theorem-edge}

Finally we are ready to prove \autoref{thm:inapprox-edge-result} and \autoref{thm:hardness-edge-result}, which are restated below.
\medskip

\apxdirectededgethm*
\exactdirectededgethm*
\begin{proof}{\textbf{\autoref{thm:hardness-edge-result}}}

    Given an instance $G$ of \kclique, we can use the construction from \autoref{sec:construction-of-Gedge} to build an instance $(D_{\edge}, \mathcal{T})$ of Directed-$2k$-\EDSP such that \Gedge is a planar DAG (\autoref{clm:Gedge-is-planar-and-dag}). The graph \Gedge has $n=O(N^2 k^2)$ vertices (\autoref{clm:size-of-G-edge}), and it is easy to observe that it can be constructed from $G$ (via first constructing \Gint) in $\poly(N,k)$ time.

    It is known that \kclique is W[1]-hard parameterized by $k$, and under ETH cannot be solved in $f(k)\cdot N^{o(k)}$ time for any computable function $f$~\cite{chen-hardness}. Combining the two directions from~\autoref{sec:2kedsp-to-clique} (with $\epsilon = 0.5$) and~\autoref{sec:clique-to-2kedsp}  we obtain a parameterized reduction from an instance $(G,k)$ of \kclique with $N$ vertices to an instance $(D_{\edge}, \mathcal{T})$ of Directed-$2k$-\EDSP where \Gedge is a planar DAG (\autoref{clm:Gedge-is-planar-and-dag}) and has $O(N^2 k^2)$ vertices (\autoref{clm:size-of-G-edge}). As a result, it follows that Directed-$k$-\EDSP on planar DAGs is W[1]-hard parameterized by number $k$ of terminal pairs, and under ETH cannot be solved in $f(k)\cdot n^{o(k)}$ time where $f$ is any computable function and $n$ is the number of vertices.
\end{proof}
\medskip

\begin{proof}{\textbf{\autoref{thm:inapprox-edge-result}}}

    Let $\delta$ and $r_0$ be the constants from~\autoref{thm:cli_inapprox}. Fix any constant $\epsilon\in (0,1/2]$. Set $\zeta = \dfrac{\delta \epsilon}{2}$ and $k=\max \Big\{\dfrac{1}{2\zeta} , \dfrac{r_0}{2\epsilon}\Big\}$.

    Suppose to the contrary that there exists an algorithm $\mathbb{A}_{\EDSP}$ running in $f(k)\cdot n^{\zeta k}$ time (for some computable function $f$) which given an instance of Directed-$k$-\EDSP with $n$ vertices can distinguish between the following two cases:
    \begin{itemize}
        \item[(1)] All $k$ pairs of the Directed-$k$-\EDSP instance can be satisfied
        \item[(2)] The max number of pairs of the Directed-$k$-\EDSP instance that can be satisfied is less than $(\frac{1}{2}+\epsilon)\cdot k$
    \end{itemize}
    We now design an algorithm $\mathbb{A}_{\Clique}$ that contradicts~\autoref{thm:cli_inapprox} for the values $q=k$ and $r=(2\epsilon)k$. Given an instance of $(G,k)$ of \kclique with $N$ vertices, we apply the reduction from~\autoref{sec:construction-of-Gedge} to construct an instance $(D_{\edge}, \mathcal{T})$ of Directed-$2k$-\EDSP where \Gedge has $n=O(N^2 k^2)$ vertices (\autoref{clm:size-of-G-edge}). It is easy to see that this reduction takes $O(N^2 k^2)$ time as well. We now show that the number of pairs which can be satisfied from the Directed-$2k$-\EDSP instance is related to the size of the max clique in $G$:
    \begin{itemize}
        \item If $G$ has a clique of size $q=k$, then by~\autoref{sec:clique-to-2kedsp} it follows that all $2k$ pairs of the instance $(D_{\edge},\mathcal{T})$ of Directed-$2k$-\EDSP can be satisfied.
        \item If $G$ does not have a clique of size $r=2\epsilon k$, then we claim that the max number of pairs in $\mathcal{T}$ that can be satisfied is less than $(\frac{1}{2}+\epsilon)\cdot 2k$. This is because if at least $(\frac{1}{2}+\epsilon)$-fraction of pairs in $\mathcal{T}$ could be satisfied then by~\autoref{sec:2kedsp-to-clique} the graph $G$ would have a clique of size $\geq (2\epsilon) k=r$.
    \end{itemize}
    Since the algorithm $\mathbb{A}_{\EDSP}$ can distinguish between the two cases of all $2k$-pairs of the instance $(D_{\edge}, \mathcal{T})$ can be satisfied or only less than $(\frac{1}{2}+\epsilon)\cdot2k$ pairs can be satisfied, it follows that $\mathbb{A}_{\Clique}$ can distinguish between the cases $\Clique(G)\geq q$ and $\Clique(G)<r$.

    The running time of the algorithm $\mathbb{A}_{\Clique}$ is the time taken for the reduction from~\autoref{sec:construction-of-Gedge} (which is $O(N^2 k^2)$) plus the running time of the algorithm $\mathbb{A}_{\EDSP}$ which is $f(2k)\cdot n^{\zeta\cdot 2k}$. It remains to show that this can be upper bounded by $g(q,r)\cdot N^{\delta r}$ for some computable function $g$:
    \begin{align*}
        & O(N^2 k^2) + f(2k)\cdot n^{\zeta\cdot 2k}                                                                                                                  \\
        & \leq  c\cdot N^2 k^2 + f(2k)\cdot d^{\zeta\cdot 2k}\cdot (N^2 k^2)^{\zeta\cdot 2k} \tag{for some constants $c,d\geq 1$: this follows since $n=O(N^2 k^2)$} \\
        & \leq c\cdot N^2 k^2 + f'(k)\cdot N^{2\zeta\cdot 2k} \tag{where $f'(k)=f(2k)\cdot d^{\zeta\cdot 2k}\cdot k^{2\zeta\cdot 2k}$}                               \\
        & \leq 2c\cdot f'(k)\cdot N^{2\zeta\cdot 2k} \tag{since $4\zeta k\geq 2$ implies $f'(k)\geq k^2$ and $N^{2\zeta\cdot 2k}\geq N^2$}                           \\
        & = 2c\cdot f'(k) \cdot N^{\delta r} \tag{since $\zeta=\frac{\delta \epsilon}{2}$ and $r=(2\epsilon)k$}
    \end{align*}
    Hence, we obtain a contradiction to~\autoref{thm:cli_inapprox} with $q=k, r=(2\epsilon)k$ and $g(k)=2c\cdot f'(k)=2c\cdot f(2k)\cdot d^{\zeta\cdot 2k}\cdot k^{2\zeta\cdot 2k}$.
\end{proof}
\medskip

\begin{remark}
    \label{rmk:Max-Degree-2-edge}
    \normalfont
    \textbf{(reducing the in-degree and out-degree of \Gedge)}
    By exactly the same process as described in~\autoref{rmk:hori-verti-sets-G}, we can reduce the max in-degree and max out-degree of \Gedge to be at most two whilst maintaining the properties that $n=|V(D_{\edge})|=O(N^{2}k^{2})$ and that \Gedge can be constructed in $\poly(N,k)$ time. The splitting operation (\autoref{def:splitting-operation-edge}) is applied only to black vertices, hence all the proofs from~\autoref{sec:characterizing-shortest-in-G-edge},~\autoref{sec:clique-to-2kedsp} and~\autoref{sec:2kedsp-to-clique} go through with minor modifications.
\end{remark}

\medskip

	\section{Lower bounds for exact \& approximate Directed-\texorpdfstring{$k$}{k}-\VDSP on \texorpdfstring{$1$}{1}-planar DAGs}
\label{sec:fpt-inapprox-vdsp-1-planar-dags}

The goal of this section is to prove lower bounds on the running time of exact (\autoref{thm:hardness-vertex-result}) and approximate (\autoref{thm:inapprox-vertex-result}) algorithms for the Directed-$k$-\VDSP problem. We have already seen the first part of the reduction (\autoref{sec:construction-of-Gint}) from \kclique resulting in the construction of the intermediate graph \Gint.~\autoref{sec:construction-of-Gvertex} describes the next part of the reduction which edits the intermediate \Gint to obtain the final graph \Gvert. This corresponds to the ancestry of the second leaf in~\autoref{fig:flowchart}. The characterization of shortest paths between terminal pairs in \Gvert is given in~\autoref{sec:characterizing-shortest-in-G-vertex}. The completeness and soundness of the reduction from \kclique to Directed-$2k$-\VDSP are proven in~\autoref{sec:clique-to-2kvdsp} and~\autoref{sec:2kvdsp-to-clique} respectively. Finally, everything is tied together in~\autoref{sec:proof-of-main-theorem-vertex} allowing us to prove~\autoref{thm:hardness-vertex-result} and~\autoref{thm:inapprox-vertex-result}.

\subsection{Obtaining the graph \Gvert from \Gint via the splitting operation}
\label{sec:construction-of-Gvertex}

Recall from~\autoref{fig:main} that every black grid vertex in \Gint has in-degree two and out-degree two. These four neighbors are named as per~\autoref{def:lefty-right-topy-bottomy-dir}.
The construction of \Gvert from \Gint differs from the construction of \Gedge from \autoref{sec:construction-of-Gedge} only in its splitting operation. This new splitting operation (analogous to \autoref{def:splitting-operation-edge}) is defined below:

\medskip

\begin{definition}
    \normalfont
    \textbf{(splitting operation to obtain \Gvert from \Gint)} For each $i,j\in [k]$ and each $q,\ell\in [N]$
    \begin{itemize}
        \item If $(q,\ell)\in S_{i,j}$ then we \vertsplit (\autoref{fig:split-vertex-yes}) the vertex $\w_{i,j}^{q,\ell}$ into \textbf{two distinct} vertices $\w_{i,j,\Hor}^{q,\ell}$, and $\w_{i,j,\Ver}^{q,\ell}$.
        \item Otherwise, if $(q,\ell)\notin S_{i,j}$, then the vertex $\w_{i,j}^{q,\ell}$ is \notsplit (\autoref{fig:split-vertex-not}) and we define $\w_{i,j,\Hor}^{q,\ell} = \w_{i,j,\Ver}^{q,\ell}$.
    \end{itemize}
    In both the cases, the $4$ edges (\autoref{def:lefty-right-topy-bottomy-dir}) incident on $\w_{i,j}^{q,\ell}$ are modified as follows:
    \begin{itemize}
        \item Replace the edge $\lefty(\w_{i,j}^{q,\ell})\to \w_{i,j}^{q,\ell}$ by the edge $\lefty(\w_{i,j}^{q,\ell})\to \w_{i,j,\Hor}^{q,\ell}$
        \item Replace the edge $\bottomy(\w_{i,j}^{q,\ell})\to \w_{i,j}^{q,\ell}$ by the edge $\bottomy(\w_{i,j}^{q,\ell})\to \w_{i,j,\Ver}^{q,\ell}$
        \item Replace the edge $\w_{i,j}^{q,\ell}\to \righty(\w_{i,j}^{q,\ell})$ by the edge $\w_{i,j,\Hor}^{q,\ell}\to \righty(\w_{i,j}^{q,\ell})$
        \item Replace the edge $\w_{i,j}^{q,\ell}\to \topy(\w_{i,j}^{q,\ell})$ by the edge $\w_{i,j,\Ver}^{q,\ell}\to \topy(\w_{i,j}^{q,\ell})$
    \end{itemize}
    \label{def:splitting-operation-vertex}
\end{definition}
\medskip

\begin{figure}[hbt!]
\centering
\begin{tikzpicture}[
vertex/.style={circle, draw=black, fill=black, text width=1.5mm, inner sep=0pt},
scale=0.8]
\node[vertex, label=above right:\footnotesize{$\w_{i,j}^{q,\ell}$}] (v) at (0,0) {} ;
\node[vertex, label=above:\footnotesize{$\lefty(\w_{i,j}^{q,\ell})$}] (l) at (-2,0) {};
\node[vertex, label=below:\footnotesize{$\righty(\w_{i,j}^{q,\ell})$}] (r) at (2,0) {};
\node[vertex, label=below:\footnotesize{$\bottomy(\w_{i,j}^{q,\ell})$}] (b) at (0,-2) {};
\node[vertex, label=above:\footnotesize{$\topy(\w_{i,j}^{q,\ell})$}] (t) at (0,2) {};
\draw[ultra thick, middlearrow={>}] (v) -- (t);
\draw[ultra thick, middlearrow={>}] (v) -- (r);
\draw[ultra thick, middlearrow={>}] (l) -- (v);
\draw[ultra thick, middlearrow={>}] (b) -- (v);

\draw[orange,double, ultra thick,->] (3,0) -- node[above=3mm, draw=none, fill=none, rectangle] {\vertsplit} (5,0);

\node[vertex, label=above:\footnotesize{$\w_{i,j,\Hor}^{q,\ell}$}] (vh) at (7.5,0) {} ;
\node[vertex, label=right:\footnotesize{$\w_{i,j,\Ver}^{q,\ell}$}] (vv) at (9.5,1) {} ;

\node[vertex, label=below:\footnotesize{$\lefty(\w_{i,j}^{q,\ell})$}] (l) at (6,0) {};
\node[vertex, label=below:\footnotesize{$\righty(\w_{i,j}^{q,\ell})$}] (r) at (10,0) {};
\node[vertex, label=below:\footnotesize{$\bottomy(\w_{i,j}^{q,\ell})$}] (b) at (8,-2) {};
\node[vertex, label=above:\footnotesize{$\topy(\w_{i,j}^{q,\ell})$}] (t) at (8,2) {};
\draw[dotted, ultra thick, middlearrow={>}] (vv) -- (t);
\draw[dotted, ultra thick, middlearrow={>}] (vh) -- (r);
\draw[dotted, ultra thick, middlearrow={>}] (l) -- (vh);
\draw[dotted, ultra thick, middlearrow={>}] (b) -- (vv);
\end{tikzpicture}

\caption{The \vertsplit operation for the vertex $\w_{i,j}^{q,\ell}$ when
$(q,\ell) \in S_{i,j}$.  The idea behind this is that the horizontal path $\lefty(\w_{i,j}^{q,\ell})\to \w_{i,j}^{q,\ell} \to \righty(\w_{i,j}^{q,\ell})$ and the vertical path $\bottomy(\w_{i,j}^{q,\ell})\to \w_{i,j}^{q,\ell} \to \topy(\w_{i,j}^{q,\ell})$ are now actually vertex-disjoint after the \vertsplit operation (but were not vertex-disjoint before since they shared the vertex $\w_{i,j}^{q,\ell}$)}
\label{fig:split-vertex-yes}
\end{figure}
\begin{figure}[hbt!]
\centering
\begin{tikzpicture}[
vertex/.style={circle, draw=black, fill=black, text width=1.5mm, inner sep=0pt},
scale=0.85]
\node[vertex, label=above right:\footnotesize{$\w_{i,j}^{q,\ell}$}] (v) at (0,0) {} ;
\node[vertex, label=above:\footnotesize{$\lefty(\w_{i,j}^{q,\ell})$}] (l) at (-2,0) {};
\node[vertex, label=below:\footnotesize{$\righty(\w_{i,j}^{q,\ell})$}] (r) at (2,0) {};
\node[vertex, label=below:\footnotesize{$\bottomy(\w_{i,j}^{q,\ell})$}] (b) at (0,-2) {};
\node[vertex, label=above:\footnotesize{$\topy(\w_{i,j}^{q,\ell})$}] (t) at (0,2) {};
\draw[ultra thick, middlearrow={>}] (v) -- (t);
\draw[ultra thick, middlearrow={>}] (v) -- (r);
\draw[ultra thick, middlearrow={>}] (l) -- (v);
\draw[ultra thick, middlearrow={>}] (b) -- (v);

\draw[orange,double, ultra thick,->] (3,0) -- node[above=3mm, draw=none, fill=none, rectangle] {\notsplit} (5,0);

\node[vertex, label={[rotate=35,xshift=14mm,yshift=-4mm] \footnotesize{$\w_{i,j,\Hor}^{q,\ell}=\w_{i,j,\Ver}^{q,\ell}$}}] (mid) at (8,0) {} ;

{(0,0)}

\node[vertex, label=below:\footnotesize{$\lefty(\w_{i,j}^{q,\ell})$}] (l) at (6,0) {};
\node[vertex, label=below:\footnotesize{\quad $\righty(\w_{i,j}^{q,\ell})$}] (r) at (10,0) {};
\node[vertex, label=below:\footnotesize{$\bottomy(\w_{i,j}^{q,\ell})$}] (b) at (8,-2) {};
\node[vertex, label=above:\footnotesize{$\topy(\w_{i,j}^{q,\ell})$}] (t) at (8,2) {};
\draw[dotted, ultra thick, middlearrow={>}] (mid) -- (t);
\draw[dotted, ultra thick, middlearrow={>}] (mid) -- (r);
\draw[dotted, ultra thick, middlearrow={>}] (l) -- (mid);
\draw[dotted, ultra thick, middlearrow={>}] (b) -- (mid);
\end{tikzpicture}

\caption{The \notsplit operation for the vertex $\w_{i,j}^{q,\ell}$ when
$(q,\ell) \notin S_{i,j}$. The idea behind this is that the horizontal path $\lefty(\w_{i,j}^{q,\ell})\to \w_{i,j}^{q,\ell} \to \righty(\w_{i,j}^{q,\ell})$ and the vertical path $\bottomy(\w_{i,j}^{q,\ell})\to \w_{i,j}^{q,\ell} \to \topy(\w_{i,j}^{q,\ell})$ are still not vertex-disjoint after the \notsplit operation since they share the vertex $\w_{i,j,\Hor}^{q,\ell}=\w_{i,j,\Ver}^{q,\ell}$.}
\label{fig:split-vertex-not}
\end{figure}

Finally, we are now ready to define the instance of Directed-$2k$-\VDSP that we have built starting from an instance $G$ of \kclique.
\medskip

\begin{definition}
    \normalfont
    \textbf{(defining the Directed-$2k$-\VDSP instance)} The instance $(D_{\vertex}, \mathcal{T})$ of Directed-$2k$-\VDSP is defined as follows:
    \begin{itemize}
        \item The graph \Gvert is obtained by applying the splitting operation (\autoref{def:splitting-operation-vertex}) to each (black) grid vertex of \Gint, i.e., the set of vertices given by $\bigcup_{1\leq i,j\leq k} V(D_{i,j})$.
        \item No \green{green} vertex is split in~\autoref{def:splitting-operation-vertex}, and hence the set of terminal pairs remains the same as defined in~\autoref{eqn:definition-of-mathcal-T} and is given by $\mathcal{T}:= \big\{(a_i, b_i) : i\in [k] \big\}\cup \big\{(c_j, d_j) : j\in [k] \big\}$.
        \item We assign a cost of one to visit each of the vertices present after the splitting operation (\autoref{def:splitting-operation-vertex}). Since each vertex in \Gint has a cost of one, it follows that each vertex in \Gvert also has a cost of one.
    \end{itemize}
    \label{def:G-vertex}
\end{definition}

\medskip

The next two sec:s analyze the structure and size of the graph \Gvert.
\medskip

\begin{claim}
    \label{clm:Gvertex-is-1-planar-and-dag}
    \normalfont
    \Gvert is a $1$-planar DAG\footnote{A $1$-planar graph is a graph that can be drawn in the Euclidean plane in such a way that each edge has at most one crossing point, where it crosses a single additional edge.}.
\end{claim}
\begin{proof}
    In~\autoref{clm:G1-is-planar-and-dag}, we have shown that \Gint is a planar DAG. The graph \Gvert is obtained from \Gint by applying the splitting operation (\autoref{def:splitting-operation-vertex}) on every (black) grid vertex, i.e., every vertex from the set $\bigcup_{1\leq i,j\leq k} V(D_{i,j})$. By~\autoref{def:lefty-right-topy-bottomy-dir}, every vertex of \Gint that is split has exactly two in-neighbors and two out-neighbors in \Gint.~\autoref{fig:split-vertex-not} maintains the planarity, but in~\autoref{fig:split-vertex-yes} we have two edges $\bottomy(\w_{i,j}^{q,\ell}) \to \w_{i,j,\Ver}^{q,\ell}$ and $\w_{i,j,\Hor}^{q,\ell} \to \righty(\w_{i,j}^{q,\ell})$ which cross each other: this seems unavoidable while preserving the global structure of the graph. Since these are the only type of edges which can cross, we have drawn \Gvert in the Euclidean plane in such a way that each edge has at most one crossing point, where it crosses a single additional edge. Therefore, \Gvert is $1$-planar.

    Since \Gint is a DAG, it has a topological order say $\mathcal{X}$. The only changes done when going from \Gint to \Gvert are the addition of new vertices and edges when black grid vertices are split according to~\autoref{def:splitting-operation-vertex}. We now explain how to modify $\mathcal{X}$ to obtain a topological order $\mathcal{X}'$ for \Gvert:
    \begin{itemize}
        \item If a black grid vertex $\w$ is \vertsplit, then we replace $\w$ by the two vertices $\w_{\Hor}$ and $\w_{\Ver}$.
        \item If a black grid vertex $\w$ is \notsplit, then we replace $\w$ by the vertex $\w_{\Hor}=\w_{\Ver}$.
    \end{itemize}
    It is easy to see from~\autoref{fig:split-vertex-not} and~\autoref{fig:split-vertex-yes} that $\mathcal{X}'$ is a topological order for \Gvert.
\end{proof}
\medskip

\begin{claim}
    \normalfont
    The number of vertices in \Gvert is $O(N^{2}k^{2})$.
    \label{clm:size-of-G-vertex}
\end{claim}
\begin{proof}
    The only change in going from \Gint to \Gvert is the splitting operation (\autoref{def:splitting-operation-vertex}). If a black grid vertex $\w$ in \Gint is \notsplit (\autoref{fig:split-vertex-not}) then we replace it by \textbf{one} vertex $\w_{\Ver}=\w_{\Hor}$ in \Gvert. If a black grid vertex $\w$ in \Gint is \vertsplit (\autoref{fig:split-vertex-yes}) then we replace it by the \textbf{two} vertices $\w_{\Hor}$ and $\w_{\Ver}$ in \Gvert. In both cases, the increase in number of vertices is only by a constant factor. The number of vertices in \Gint is $O(N^2 k^2)$ from~\autoref{clm:size-of-G-int}, and hence it follows that the number of vertices in \Gvert is $O(N^2 k^2)$.
\end{proof}
\medskip

\subsection{Characterizing shortest paths in \Gvert}
\label{sec:characterizing-shortest-in-G-vertex}

The goal of this section is to characterize the structure of shortest paths between terminal pairs in \Gvert. Recall (\autoref{def:G-vertex}) that the set of terminal pairs is given by $\mathcal{T}:= \big\{(a_i, b_i) : i\in [k] \big\}\cup \big\{(c_j, d_j) : j\in [k] \big\}$. Since each edge of \Gvert has length one (\autoref{def:G-vertex}), we measure the length of paths in \Gedge by counting the number of vertices.

We now define canonical paths in \Gvert by adapting the definition of canonical paths (\autoref{def:hori-canonical-G1} and~\autoref{def:vert-canonical-G1}) in \Gint in accordance with the changes in going from \Gint to \Gvert.

\medskip

\begin{definition}
    \textbf{(horizontal canonical paths in \Gvert)}
    Fix any $j\in [k]$. For each $r\in [N]$, we define $\CanVertex(r\ ;\ c_j \leadsto d_j)$ to be the $c_j\leadsto d_j$ path in \Gvert obtained from the path $\CanInter(r\ ;\ c_j \leadsto d_j)$ in \Gint (recall~\autoref{def:hori-canonical-G1}) in the following way:
    \begin{itemize}
        \item The first and last \magenta{magenta} edges are unchanged
        \item If a black grid vertex $\w$ from $\CanInter(r\ ;\ c_j \leadsto d_j)$ is \notsplit (\autoref{fig:split-vertex-not}), then
        \begin{itemize}
            \item The unique incoming edge into $\w$ is changed to be incoming into $\w_{\Hor}=\w_{\Ver}$
            \item The unique outgoing edge from $\w$ is changed to be outgoing from $\w_{\Hor}=\w_{\Ver}$
        \end{itemize}
        \item If a black grid vertex $\w$ from $\CanInter(r\ ;\ c_j \leadsto d_j)$ is \vertsplit (\autoref{fig:split-vertex-yes}), then
        \begin{itemize}
            \item The unique incoming edge into $\w$ is changed to be incoming into $\w_{\Hor}$
            \item The unique outgoing edge from $\w$ is changed to be outgoing from $\w_{\Hor}$
        \end{itemize}
    \end{itemize}
    \label{def:hori-canonical-Gvertex}
\end{definition}
\medskip

\begin{definition}
    \textbf{(vertical canonical paths in \Gvert)}
    Fix any $j\in [k]$. For each $r\in [N]$, we define $\CanVertex(r\ ;\ a_j \leadsto b_j)$ to be the $a_j\leadsto b_j$ path in \Gvert obtained from the path $\CanInter(r\ ;\ a_j \leadsto b_j)$ in \Gint (recall~\autoref{def:vert-canonical-G1}) in the following way:
    \begin{itemize}
        \item The first and last \magenta{magenta} edges are unchanged
        \item If a black grid vertex $\w$ from $\CanInter(r\ ;\ a_j \leadsto b_j)$ is \notsplit (\autoref{fig:split-vertex-not}), then
        \begin{itemize}
            \item The unique incoming edge into $\w$ is changed to be incoming into $\w_{\Hor}=\w_{\Ver}$
            \item The unique outgoing edge from $\w$ is changed to be outgoing from $\w_{\Hor}=\w_{\Ver}$
        \end{itemize}
        \item If a black grid vertex $\w$ from $\CanInter(r\ ;\ a_j \leadsto b_j)$ is \vertsplit (\autoref{fig:split-vertex-yes}), then
        \begin{itemize}
            \item The unique incoming edge into $\w$ is changed to be incoming into $\w_{\Ver}$
            \item The unique outgoing edge from $\w$ is changed to be outgoing from $\w_{\Ver}$
        \end{itemize}
    \end{itemize}
    \label{def:vert-canonical-Gvertex}
\end{definition}
\medskip

The next lemma shows that if $j\in [k]$ then any shortest $c_j \leadsto d_j$ path in \Gvert must be a horizontal canonical path and vice versa.

\medskip

\begin{definition}
    \textbf{(Image of a horizontal canonical path from \Gint in \Gvert)}
    Fix a $j\in [k]$ and $r\in [N]$. For each $\CanInter(r\ ;\ c_j \leadsto d_j)$ path $R$ in \Gint, we define an image of R as follows
    \begin{itemize}
        \item The first and last \magenta{magenta} edges are unchanged.
        \item If a black grid vertex $\w$ from $\CanInter(r\ ;\ c_j \leadsto d_j)$ is \notsplit (\autoref{fig:split-vertex-not}), then
        \begin{itemize}
            \item The unique edge $\lefty(\w) \to \w$ is replaced with the edge $\lefty(\w) \to \w_{\Hor}=\w_{\Ver}$;
            \item The unique edge $\w \to \righty(\w)$ is replaced with the edge $\w_{\Hor}=\w_{\Ver} \to \righty(\w)$;
        \end{itemize}
        \item If a black grid vertex $\w$ from $\CanInter(r\ ;\ c_j \leadsto d_j)$ is \vertsplit (\autoref{fig:split-vertex-yes}), then
        \begin{itemize}
            \item The series of edges $\lefty(\w) \to \w \to \righty(\w)$ is replaced with either the path  $\lefty(\w) \to \w_{\Ver} \to \righty(\w)$ or $\lefty(\w) \to \w_{\Hor} \to \righty(\w)$;
        \end{itemize}
    \end{itemize}
    \label{def:hor-image-of-a-path-G-vert}
\end{definition}
\medskip

\begin{definition}
    \textbf{(Image of a vertical canonical path from \Gint in \Gvert)}
    Fix a $i\in [k]$ and $r\in [N]$. For each $\CanInter(r\ ;\ a_i \leadsto b_i)$ path $R$ in \Gint, we define an image of R as follows
    \begin{itemize}
        \item The first and last \magenta{magenta} edges are unchanged.
        \item If a black grid vertex $\w$ from $\CanInter(r\ ;\ a_i \leadsto b_i)$ is \notsplit (\autoref{fig:split-vertex-not}), then
        \begin{itemize}
            \item The unique edge $\topy(\w) \to \w$ is replaced with the edge $\topy(\w) \to \w_{\Hor}=\w_{\Ver}$;
            \item The unique edge $\w \to \bottomy(\w)$ is replaced with the edge $\w_{\Hor}=\w_{\Ver} \to \bottomy(\w)$;
        \end{itemize}
        \item If a black grid vertex $\w$ from $\CanInter(r\ ;\ a_i \leadsto b_i)$ is \vertsplit (\autoref{fig:split-vertex-yes}), then
        \begin{itemize}
            \item The series of edges $\topy(\w) \to \w \to \bottomy(\w)$ is replaced with either the path $\topy(\w) \to \w_{\Ver} \to \bottomy(\w)$ or $\topy(\w) \to \w_{\Hor} \to \bottomy(\w)$;
        \end{itemize}
    \end{itemize}
    \label{def:vert-image-of-a-path-G-vert}
\end{definition}
\medskip

Note that a single path, $R$, in \Gint can have several images in \Gvert. This is because for every black vertex on $R$ that is \twosplit there are two choices of sub-path to add: either the path $\w_{\LB} \to \w_{\Hor} \to \w_{\TR}$ or the path $\w_{\LB} \to \w_{\Ver} \to \w_{\TR}$.

The following two lemmas (\autoref{lem:horizontal-canonical-is-shortest-G-vertex} and~\autoref{lem:vertical-canonical-is-shortest-G-vertex}) analyze the structure of shortest paths between terminal pairs in \Gvert. First, we define the image of a path from \Gint in the graph \Gvert.
\medskip

\begin{lemma}
    \normalfont
    Let $j\in [k]$. The horizontal canonical paths in \Gvert satisfy the following two properties:
    \begin{itemize}
        \item[(i)] For each $r\in [N]$, the path $\CanVertex(r\ ;\ c_j \leadsto d_j)$ is a shortest $c_j \leadsto d_j$ path in \Gvert.
        \item[(ii)] If $P$ is a shortest $c_j\leadsto d_j$ path in \Gvert, then $P$ must be an image (\autoref{def:hor-image-of-a-path-G-vert}) of $\CanVertex(\ell\ ;\ c_j \leadsto d_j)$ for some $\ell\in [N]$.
    \end{itemize}
    \label{lem:horizontal-canonical-is-shortest-G-vertex}
\end{lemma}
\begin{proof}
    The proof of this lemma can be obtained in the same way as shown for \Gint in \autoref{lem:horizontal-canonical-is-shortest-G1} with some minor observational changes. Note that any path in \Gint contains only \green{green} and black vertices. The splitting operation (\autoref{def:splitting-operation-vertex}) applied to each black vertex of \Gint has the following property: if a path $Q$ contains a black vertex $\w$ in \Gint, then in the corresponding path in \Gvert this vertex $\w$ is \textbf{always replaced by one other vertex}:
    \begin{itemize}
        \item If $\w$ is \notsplit (\autoref{fig:split-vertex-not}), then it is replaced in $Q$ the vertex $\w_{\Hor} = \w_{\Ver}$.
        \item If $\w$ is \vertsplit (\autoref{fig:split-vertex-yes}), then it is replaced in $Q$ either by the vertex $\w_{\Ver}$ or the vertex $\w_{\Hor}$.
    \end{itemize}
    Therefore, if a path $Q$ contains $\alpha$ \green{green} vertices and $\beta$ black vertices in \Gint, then the corresponding path in \Gvert contains $\alpha$ \green{green} vertices and $\beta$ black vertices. The proof of the first part of the lemma now follows from~\autoref{lem:horizontal-canonical-is-shortest-G1}(i),~\autoref{def:splitting-operation-vertex} and~\autoref{def:hori-canonical-Gvertex}. The proof of the second part of the lemma follows from~\autoref{lem:horizontal-canonical-is-shortest-G1}(ii),~\autoref{def:splitting-operation-vertex} and~\autoref{def:hor-image-of-a-path-G-vert}.
\end{proof}
\medskip

The proof of the next lemma is very similar to that of \autoref{lem:horizontal-canonical-is-shortest-G-vertex}, and we skip repeating the details.

\medskip

\begin{lemma}
    \normalfont
    Let $i\in [k]$. The vertical canonical paths in \Gvert satisfy the following two properties:
    \begin{itemize}
        \item[(i)] For each $r\in [N]$, the path $\CanVertex(r\ ;\ a_i \leadsto b_i)$ is a shortest $a_i \leadsto b_i$ path in \Gvert.
        \item[(ii)] If $P$ is a shortest $a_i\leadsto b_i$ path in \Gvert, then $P$ must be and image (\autoref{def:vert-image-of-a-path-G-vert}) of $\CanVertex(\ell\ ;\ a_i \leadsto b_i)$ for some $\ell\in [N]$.
    \end{itemize}
    \label{lem:vertical-canonical-is-shortest-G-vertex}
\end{lemma}
\medskip

\subsection{
    \texorpdfstring{\underline{Completeness}: $G$ has a $k$-clique $\Rightarrow$ All pairs in the instance $(D_{\vertex}, \mathcal{T})$ of Directed-$2k$-VDSP can be satisfied}
    {Completeness: G has a k-clique -> All pairs in the instance (Dvert, T) of 2k-VDSP can be satisfied}
}
\label{sec:clique-to-2kvdsp}

In this section, we show that if the instance $G$ of \kclique has a solution then the instance $(D_{\vertex}, \mathcal{T})$ of Directed-$2k$-\VDSP also has a solution. The proofs are very similar to those of the corresponding results from~\autoref{sec:clique-to-2kedsp}.
Suppose the instance $G=(V,E)$ of \kclique has a clique $X=\{v_{\gamma_1}, v_{\gamma_2}, \ldots, v_{\gamma_k} \}$ of size $k$. Let $Y=\{\gamma_1, \gamma_2, \ldots, \gamma_k\}\in [N]$. Now for each $i\in [k]$ we choose the path as follows:
\begin{itemize}
    \item The path $R_i$ to satisfy $a_i\leadsto b_i$ is chosen to be the horizontal canonical path  $\CanVertex(\gamma_i \ ;\ a_i \leadsto b_i)$ described in~\autoref{def:hori-canonical-Gvertex}.
    \item The path $T_i$ to satisfy $c_i\leadsto d_i$ is chosen to be vertical canonical path  $\CanVertex(\gamma_i \ ;\ c_i \leadsto d_i)$ described in~\autoref{def:vert-canonical-Gvertex}.
\end{itemize}

Now we show that the collection of paths given by $\mathcal{Q}:=\{R_1, R_2, \ldots, R_k, T_1, T_2, \ldots, T_K\}$ forms a solution for the instance  $(D_{\edge}, \mathcal{T})$ of Directed-$2k$-\VDSP via the following two lemmas which argue being shortest for each terminal pair and pairwise vertex-disjointness respectively:
\medskip

\begin{lemma}
    \normalfont
    For each $i \in [k]$, the path $R_i$ (resp. $T_i$) is a shortest $a_i \leadsto b_i$ (resp. $c_i \leadsto d_i$) path in \Gvert.
    \label{lem:completeness-vdsp-shortest}
\end{lemma}
\begin{proof}
    Fix any $i\in [k]$.~\autoref{lem:horizontal-canonical-is-shortest-G-vertex}(i) implies that $T_i$ is shortest $c_i \leadsto d_i$ path in  \Gvert.~\autoref{lem:vertical-canonical-is-shortest-G-vertex}(i) implies that $R_i$ is shortest $a_i \leadsto b_i$ path in \Gvert.
\end{proof}
\medskip

Before proving~\autoref{lem:completeness-vdsp-disjoint}, we first set up notation for some special sets of vertices in \Gedge which helps to streamline some of the subsequent proofs.

\medskip

\begin{definition}
    \label{def:horizontal-vertical-sets-in-Gvertex}
    \textbf{(horizontal \& vertical levels in \Gvert)}
    For each $(i,j)\in [k]\times [k]$, let $D_{i,j}^{\vertsplitt}$ to be the graph obtained by applying the splitting operation (\autoref{def:splitting-operation-edge}) to each vertex of $D_{i,j}$. For each $j\in [k]$, we define the following set of vertices:
    \begin{equation}
        \begin{aligned}
            \HorizontalVertex(j)     & = \{ c_j, d_j \} \cup \left( \bigcup_{i=1}^{k} V(D_{i,j}^{\vertsplitt})\right)\ \quad \\
            \quad \VerticalVertex(j) & = \{ a_j, b_j \} \cup \left( \bigcup_{i=1}^{k} V(D_{j,i}^{\vertsplitt})\right)
        \end{aligned}\label{eqn:hori-verti-sets-Gvert}
    \end{equation}
\end{definition}
\medskip

The next lemma shows that any two paths from $\mathcal{Q}$ are vertex-disjoint.
\medskip

\begin{lemma}
    \normalfont
    Let $P\neq P'$ be any pair of paths from the collection  $\mathcal{Q}=\{R_1, R_2, \ldots, R_k, T_1, T_2, \ldots, T_K\}$. Then $P$ and $P'$ are vertex-disjoint.
    \label{lem:completeness-vdsp-disjoint}
\end{lemma}
\begin{proof}

    By~\autoref{def:horizontal-vertical-sets-in-Gvertex}, it follows that every edge of the path $R_i$ has both endpoints in $\VerticalVertex(i)$ for every $i\in [k]$. Since $\VerticalVertex(i) \cap \VerticalVertex(i')=\emptyset$ for every $1\leq i\neq i'\leq k$, it follows that the collection of paths $\{R_1, R_2, \ldots, R_k\}$ are pairwise vertex-disjoint.

    By~\autoref{def:horizontal-vertical-sets-in-Gvertex}, it follows that every edge of the path $T_j$ has both endpoints in $\HorizontalVertex(j)$ for every $j\in [k]$. Since $\HorizontalVertex(j) \cap \HorizontalVertex(j')=\emptyset$ for every $1\leq j\neq j'\leq k$, it follows that the collection of paths $\{T_1, T_2, \ldots, T_k\}$ are pairwise vertex-disjoint.

    It remains to show that every pair of paths which contains one path from $\{R_1, R_2, \ldots, R_k\}$ and other path from $\{T_1, T_2, \ldots, T_k\}$ are vertex-disjoint.
    \begin{claim}
        \label{clm:reduction-vertex-disjoint}
        \normalfont
        For each $(i,j)\in [k]\times [k]$, the paths $R_i$ and $T_j$ are vertex-disjoint in \Gvert.
    \end{claim}

    \begin{proof}
        Fix any $(i,j)\in [k]\times [k]$. First we argue that the vertex $\w_{i,j}^{\gamma_i, \gamma_j}$ is \vertsplit, i.e., $(\gamma_i, \gamma_j)\in S_{i,j}$:
        \begin{itemize}
            \item If $i=j$ then $\gamma_i = \gamma_j$ and hence by~\autoref{eqn:clique-to-gt-reduction-D} we have $(\gamma_i, \gamma_j)\in S_{i,j}$
            \item If $i\neq j$, then $v_{\gamma_i} - v_{\gamma_j}\in E(G)$ since $X$ is a clique. Again, by~\autoref{eqn:clique-to-gt-reduction-D} we have $(\gamma_i, \gamma_j)\in S_{i,j}$.
        \end{itemize}
        Hence, by~\autoref{def:splitting-operation-vertex}, it follows that the vertex $\w_{i,j}^{\gamma_i, \gamma_j}$ is \vertsplit, i.e., $\w_{i,j,\Hor}^{\gamma_i, \gamma_j}\neq \w_{i,j,\Ver}^{\gamma_i, \gamma_j}$.

        By the construction of \Gint (\autoref{fig:main}) and definitions of canonical paths (\autoref{def:hori-canonical-G1} and~\autoref{def:vert-canonical-G1}), it is easy to verify that any pair of horizontal canonical path and vertical canonical path in \Gint have only one vertex in common.

        By the splitting operation (\autoref{def:splitting-operation-vertex}) and definitions of the paths $R_i$ (\autoref{def:vert-canonical-Gvertex}) and $T_j$ (\autoref{def:hori-canonical-Gvertex}), it follows that
        \begin{itemize}
            \item $R_i$ contains $\w_{i,j,\Ver}^{\gamma_i, \gamma_j}$ but does not contain $\w_{i,j,\Hor}^{\gamma_i, \gamma_j}$
            \item $T_j$ contains $\w_{i,j,\Hor}^{\gamma_i, \gamma_j}$ but does not contain $\w_{i,j,\Ver}^{\gamma_i, \gamma_j}$
        \end{itemize}
        Hence, it follows that $R_i$ and $T_j$ are vertex-disjoint.
    \end{proof}
    This concludes the proof of~\autoref{lem:completeness-vdsp-disjoint}.
\end{proof}
\medskip

\noindent From~\autoref{lem:completeness-vdsp-shortest} and~\autoref{lem:completeness-vdsp-disjoint}, it follows that the collection of paths given by $\mathcal{Q}=\{R_1, R_2, \ldots, R_k,$ $T_1, T_2, \ldots, T_K\}$ forms a solution for the instance $(D_{\vertex}, \mathcal{T})$ of Directed-$2k$-\VDSP.

\subsection{
    \texorpdfstring{\underline{Soundness}: $(\frac{1}{2} +\epsilon)$-fraction of the pairs in the instance $(D_{\vertex}, \mathcal{T})$ of Directed-$2k$-VDSP can be satisfied $\Rightarrow$ $G$ has a clique of size $\geq 2\epsilon \cdot k$}
    {Soundness: (1/2 + theta)-fraction of the pairs in the instance (Dvert,T) of Directed-2k-VDSP can be satisfied -> G has a clique of size >= 2 x theta x k}}
\label{sec:2kvdsp-to-clique}

In this section we show that if at least $(\frac{1}{2} +\epsilon)$-fraction of the $2k$ pairs from the instance $(D_{\vertex}, \mathcal{T})$ of $2k$-\VDSP can be satisfied then the graph $G$ has a clique of size $2\epsilon \cdot k$.
Let $\mathcal{P}$ be a collection of paths in \Gvert which satisfies at least $(\frac{1}{2} +\epsilon)$-fraction of the $2k$ terminal pairs from the instance $(D_{\vertex}, \mathcal{T})$ of $2k$-\VDSP.
\medskip

\begin{definition}
    An index $i \in [k]$ is called \emph{good} if both the terminal pairs $a_i \leadsto b_i$ and $c_i \leadsto d_i$ are satisfied by $\mathcal{P}$.
    \label{def:good-for-vertex}
\end{definition}
\medskip

The proof of the next lemma, which gives a lower bound on the number of good indices, is exactly the same as that of~\autoref{lem:good_size} and we do not repeat it here.
\medskip

\begin{lemma}
    \label{lem:good_size-vertex}
    \normalfont
    Let $Y \subseteq [k]$ be the set of good indices. Then $|Y| \ge 2\epsilon \cdot k$.
\end{lemma}
\medskip

\begin{lemma}
    \label{lem:good-equal-vertex}
    \normalfont
    If $i\in [k]$ is good, then there exists $\delta_i \in [N]$ such that the two paths in $\mathcal{P}$ satisfying $a_i\leadsto b_i$ and $c_i\leadsto d_i$ in \Gedge are the vertical canonical path  $\CanVertex(\delta_i\ ;\ a_i \leadsto b_i)$ and the horizontal canonical path $\CanVertex(\delta_i\ ;\ c_i \leadsto d_i)$ respectively.
\end{lemma}
\begin{proof}
    If $i$ is good, then by~\autoref{def:good-for-vertex} both the pairs $a_i \leadsto b_i$ and $c_i \leadsto d_i$ are satisfied by $\mathcal{P}$. Let $P_1, P_2\in \mathcal{P}$ be the paths that satisfy the terminal pairs $(a_i, b_i)$ and $(c_i, d_i)$ respectively.
    Since $P_1$ is a shortest $a_i \leadsto b_i$ path in \Gvert, by~\autoref{lem:vertical-canonical-is-shortest-G-vertex}(ii) it follows that $P_1$ is the vertical canonical path $\CanVertex(\alpha\ ;\ a_i \leadsto b_i)$ for some $\alpha\in [N]$. Since $P_2$ is a shortest $c_i \leadsto d_i$ path in \Gvert, by~\autoref{lem:horizontal-canonical-is-shortest-G-vertex}(ii) it follows that $P_2$ is the horizontal canonical path $\CanVertex(\beta\ ;\ c_i \leadsto d_i)$ for some $\beta\in [N]$.

    Using the fact that $P_1$ and $P_2$ are vertex-disjoint in \Gvert, we now claim that $\w_{i,i}^{\alpha,\beta}$ is \vertsplit:
    \begin{claim}
        \normalfont
        The vertex $\w_{i,i}^{\alpha,\beta}$ is \vertsplit by the splitting operation of~\autoref{def:splitting-operation-vertex}.
        \label{clm:must-be-vert-split-vertex-i-i}
    \end{claim}
    \begin{proof}
        By~\autoref{def:splitting-operation-vertex}, every black vertex of \Gint is either \vertsplit or \notsplit. If $\w_{i,i}^{\alpha,\beta}$ was \notsplit (\autoref{fig:split-vertex-not}), then by~\autoref{def:hori-canonical-Gvertex} and~\autoref{def:vert-canonical-Gvertex}, the vertex $\w_{i,i,\Hor}^{\alpha,\beta}=\w_{i,i,\Ver}^{\alpha,\beta}$ belongs to both $P_1$ and $P_2$ contradicting the fact that they are vertex-disjoint.
    \end{proof}
    By~\autoref{clm:must-be-vert-split-vertex-i-i}, we know that the vertex $\w_{i,i}^{\alpha,\beta}$ is \vertsplit. Hence, from~\autoref{eqn:clique-to-gt-reduction-D} and~\autoref{def:splitting-operation-vertex}, it follows that $\alpha=\beta$ which concludes the proof of the lemma.
\end{proof}
\medskip

\begin{lemma}
    \label{lem:good_vertices}
    \normalfont
    If both $i,j \in [k]$ are good and $i \neq j$, then $v_{\delta_i}-v_{\delta_j} \in E(G)$.
\end{lemma}
\begin{proof}
    Since $i$ and $j$ are good, by~\autoref{def:good-for-vertex}, there are paths $Q_1, Q_2 \in \mathcal{P}$ satisfying the pairs $(a_i, b_i), (c_j, d_j)$ respectively. By~\autoref{lem:good-equal-vertex}, it follows that
    \begin{itemize}
        \item $Q_1$ is the vertical canonical path $\CanVertex(\delta_i\ ;\ a_i\leadsto b_i)$.
        \item $Q_2$ is the horizontal canonical path $\CanVertex(\delta_j\ ;\ c_j\leadsto d_j)$.
    \end{itemize}

    Using the fact that $Q_1$ and $Q_2$ are vertex-disjoint in \Gvert, we now claim that $\w_{i,j}^{\delta_i,\delta_j}$ is \vertsplit:
    \begin{claim}
        \normalfont
        The vertex $\w_{i,j}^{\delta_i,\delta_j}$ is \vertsplit by the splitting operation of~\autoref{def:splitting-operation-vertex}.
        \label{clm:must-be-vert-split-edge-i-j}
    \end{claim}
    \begin{proof}
        By~\autoref{def:splitting-operation-vertex}, every black vertex of \Gint is either \vertsplit or \notsplit. If $\w_{i,j}^{\delta_j,\delta_j}$ was \notsplit (\autoref{fig:split-vertex-not}), then by~\autoref{def:hori-canonical-Gvertex} and~\autoref{def:vert-canonical-Gvertex}, the vertex $\w_{i,j,\Hor}^{\delta_i,\delta_j}=\w_{i,j,\Ver}^{\delta_i,\delta_j}$ belongs to both $Q_1$ and $Q_2$ contradicting the fact that they are vertex-disjoint
    \end{proof}
    By~\autoref{clm:must-be-vert-split-edge-i-j}, we know that the vertex $\w_{i,j}^{\delta_i,\delta_j}$ is \vertsplit. Since $i\neq j$, from~\autoref{eqn:clique-to-gt-reduction-D} and~\autoref{def:splitting-operation-vertex}, it follows that $v_{\delta_i}-v_{\delta_j}\in E(G)$ which concludes the proof of the lemma.
\end{proof}
\medskip

From~\autoref{lem:good_size-vertex} and~\autoref{lem:good_vertices}, it follows that the set $X:=\{v_{\delta_i}\ : i\in Y\}$ is a clique of size $\geq (2\epsilon)k$ in $G$.

\subsection{Proof of \autoref{thm:inapprox-vertex-result} and \autoref{thm:hardness-vertex-result}}
\label{sec:proof-of-main-theorem-vertex}

Finally we are ready to prove \autoref{thm:inapprox-vertex-result} and \autoref{thm:hardness-vertex-result}, which are restated below.
\medskip

\apxdirectedvertexthm*
\exactdirectedvertexthm*

\begin{proof}{\textbf{\autoref{thm:hardness-vertex-result}}}

    Given an instance $G$ of \kclique, we can use the construction from \autoref{sec:construction-of-Gvertex} to build an instance $(D_{\vertex}, \mathcal{T})$ of Directed-$2k$-\VDSP such that \Gvert is a $1$-planar DAG (\autoref{clm:Gvertex-is-1-planar-and-dag}). The graph \Gvert has $n=O(N^2 k^2)$ vertices (\autoref{clm:size-of-G-vertex}), and it is easy to observe that it can be constructed from $G$ (via first constructing \Gint) in $\poly(N,k)$ time.

    It is known that \kclique is W[1]-hard parameterized by $k$, and under ETH cannot be solved in $f(k)\cdot N^{o(k)}$ time for any computable function $f$~\cite{chen-hardness}. Combining the two directions from~\autoref{sec:2kvdsp-to-clique} (with $\epsilon = 0.5$) and~\autoref{sec:clique-to-2kvdsp}  we obtain a parameterized reduction from an instance $(G,k)$ of \kclique with $N$ vertices to an instance $(D_{\vertex}, \mathcal{T})$ of Directed-$2k$-\VDSP where \Gvert is a $1$-planar DAG (\autoref{clm:Gvertex-is-1-planar-and-dag}) and has $O(N^2 k^2)$ vertices (\autoref{clm:size-of-G-vertex}). As a result, it follows that $k$-\VDSP on $1$-planar DAGs is W[1]-hard parameterized by number $k$ of terminal pairs, and under ETH cannot be solved in $f(k)\cdot n^{o(k)}$ time where $f$ is any computable function and $n$ is the number of vertices.
\end{proof}
\medskip

The proof of~\autoref{thm:inapprox-vertex-result} is very similar to that of~\autoref{thm:inapprox-edge-result}, but we repeat the arguments here given the importance of~\autoref{thm:inapprox-vertex-result} in the paper.

\medskip

\begin{proof}{\textbf{\autoref{thm:inapprox-vertex-result}}}

    Let $\delta$ and $r_0$ be the constants from~\autoref{thm:cli_inapprox}. Fix any constant $\epsilon\in (0,1/2]$. Set $\zeta = \dfrac{\delta \epsilon}{2}$ and $k=\max \Big\{\dfrac{1}{2\zeta} , \dfrac{r_0}{2\epsilon}\Big\}$.

    Suppose to the contrary that there exists an algorithm $\mathbb{A}_{\VDSP}$ running in $f(k)\cdot n^{\zeta k}$ time (for some computable function $f$) which given an instance of Directed-$k$-\VDSP with $n$ vertices can distinguish between the following two cases:
    \begin{itemize}
        \item[(1)] All $k$ pairs of the Directed-$k$-\VDSP instance can be satisfied
        \item[(2)] The max number of pairs of the Directed-$k$-\VDSP instance that can be satisfied is less than $(\frac{1}{2}+\epsilon)\cdot k$
    \end{itemize}
    We now design an algorithm $\mathbb{A}_{\Clique}$ that contradicts~\autoref{thm:cli_inapprox} for the values $q=k$ and $r=(2\epsilon)k$. Given an instance of $(G,k)$ of \kclique with $N$ vertices, we apply the reduction from~\autoref{sec:construction-of-Gvertex} to construct an instance $(D_{\vertex}, \mathcal{T})$ of Directed-$2k$-\VDSP where \Gvert has $n=O(N^2 k^2)$ vertices (\autoref{clm:size-of-G-vertex}). It is easy to see that this reduction takes $O(N^2 k^2)$ time as well. We now show that the number of pairs which can be satisfied from the Directed-$2k$-\VDSP instance is related to the size of the max clique in $G$:
    \begin{itemize}
        \item If $G$ has a clique of size $q=k$, then by~\autoref{sec:clique-to-2kvdsp} it follows that all $2k$ pairs of the instance $(D_{\vertex},\mathcal{T})$ of Directed-$2k$-\VDSP can be satisfied.
        \item If $G$ does not have a clique of size $r=2\epsilon k$, then we claim that the max number of pairs in $\mathcal{T}$ that can be satisfied is less than $(\frac{1}{2}+\epsilon)\cdot 2k$. This is because if at least $(\frac{1}{2}+\epsilon)$-fraction of pairs in $\mathcal{T}$ could be satisfied then by~\autoref{sec:2kvdsp-to-clique} the graph $G$ would have a clique of size $\geq (2\epsilon) k=r$.
    \end{itemize}
    Since the algorithm $\mathbb{A}_{\VDSP}$ can distinguish between the two cases of all $2k$-pairs of the instance $(D_{\vertex}, \mathcal{T})$ can be satisfied or only less than $(\frac{1}{2}+\epsilon)\cdot2k$ pairs can be satisfied, it follows that $\mathbb{A}_{\Clique}$ can distinguish between the cases $\Clique(G)\geq q$ and $\Clique(G)<r$.

    The running time of the algorithm $\mathbb{A}_{\Clique}$ is the time taken for the reduction from~\autoref{sec:construction-of-Gvertex} (which is $O(N^2 k^2)$) plus the running time of the algorithm $\mathbb{A}_{\VDSP}$ which is $f(2k)\cdot n^{\zeta\cdot 2k}$. It remains to show that this can be upper bounded by $g(q,r)\cdot N^{\delta r}$ for some computable function $g$:
    \begin{align*}
        & O(N^2 k^2) + f(2k)\cdot n^{\zeta\cdot 2k}                                                                                                                  \\
        & \leq  c\cdot N^2 k^2 + f(2k)\cdot d^{\zeta\cdot 2k}\cdot (N^2 k^2)^{\zeta\cdot 2k} \tag{for some constants $c,d\geq 1$: this follows since $n=O(N^2 k^2)$} \\
        & \leq c\cdot N^2 k^2 + f'(k)\cdot N^{2\zeta\cdot 2k} \tag{where $f'(k)=f(2k)\cdot d^{\zeta\cdot 2k}\cdot k^{2\zeta\cdot 2k}$}                               \\
        & \leq 2c\cdot f'(k)\cdot N^{2\zeta\cdot 2k} \tag{since $4\zeta k\geq 2$ implies $f'(k)\geq k^2$ and $N^{2\zeta\cdot 2k}\geq N^2$}                           \\
        & = 2c\cdot f'(k) \cdot N^{\delta r} \tag{since $\zeta=\frac{\delta \epsilon}{2}$ and $r=(2\epsilon)k$}
    \end{align*}
    Hence, we obtain a contradiction to~\autoref{thm:cli_inapprox} with $q=k, r=(2\epsilon)k$ and $g(k)=2c\cdot f'(k)=2c\cdot f(2k)\cdot d^{\zeta\cdot 2k}\cdot k^{2\zeta\cdot 2k}$.
\end{proof}
\medskip

\begin{remark}
    \label{rmk:Max-Degree-2-vertex}
    \normalfont
    \textbf{(reducing the in-degree and out-degree of \Gvert)}
    By exactly the same process as described in~\autoref{rmk:hori-verti-sets-G}, we can reduce the max in-degree and max out-degree of \Gvert to be at most two whilst maintaining the properties that $n=|V(D_{\vertex})|=O(N^{2}k^{2})$ and that \Gvert can be constructed in $\poly(N,k)$ time. The splitting operation (\autoref{def:splitting-operation-vertex}) is applied only to black vertices, hence all the proofs from~\autoref{sec:characterizing-shortest-in-G-vertex},~\autoref{sec:clique-to-2kvdsp} and~\autoref{sec:2kvdsp-to-clique} go through with minor modifications.
\end{remark}
\medskip

	\section{Setting up the reductions for \texorpdfstring{$k$}{k}-\dsp on undirected graphs}
\label{sec:setting-up-the-U}

This section describes the common part of the reductions from \kclique to \UEDSP and \UVDSP, which corresponds to the top of the right-hand branch in~\autoref{fig:flowchart}. First, in~\autoref{sec:construction-of-GintU} we construct the intermediate directed graph \GintU which is later used to obtain the graphs \GedgeU (\autoref{sec:fpt-inapprox-edsp-planar-undir}) and \GvertU (\autoref{sec:fpt-inapprox-vdsp-undir}) used to obtain lower bounds for \UEDSP and \UVDSP respectively.
In~\autoref{sec:characterizing-shortest-in-G-int-U}, we then characterize shortest paths (between terminal pairs) in this intermediate graph \GintU.

We note that the intermediate graph \GintU graph is (essentially) the undirected version of the graph that was constructed for the W[1]-hardness reduction of $k$-Directed-\EDP from \gtleq by~\cite{rajesh-ciac-21}.

\subsection{Construction of the intermediate graph \GintU}
\label{sec:construction-of-GintU}

Given an instance $G=(V,E)$ of \kclique with $V=\{v_1, v_2, \ldots, v_N\}$, we now build an instance of an intermediate graph \GintU (\autoref{fig:mainUndir}). This graph \GintU is later modified to obtain the final graphs $U_{\edge}$ (\autoref{sec:construction-of-Gedge-U}) and $U_{\vertex}$, from which we obtain lower bounds for the Undirected-$k$-\EDSP and Undirected-$k$-\VDSP problems, respectively.

Before constructing the graph \GintU, we first define the following sets:
\begin{equation}
    \label{eqn:clique-to-gt-reduction-U}
    \begin{aligned}
        \text{For each}\ i\in [k],\ \text{let}\ S_{i,i}: = \{(a,a)\ :\ 1\leq a\leq N\} \\
        \text{For each pair}\ 1\leq i\neq j\leq k,\ \text{let}\ S_{i,j}:= \{ (a,b)\ :\ v_{a}-v_{b} \in E  \}
    \end{aligned}
\end{equation}

\noindent
We now construct the undirected graph \GintU via the following steps (refer to~\autoref{fig:mainUndir}):

\begin{figure}[!p]
\centering
\begin{tikzpicture}[scale=0.55]

\foreach \i in {0,1,2}
    \foreach \j in {0,1,2}
{
\begin{scope}[shift={(6*\i,6*\j)}]

        \foreach \x in {1,2,...,5}
        \foreach \y in {1,2,...,5}
    {
        \draw [black] plot [only marks, mark size=3, mark=*] coordinates {(\x,\y)};
    }

        \foreach \x in {1,2,...,5}
    \foreach \y in {1,2,3,4}
    {
        \path (\x,\y) node(a) {} (\x,\y+1) node(b) {};
        \draw[thick] (a) -- (b);
    }

        \foreach \y in {1,2,...,5}
        \foreach \x in {1,2,3,4}
    {
        \path (\x,\y) node(a) {} (\x+1,\y) node(b) {};
        \draw[thick] (a) -- (b);
    }

\end{scope}
}

\foreach \i in {0,1,2}
\foreach \j in {1,2}
{
\begin{scope}[shift={(6*\i,6*\j-6)}]

\foreach \x in {1,2,...,5}
{

    \path (\x,5) node(a) {} (\x,7) node(b) {};
        \draw[red,very thick] (a) -- (b);

}

\end{scope}
}

\foreach \j in {0,1,2}
\foreach \i in {0,1}
{
\begin{scope}[shift={(6*\i,6*\j)}]

\foreach \y in {1,2,...,5}
{
        \path (5,\y) node(a) {} (7,\y) node(b) {};
        \draw[red,very thick] (a) -- (b);
}

\end{scope}
}

\draw [green] plot [only marks, mark size=3, mark=*] coordinates {(-1,3)}
node[label={[xshift=-3mm,yshift=-4mm] $c_{1}$}] {} ;

\draw [green] plot [only marks, mark size=3, mark=*] coordinates {(-1,9)}
node[label={[xshift=-3mm,yshift=-4mm] $c_{2}$}] {} ;

\draw [green] plot [only marks, mark size=3, mark=*] coordinates {(-1,15)}
node[label={[xshift=-3mm,yshift=-4mm] $c_{3}$}] {} ;

\draw [green] plot [only marks, mark size=3, mark=*] coordinates {(19,3)}
node[label={[xshift=3mm,yshift=-4mm] $d_{1}$}] {} ;

\draw [green] plot [only marks, mark size=3, mark=*] coordinates {(19,9)}
node[label={[xshift=3mm,yshift=-4mm] $d_{2}$}] {} ;

\draw [green] plot [only marks, mark size=3, mark=*] coordinates {(19,15)}
node[label={[xshift=3mm,yshift=-4mm] $d_{3}$}] {} ;

\foreach \k in {0,1,2}
{
\begin{scope}[shift={(0,6*\k)}]
\foreach \y in {1,2,...,5}
    {
        \path (-1,3) node(a) {} (1,\y) node(b) {};
        \draw[magenta,very thick] (a) -- (b);

        \path (17,\y) node(a) {} (19,3) node(b) {};
        \draw[magenta,very thick] (a) -- (b);
    }
\end{scope}
}

\draw [green] plot [only marks, mark size=3, mark=*] coordinates {(3,-1)}
node[label={[xshift=0mm,yshift=-7mm] $a_{1}$}] {} ;

\draw [green] plot [only marks, mark size=3, mark=*] coordinates {(9,-1)}
node[label={[xshift=0mm,yshift=-7mm] $a_{2}$}] {} ;

\draw [green] plot [only marks, mark size=3, mark=*] coordinates {(15,-1)}
node[label={[xshift=0mm,yshift=-7mm] $a_{3}$}] {} ;

\draw [green] plot [only marks, mark size=3, mark=*] coordinates {(3,19)}
node[label={[xshift=0mm,yshift=0mm] $b_{1}$}] {} ;

\draw [green] plot [only marks, mark size=3, mark=*] coordinates {(9,19)}
node[label={[xshift=0mm,yshift=0mm] $b_{2}$}] {} ;

\draw [green] plot [only marks, mark size=3, mark=*] coordinates {(15,19)}
node[label={[xshift=0mm,yshift=0mm] $b_{3}$}] {} ;

\foreach \k in {0,1,2}
{
\begin{scope}[shift={(6*\k,0)}]
\foreach \x in {1,2,...,5}
    {
        \path (3,-1) node(a) {} (\x,1) node(b) {};
        \draw[magenta,very thick] (a) -- (b);

        \path (3,19) node(a) {} (\x,17) node(b) {};
        \draw[magenta,very thick] (b) -- (a);
    }
\end{scope}
}

\draw [rotate=45,black] plot [only marks, mark size=0, mark=*] coordinates
{(0,0)}
node[label={[rotate=45,xshift=0mm,yshift=-2mm] Origin}] {} ;

\end{tikzpicture}
\caption{The intermediate undirected graph \GintU constructed from an instance $(G,k)$ of \kclique (with $k=3$ and $N=5$) via the construction described in~\autoref{sec:construction-of-GintU}.
\label{fig:mainUndir}
}
\end{figure}
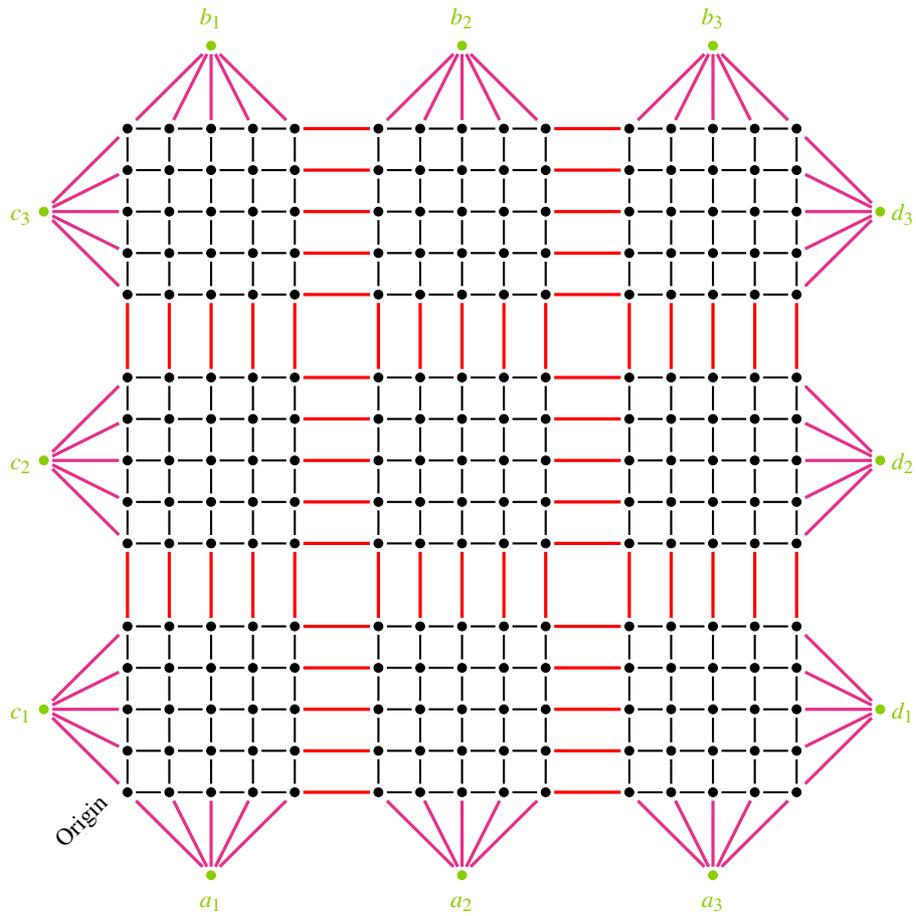

\begin{enumerate}
    \item \textbf{Origin}: The origin (vertex) is marked at the bottom left corner of \GintU (\autoref{fig:mainUndir}). This is defined just so we can view the naming of the vertices as per the usual $X-Y$ coordinate system: increasing horizontally towards the right, and vertically towards the top.

    \item \textbf{Grid (black) vertices and edges}: For each pair $1\leq i,j\leq k$, we introduce an undirected $N\times N$ grid $U_{i,j}$ where the column numbers increase from $1$ to $N$ from left to right, and the row numbers increase from $1$ to $N$ from bottom to top. For each $1\leq q,\ell\leq N$ the vertex in $q^{\text{th}}$ column and $\ell^{\text{th}}$ row of $U_{i,j}$ is denoted by $\w_{i,j}^{q,\ell}$. The vertex set and edge set of $U_{i,j}$ are defined formally as:
    \begin{itemize}

        \item $V(U_{i,j})= \big\{ \w_{i,j}^{q,\ell} : 1\leq q,\ell\leq N \big\}$

        \item $E(U_{i,j}) = \left(\bigcup_{(q,\ell)\in [N]\times [N-1]} \w_{i,j}^{q,\ell} - \w_{i,j}^{q,\ell+1} \right) \cup \left( \bigcup_{(q,\ell)\in [N-1]\times [N]} \w_{i,j}^{q,\ell} - \w_{i,j}^{q+1,\ell} \right)$
    \end{itemize}

    All black vertices have a \emph{cost} of~$1$. All vertices and edges of $U_{i,j}$ are shown in~\autoref{fig:mainUndir} in \black{black}. We later modify the grid $U_{i,j}$ in a problem-specific way (\autoref{def:splitting-operation-edge-U} and~\autoref{def:splitting-operation-vertex-U}) to \emph{represent} the set $S_{i,j}$ defined in~\autoref{eqn:clique-to-gt-reduction-U}.

    For each $1\leq i,j\leq k$ we define the set of \emph{boundary} vertices of the grid $U_{i,j}$ as follows:
    \begin{equation}
        \label{eqn:left-right-top-bottom-G-int-undir}
        \begin{aligned}
            \Le(U_{i,j}) := \big\{ \w_{i,j}^{1,\ell}\ :\ \ell\in [N]  \big\}\ ;\
            \Ri(U_{i,j}) := \big\{ \w_{i,j}^{N,\ell}\ :\ \ell\in [N]  \big\} \,.\\
            \To(U_{i,j}) := \big\{ \w_{i,j}^{\ell,N}\ :\ \ell\in [N]  \big\}\ ;\
            \Bo(U_{i,j}) := \big\{ \w_{i,j}^{\ell,1}\ :\ \ell\in [N]  \big\}
        \end{aligned}
    \end{equation}

    \item \textbf{Arranging the $N\times N$ grids}: As in the directed case, we place the $k^2$ undirected grids $\Big\{ U_{i,j} :\ (i,j)\in [k]\times [k] \Big\}$ into a big $k\times k$ \emph{grid of grids} left to right with increasing $i$ and from bottom to top with increasing $j$. In particular, the grid $U_{1,1}$ is at bottom left corner of the construction, the grid $U_{k,k}$ at the top right corner, and so on.

    \item \textbf{\red{Red} edges for horizontal connections}: For each $(i,j)\in [k-1]\times [k]$, add a set of $N$ edges that form a perfect matching between $\Ri(U_{i,j})$ and $\Le(U_{i+1,j})$ given by $\Matching\left( U_{i,j}, U_{i+1,j}  \right):= \big\{ \w_{i,j}^{N,\ell} - \w_{i+1,j}^{1,\ell}\ :\ \ell\in [N] \big\}$. Note that we can draw these perfect matchings without crossing in the plane (\autoref{fig:mainUndir}).

    \item \textbf{\red{Red} edges for vertical connections}: For each $(i,j)\in [k]\times [k-1]$, add a set of $N$ edges that form a perfect matching between $\To(U_{i,j})$ and $\Bo(U_{i,j+1})$ given by $\Matching\left( U_{i,j}, U_{i,j+1}  \right):= \big\{ \w_{i,j}^{\ell,N} - \w_{i,j+1}^{\ell,1}\ :\ \ell\in [N] \big\}$. Note that we can draw these perfect matchings without crossing in the plane (\autoref{fig:mainUndir}).

    \item \textbf{\green{Green} (terminal) vertices and \magenta{magenta} edges}: For each $i\in [k]$, we define the four sets of terminal below:
    \begin{equation*}
        \label{eqn:A-B-C-D-D-U}
        \begin{aligned}
            A := \big\{ a_i\ :\ i\in [k] \big\}\  \text{forming a bottom row}\ ;\ B := \big\{ b_i\ :\ i\in [k] \big\}\ \text{a top row}\\
            C := \big\{ c_i\ :\ i\in [k] \big\}\  \text{a left column}\ ;\ D := \big\{ d_i\ :\ i\in [k]\big\} \ \text{a right column}
        \end{aligned}
    \end{equation*}

    For each $i\in [k]$, we add the edges (shown in~\autoref{fig:mainUndir} in \magenta{magenta})
    \begin{equation}
        \label{eqn:source-sink-A-B-undir}
        \begin{aligned}
            \Source(A) := \big\{ a_i - \w_{i,1}^{\ell,1}\ :\ \ell\in [N] \big\}\ ;\
            \Sink(B) := \big\{ \w_{i,k}^{\ell,N} - b_{i}\ :\ \ell\in [N] \big\}
        \end{aligned}
    \end{equation}
    For each $j\in [k]$, we add the edges (shown in~\autoref{fig:mainUndir} in \magenta{magenta})

    \begin{equation}
        \label{eqn:source-sink-C-D-undir}
        \begin{aligned}
            \Source(C) := \big\{ c_j - \w_{1,j}^{1,\ell}\ :\ \ell\in [N] \big\}\ ;\
            \Sink(D) := \big\{ \w_{k,j}^{N,\ell} - d_{j}\ :\ \ell\in [N] \big\}
        \end{aligned}
    \end{equation}

\end{enumerate}

\noindent
\medskip
\begin{definition}
    \label{def:lefty-right-topy-bottomy-U}
    \textbf{(four neighbors of each grid vertex in \GintU)} Consider the drawing of \GintU from \autoref{fig:mainUndir}. This gives the natural notion of four neighbors for every black grid vertex: one to the left, right, bottom and top of each. For each (black) grid vertex $\z\in$ \GintU we define these as follows
    \begin{itemize}
        \item $\lefty(\z)$ is the vertex to the left of $\z$ which shares an edge with $\z$
        \item $\bottomy(\z)$ is the vertex below $\z$ which shares an edge with $\z$
        \item $\righty(\z)$ is the vertex to the right of $\z$ which shares an edge with $\z$
        \item $\topy(\z)$ is the vertex above $\z$ which shares an edge with $\z$
    \end{itemize}
    Note that in the case that $\z$ lies on the edge of the grid in \autoref{fig:mainUndir}, up to $2$ of its neighbours are in fact \green{green} terminal vertices.
\end{definition}
\medskip
This completes the construction of the undirected graph \GintU (\autoref{fig:mainUndir}). The next two claims analyze the structure and size of this graph:

\medskip
\begin{claim}
    \label{clm:G-int-U-is-planar-and-dag}
    \normalfont
    \GintU is planar.
\end{claim}
\begin{proof}
    ~\autoref{fig:mainUndir} gives a planar embedding of \GintU.

\end{proof}
\medskip
\begin{claim}
    \normalfont
    The number of vertices in \GintU is $O(N^{2}k^{2})$.
    \label{clm:size-of-G-int-U}
\end{claim}
\begin{proof}
    \GintU has $k^2$ different $N\times N$ grids viz. $\{U_{i,j}\}_{1\leq i,j\leq k}$. Hence, \GintU has $N^{2}k^2$ black vertices. Adding the $4k$ \green{green} vertices from $A\cup B\cup C\cup D$ it follows that number of vertices in \GintU is $N^{2}k^2+ 4k = O(N^2 k^2)$.
\end{proof}
\medskip

\subsection{Characterizing shortest paths in \GintU}
\label{sec:characterizing-shortest-in-G-int-U}

The goal of this section is to characterize the structure of shortest paths between terminal pairs in \GintU.
In order to do this, we need to define the set of terminal pairs $\mathcal{T}$ and assign a cost to vertices of \GintU.

\begin{equation}
    \label{eqn:definition-of-mathcal-T-undir}
    \text{The set of terminal pairs is}\ \mathcal{T}:= \big\{(a_i, b_i) : i\in [k] \big\}\cup \big\{(c_j, d_j) : j\in [k] \big\}\ ;
\end{equation}

\medskip
\begin{definition}
    \label{def:weights-in-undirected} (\textbf{costs of vertices in \GintU}) Each black vertex in \GintU has a cost of 1 and each green vertex has a cost of $2kN$.

    \normalfont
    \autoref{def:weights-in-undirected} gives a cost to each vertex of \GintU which then naturally leads to the notion of cost of a path as the sum of costs of the vertices on it. We show in \autoref{rmk:vert-cost-1-undir}, how we can adapt our graph to an equivalent one with all vertex costs of $1$ and hence we could equivalently measure the cost of a given path by counting either the number of edges or the number of vertices. Thus our choice to measure the cost in terms of the number of vertices has no bearing on the results that we obtain.

\end{definition}

\medskip
We now define row-paths and column-paths which are the building blocks of what we later term as \emph{canonical paths}.

\medskip
\begin{definition}
    \normalfont
    \textbf{(row-paths and column-paths in \GintU)} For each $(i,j)\in [k]\times [k]$ and $\ell\in [N]$ we define
    \begin{itemize}
        \item $\Row_{\ell}(U_{i,j})$ to be the $\w_{i,j}^{1,\ell}- \w_{i,j}^{N,\ell}$ path in $U_{\inter}[U_{i,j}]$ consisting of the following edges (in order): for each $r\in [N-1]$ take the black edge $\w_{i,j}^{r,\ell} - \w_{i,j}^{r+1,\ell}$.

        \item $\Column_{\ell}(U_{i,j})$ to be the $\w_{i,j}^{\ell,1}- \w_{i,j}^{\ell,N}$ path in $U_{\inter}[U_{i,j}]$  consisting of the following edges (in order): for each $r\in [N-1]$ take the black edge $\w_{i,j}^{\ell,r} - \w_{i,j}^{\ell,r+1}$.
    \end{itemize}

    \label{def:Column-Row-G-int-U}
\end{definition}
\medskip

Each row-path and each column-path in \GintU contains exactly $N$ (black) vertices: hence, by~\autoref{def:weights-in-undirected}, the cost of any row-path or column-path in \GintU is $N$. We are now ready to define horizontal canonical paths and vertical canonical paths in \GintU:

\medskip
\begin{definition}
    \textbf{(horizontal canonical paths in \GintU)}
    Fix any $j\in [k]$. For each $r\in [N]$, we define $\CanInterU(r\ ;\ c_j - d_j)$ to be the $c_j - d_j$ path in \GintU given by the following edges (in order):
    \begin{itemize}
        \item Start with the \magenta{magenta} edge $c_j - \w_{1,j}^{1,r}$
        \item For each $i\in [k-1]$, the path $\w_{i,j}^{1,r} - \w_{i+1,j}^{1,r}$  obtained by concatenating the $\w_{i,j}^{1,r} - \w_{i,j}^{N,r}$ path $\Row_{r}(U_{i,j})$ from~\autoref{def:Column-Row-G-int-U} with the \red{red} edge $\w_{i,j}^{N,r} - \w_{i+1,j}^{1,r}$.
        \item
        Then use the $\w_{k,j}^{1,r} - \w_{k,j}^{N,r}$ path $\Row_{r}(U_{k,j})$ from~\autoref{def:Column-Row-G-int-U} to reach the vertex $\w_{k,j}^{N,r}$.
        \item Finally, use the \magenta{magenta} edge $\w_{k,j}^{N,r} - d_j$ to reach $d_j$.
    \end{itemize}

    \label{def:hori-canonical-G-int-U}
\end{definition}
\medskip

\begin{definition}
    \textbf{(vertical canonical paths in \GintU)}
    Fix any $i\in [k]$. For each $r\in [N]$, we define $\CanInterU(r\ ;\ a_i - b_i)$ to be the $a_i -b_i$ path in \GintU given by the following edges (in order):
    \begin{itemize}
        \item Start with the \magenta{magenta} edge $a_i - \w_{i,1}^{r,1}$
        \item For each $j\in [k-1]$ use the $\w_{i,j}^{r,1} - \w_{i,j+1}^{r,1}$ path obtained by concatenating the $\w_{i,j}^{r,1} - \w_{i,j}^{r,N}$ path $\Column_{r}(U_{i,j})$ from~\autoref{def:Column-Row-G-int-U} with the \red{red} edge $\w_{i,j}^{r,N} - \w_{i,j+1}^{r,1}$.
        \item Then use the $\w_{i,k}^{r,1} - \w_{i,k}^{r,N}$ path $\Column_{r}(U_{i,k})$ from~\autoref{def:Column-Row-G-int-U} to reach the vertex $\w_{i,k}^{r,N}$.
        \item Finally, use the \magenta{magenta} edge $\w_{j,k}^{r,N} - b_i$ to reach $b_i$.
    \end{itemize}

    \label{def:vert-canonical-G-int-U}
\end{definition}
\medskip

We now calculate the cost of horizontal canonical and vertical canonical paths:
\medskip
\begin{observation}
    \normalfont
    From~\autoref{def:hori-canonical-G-int-U}, every horizontal canonical path in \GintU starts and ends with a \green{green} vertex, and has~$kN$ black vertices between ($k$ different row-paths each of which has $N$ black vertices). From~\autoref{def:weights-in-undirected}, it follows that
    each horizontal canonical path in \GintU has a cost of exactly $5kN$. Similarly, it is easy to see that each vertical canonical path in \GintU has a cost of exactly $5kN$.
    \label{obs:Canonical-length-U}
\end{observation}
\medskip

\begin{remark}
    \label{rmk:vert-cost-1-undir}
    \normalfont
    \textbf{(Reducing the cost of vertices in \GintU)}
    The only vertices in \GintU which have a cost greater than $1$ are $A \cup B \cup C \cup D$. We show how to reduce the cost of vertices from $A$ whilst preserving the structure of vertical canonical paths (\autoref{def:vert-canonical-G-int-U}). The argument for vertices from $B \cup C \cup D$ is analogous. Fix  $i\in [k]$. The cost of visiting $a_i$ is $2kN$ in \GintU. For every $q \in N$, replace the edge $a_i - \w_{i,1}^{q, 1}$ with a path $a_i - \w_{i,1}^{q, 1}$ that visits exactly $2kN-1$ new (black) vertices along the way. Each of these new vertices have a cost of $1$, and the cost of $a_i$ is then also set to $1$. All edges created have a \magenta{magenta} colour and $a_i$ maintains its \green{green} colour. For each of these new routes $a_i - \w_{i,1}^{q, 1}$, any path that previously took an edge $a_i - \w_{i,1}^{q, 1}$ now visit either none or all of the $2kN -1$ new (black) vertices along it.

    In applying this reduction we must redefine the initial step of the vertical canonical paths such that they all start by taking the $2kN-1$ edges along the path $a_i - \w_{i,1}^{q, 1}$ for any $i \in [k]$. This increases the cost of every canonical path by a constant amount $(2kN)$ and thus our claims about the properties of the canonical paths still hold after the reduction.

    It is easy to see that this editing to \GintU adds $O(k\cdot N)$ new vertices and takes $\poly(N)$ time, and therefore it is still true (from~\autoref{clm:size-of-G-int-U}) that $n=|V(U_{\edge})|=O(N^{2}k^{2})$ and \GintU can be constructed in $\poly(N,k)$ time.

    Observe, also, that this process ensures that every vertex in \GintU has maximum degree of $4$.
\end{remark}
\medskip

The next two lemmas give a characterization of the shortest paths between terminal pairs.

\medskip
\begin{lemma}
    \normalfont
    Let $j\in [k]$. The horizontal canonical paths in \GintU satisfy the following two properties:
    \begin{itemize}
        \item[(i)] For each $r\in [N]$, the path $\CanInterU(r\ ;\ c_j - d_j)$ is a shortest $c_j - d_j$ path in \GintU.
        \item[(ii)] If $P$ is a shortest $c_j - d_j$ path in \GintU, then $P$ must be $\CanInterU(\ell\ ;\ c_j - d_j)$ for some $\ell\in [N]$.
    \end{itemize}
    \label{lem:horizontal-canonical-is-shortest-G-int-U}
\end{lemma}
\begin{proof}
    Towards proving the lemma, we first show a preliminary claim which lower bounds the cost of any $c_j - d_j$ path in \GintU:
    \begin{claim}
        \label{clm:each-hori-at-least-5kN}
        Any $c_j - d_j$ path has cost $\geq 5kN$.
    \end{claim}
    \begin{proof}
        Let $Q$ be any $c_j - d_j$ path in \GintU. If $Q$ contains any green vertex besides $c_j$ or $d_j$, then the cost of $Q$ is $\geq 3\cdot 2kN = 6kN$ since each green vertex has cost $2kN$ (\autoref{def:weights-in-undirected}).

        Hence, it remains to consider $c_j - d_j$ paths which contain only two green vertices viz. $c_j$ and $d_j$. Let $U_{\inter}^{*}$ be the graph obtained from \GintU by deleting the vertices from $A\cup B$. The paths we need to consider in this case now are contained in the graph $U_{\inter}^{*}$. For each $1\leq i\leq k$ and each $1\leq q\leq N$, define the following set of vertices
        \[ \text{Column}(i,q) :=  \bigcup_{1\leq s\leq k; 1\leq \ell, N} \w_{i,s}^{q,\ell}\]
        It is easy to see that $c_j$ and $d_j$ belong to different connected components of $U^{*}_{\inter}$ if we delete all the vertices of $\text{Column}(i,q)$ for any $1\leq i\leq k$ and $1\leq q\leq N$. Moreover, if $(i,q)\neq (i',q')$ then $\text{Column}(i,q) \cap \Column(i',q')=\emptyset$. Hence, it follows that $Q$ contains at least one (black) vertex from $\text{Column}(i,q)$ for each $1\leq i\leq k$ and $1\leq q\leq N$. Since all these vertices are black, the weight of internal (black) vertices of $Q$ is at least $kN$. Therefore, including the two \green{green} endpoints $c_j$ and $d_j$, the weight of any $c_j - d_j$ path is at least $5kN$.
    \end{proof}

    The proof of the first part of the lemma now follows from~\autoref{clm:each-hori-at-least-5kN} and~\autoref{obs:Canonical-length-U}.

    Now we prove the second part of the lemma. Let $X$ be any shortest $c_j - d_j$ path in \GintU. By~\autoref{clm:each-hori-at-least-5kN} and~\autoref{obs:Canonical-length-U}, it follows that the weight of $X$ is exactly $5kN$. The two \green{green} endpoints $c_j$ and $d_j$ incur a total cost of $2kN + 2kN = 4kN$. This leaves a budget of $kN$ available for other vertices of $X$. In particular, $X$ cannot contain any other green vertex besides $c_j$ and $d_j$. Hence, following the proof of~\autoref{clm:each-hori-at-least-5kN}, it follows that $X$ contains at least one vertex from $\text{Column}(i,q)$ for each $1\leq i\leq k$ and $1\leq q\leq N$. This takes up a budget of at least $kN$ which is all that was available. Hence, $X$ contains exactly one vertex from $\text{Column}(i,q)$ for each $1\leq i\leq k$ and $1\leq q\leq N$. Looking at the structure of \GintU, it follows that $X$ must be the same as $\CanInterU(r\ ;\ c_j - d_j)$ for some $r\in [N]$ (where $r$ is the first black vertex from $\Le(U_{1,j})$ that occurs on $X$ after it leaves $c_j$).
    This concludes the proof of~\autoref{lem:horizontal-canonical-is-shortest-G-int-U}.
\end{proof}
\medskip

The proof of the next lemma is very similar to that of \autoref{lem:horizontal-canonical-is-shortest-G-int-U}, and we skip repeating the details.
\medskip

\begin{lemma}
    \normalfont
    Let $i\in [k]$. The vertical canonical paths in \GintU satisfy the following two properties:
    \begin{itemize}
        \item For each $r\in [N]$, the path $\CanInterU(r\ ;\ a_i - b_i)$ is a shortest $a_i - b_i$ path in \GintU.
        \item If $P$ is a shortest $a_i - b_i$ path in \GintU, then $P$ must be $\CanInterU(\ell\ ;\ a_i - b_i)$ for some $\ell\in [N]$.
    \end{itemize}
    \label{lem:vertical-canonical-is-shortest-G-int-U}
\end{lemma}
\medskip

	\section{Lower bounds for Undirected-\texorpdfstring{$k$}{k}-\EDSP on planar graphs}
\label{sec:fpt-inapprox-edsp-planar-undir}

The goal of this section is to prove lower bounds on the running time of exact (\autoref{thm:hardness-edge-result-U}) and approximate (\autoref{thm:inapprox-edge-result-U}) algorithms for the Undirected-$k$-\EDSP problem. We have already seen the first part of the reduction (\autoref{sec:construction-of-GintU}) from \kclique resulting in the construction of the intermediate graph \GintU.~\autoref{sec:construction-of-Gedge-U} describes the next part of the reduction which edits the intermediate graph \GintU to obtain the final graph \GedgeU. This corresponds to the ancestry of the third leaf in~\autoref{fig:flowchart}. The characterization of shortest paths between terminal pairs in \GedgeU is given in~\autoref{sec:characterizing-shortest-in-G-edge-U}. The completeness and soundness of the reduction from \kclique to Undirected-$2k$-\EDSP are proven
in~\autoref{sec:clique-to-2kedsp-U} and~\autoref{sec:2kedsp-to-clique-U}, respectively. Finally, we state our final results in~\autoref{sec:proof-of-main-theorem-edge-U} allowing us to prove~\autoref{thm:hardness-edge-result-U} and~\autoref{thm:inapprox-edge-result-U}.

\subsection{Obtaining the graph \GedgeU from \GintU via the splitting operation}
\label{sec:construction-of-Gedge-U}

Observe in \autoref{fig:mainUndir} that every black grid vertex in \GintU has degree exactly four, and these four neighbors are named as per~\autoref{def:lefty-right-topy-bottomy-U}.
We now define the splitting operation which allows us to obtain the graph \GedgeU from the graph \GintU constructed in~\autoref{sec:construction-of-GintU}.

\medskip
\begin{definition}
    \textbf{(splitting operation to obtain \GedgeU from \GintU)} For each $i,j\in [k]$ and each $q,\ell\in [N]$
    \begin{itemize}
        \item If $(q,\ell)\notin S_{i,j}$, then we \onesplit (\autoref{fig:split-edge-not-U}) the vertex $\w_{i,j}^{q,\ell}$ into \textbf{three distinct} vertices $\w_{i,j,\LB}^{q,\ell}, \w_{i,j,\Mid}^{q,\ell}$ and $\w_{i,j,\TR}^{q,\ell}$ and add the path $\w_{i,j,\LB}^{q,\ell} - \w_{i,j,\Mid}^{q,\ell} - \w_{i,j,\TR}^{q,\ell}$ (denoted by dotted edges in~\autoref{fig:split-edge-not-U}).

        \item Otherwise, if $(q,\ell)\in S_{i,j}$ then we \twosplit (\autoref{fig:split-edge-yes-U}) the vertex $\w_{i,j}^{q,\ell}$ into \textbf{four distinct} vertices $\w_{i,j,\LB}^{q,\ell}, \w_{i,j,\Hor}^{q,\ell}, \w_{i,j,\Ver}^{q,\ell}$ and $\w_{i,j,\TR}^{q,\ell}$ and add the two paths $\w_{i,j,\LB}^{q,\ell} - \w_{i,j,\Hor}^{q,\ell} - \w_{i,j,\TR}^{q,\ell}$ and $\w_{i,j,\LB}^{q,\ell} - \w_{i,j,\Ver}^{q,\ell} - \w_{i,j,\TR}^{q,\ell}$ (denoted by dotted edges in~\autoref{fig:split-edge-yes-U}).

    \end{itemize}
    The 4 edges (\autoref{def:lefty-right-topy-bottomy-U-edge}) incident on $\w_{i,j}^{q,\ell}$ are now changed as follows:
    \begin{itemize}
        \item Replace the edge $\lefty(\w_{i,j}^{q,\ell}) - \w_{i,j}^{q,\ell}$ by the edge $\lefty(\w_{i,j}^{q,\ell}) - \w_{i,j,\LB}^{q,\ell}$
        \item Replace the edge $\bottomy(\w_{i,j}^{q,\ell}) - \w_{i,j}^{q,\ell}$ by the edge $\bottomy(\w_{i,j}^{q,\ell}) - \w_{i,j,\LB}^{q,\ell}$
        \item Replace the edge $\w_{i,j}^{q,\ell} - \righty(\w_{i,j}^{q,\ell})$ by the edge $\w_{i,j,\TR}^{q,\ell} - \righty(\w_{i,j}^{q,\ell})$
        \item Replace the edge $\w_{i,j}^{q,\ell} - \topy(\w_{i,j}^{q,\ell})$ by the edge $\w_{i,j,\TR}^{q,\ell} - \topy(\w_{i,j}^{q,\ell})$
    \end{itemize}
    \label{def:splitting-operation-edge-U}
\end{definition}
\medskip

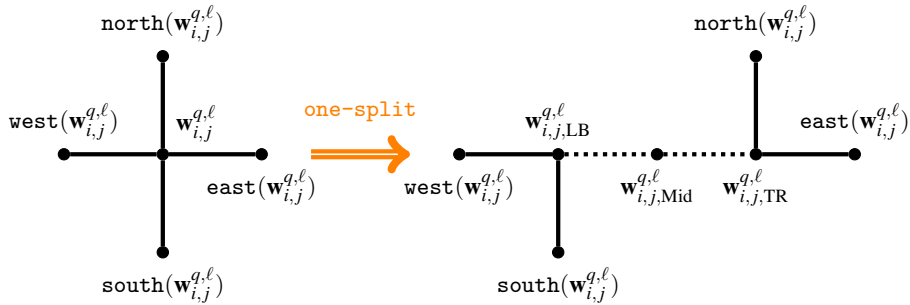
\begin{figure}[hbt!]
\centering
\begin{tikzpicture}[
vertex/.style={circle, draw=black, fill=black, text width=1.5mm, inner sep=0pt},
scale=0.65]
\node[vertex, label=above right:\footnotesize{$\w_{i,j}^{q,\ell}$}] (v) at (0,0) {} ;
\node[vertex, label=above:\footnotesize{$\lefty(\w_{i,j}^{q,\ell})$}] (l) at (-2,0) {};
\node[vertex, label=below:\footnotesize{$\righty(\w_{i,j}^{q,\ell})$}] (r) at (2,0) {};
\node[vertex, label=below:\footnotesize{$\bottomy(\w_{i,j}^{q,\ell})$}] (b) at (0,-2) {};
\node[vertex, label=above:\footnotesize{$\topy(\w_{i,j}^{q,\ell})$}] (t) at (0,2) {};
\draw[ultra thick] (v) -- (t);
\draw[ultra thick] (v) -- (r);
\draw[ultra thick] (l) -- (v);
\draw[ultra thick] (b) -- (v);

\draw[orange,double, ultra thick,->] (3,0) -- node[above=3mm, draw=none, fill=none, rectangle] {\onesplit } (5,0);

\node[vertex, label=below:\footnotesize{$\w_{i,j,\TR}^{q,\ell}$}] (vtr) at (12,0) {} ;
\node[vertex, label=below:\footnotesize{$\w_{i,j,\Mid}^{q,\ell}$}] (vm) at (10,0) {} ;
\node[vertex, label=above:$\w_{i,j,\LB}^{q,\ell}$] (vlb) at (8,-0) {} ;

\node[vertex, label=below:\footnotesize{$\lefty(\w_{i,j}^{q,\ell})$}] (l) at (6,0) {};
\node[vertex, label=above:\footnotesize{$\righty(\w_{i,j}^{q,\ell})$}] (r) at (14,0) {};
\node[vertex, label=below:\footnotesize{$\bottomy(\w_{i,j}^{q,\ell})$}] (b) at (8,-2) {};
\node[vertex, label=above:\footnotesize{$\topy(\w_{i,j}^{q,\ell})$}] (t) at (12,2) {};
\draw[ultra thick] (vtr) -- (t);
\draw[ultra thick] (vtr) -- (r);
\draw[ultra thick] (l) -- (vlb);
\draw[ultra thick] (b) -- (vlb);
\draw[dotted,ultra thick] (vlb) -- (vm);
\draw[dotted,ultra thick] (vm) -- (vtr);
\end{tikzpicture}

\caption{The \onesplit operation for the vertex $\w_{i,j}^{q,\ell}$ when
$(q,\ell)\notin S_{i,j}$. The idea behind this splitting is that the horizontal path $\lefty(w_{i,j}^{q,\ell}) - w_{i,j}^{q,\ell} - \righty(w_{i,j}^{q,\ell})$ and vertical path $\bottomy(w_{i,j}^{q,\ell}) - w_{i,j}^{q,\ell} - \topy(w_{i,j}^{q,\ell})$  are no longer edge-disjoint after the \onesplit operation as they must share the path $w_{i,j,\LB}^{q,\ell} - w_{i,j,\Mid}^{q,\ell} - w_{i,j,\TR}^{q,\ell}$.}

\label{fig:split-edge-not-U}
\end{figure}
\begin{figure}[hbt!]
\centering
\begin{tikzpicture}[
vertex/.style={circle, draw=black, fill=black, text width=1.5mm, inner sep=0pt},
scale=0.65]
\node[vertex, label=above right:\footnotesize{$\w_{i,j}^{q,\ell}$}] (v) at (0,0) {} ;
\node[vertex, label=above:\footnotesize{$\lefty(\w_{i,j}^{q,\ell})$}] (l) at (-2,0) {};
\node[vertex, label=below:\footnotesize{$\righty(\w_{i,j}^{q,\ell})$}] (r) at (2,0) {};
\node[vertex, label=below:\footnotesize{$\bottomy(\w_{i,j}^{q,\ell})$}] (b) at (0,-2) {};
\node[vertex, label=above:\footnotesize{$\topy(\w_{i,j}^{q,\ell})$}] (t) at (0,2) {};
\draw[ultra thick] (v) -- (t);
\draw[ultra thick] (v) -- (r);
\draw[ultra thick] (l) -- (v);
\draw[ultra thick] (b) -- (v);

\draw[orange,double, ultra thick,->] (3,0) -- node[above=3mm, draw=none, fill=none, rectangle] {\twosplit } (5,0);

\node[vertex, label=below:\footnotesize{$\w_{i,j,\TR}^{q,\ell}$}] (vtr) at (12,0) {} ;
\node[vertex, label=below:\footnotesize{$\w_{i,j,\Hor}^{q,\ell}$}] (vh) at (10,-1) {} ;
\node[vertex, label=above:\footnotesize{$\w_{i,j,\Ver}^{q,\ell}$}] (vv) at (10,1) {} ;
\node[vertex, label=above:$\w_{i,j,\LB}^{q,\ell}$] (vlb) at (8,-0) {} ;

\node[vertex, label=below:\footnotesize{$\lefty(\w_{i,j}^{q,\ell})$}] (l) at (6,0) {};
\node[vertex, label=above:\footnotesize{$\righty(\w_{i,j}^{q,\ell})$}] (r) at (14,0) {};
\node[vertex, label=below:\footnotesize{$\bottomy(\w_{i,j}^{q,\ell})$}] (b) at (8,-2) {};
\node[vertex, label=above:\footnotesize{$\topy(\w_{i,j}^{q,\ell})$}] (t) at (12,2) {};
\draw[ultra thick] (vtr) -- (t);
\draw[ultra thick] (vtr) -- (r);
\draw[ultra thick] (l) -- (vlb);
\draw[ultra thick] (b) -- (vlb);
\draw[dotted,ultra thick] (vlb) -- (vh);
\draw[dotted,ultra thick] (vh) -- (vtr);

\draw[dotted,ultra thick] (vlb) -- (vv);
\draw[dotted,ultra thick] (vv) -- (vtr);
\end{tikzpicture}

\caption{The \twosplit operation for the vertex $\w_{i,j}^{q,\ell}$ when
$(q,\ell)\in S_{i,j}$. The idea behind this splitting is that the horizontal path $\lefty(w_{i,j}^{q,\ell}) - w_{i,j}^{q,\ell} - \righty(w_{i,j}^{q,\ell})$ and vertical path $\bottomy(w_{i,j}^{q,\ell}) - w_{i,j}^{q,\ell} - \topy(w_{i,j}^{q,\ell})$  are still edge-disjoint after the \twosplit operation if we replace them with the paths $\lefty(w_{i,j}^{q,\ell}) - w_{i,j,\LB}^{q,\ell} - w_{i,j,\Hor}^{q,\ell} - w_{i,j,\TR}^{q,\ell} - \righty(w_{i,j}^{q,\ell})$ and $\bottomy(w_{i,j}^{q,\ell}) - w_{i,j,\LB}^{q,\ell} - w_{i,j,\Ver}^{q,\ell} - w_{i,j,\TR}^{q,\ell} - \topy(w_{i,j}^{q,\ell})$ respectively.
}
\label{fig:split-edge-yes-U}

\end{figure}
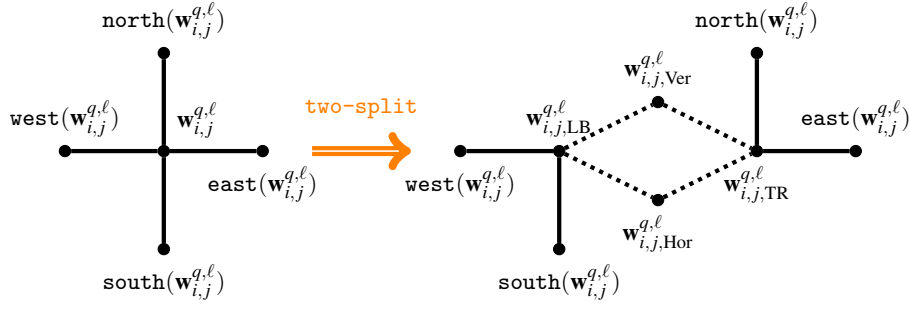

Finally, we are now ready to define the instance of Undirected-$2k$-\EDSP that we have built starting from an instance $G$ of \kclique.

\medskip
\begin{definition}
    \normalfont
    \textbf{(defining the Undirected-$2k$-\EDSP instance)} The instance $(U_{\edge}, \mathcal{T})$ of Undirected-$2k$-\EDSP is defined as follows:
    \begin{itemize}
        \item The graph \GedgeU is obtained by applying the splitting operation (\autoref{def:splitting-operation-edge-U}) to each (black) grid vertex of \GintU, i.e., the set of vertices given by $\bigcup_{1\leq i,j\leq k} V(U_{i,j})$.
        \item No \green{green} vertex is split in~\autoref{def:splitting-operation-edge-U}, and hence the set of terminal pairs remains the same as defined in~\autoref{eqn:definition-of-mathcal-T-undir} and is given by $\mathcal{T}:= \big\{(a_i, b_i) : i\in [k] \big\}\cup \big\{(c_j, d_j) : j\in [k] \big\}$.
        \item We assign a cost of $1$ to each new vertex created during the splitting operation (\autoref{def:splitting-operation-edge-U}). Since each vertex of \GintU has a cost of $1$, it follows that each vertex of \GedgeU also has a visit cost of $1$.
    \end{itemize}
    \label{def:G-edge-U}
\end{definition}
\medskip

\begin{claim}
    \label{clm:GedgeU-is-planar}
    \normalfont
    $U_{\edge}$ is planar.
\end{claim}
\begin{proof}
    In~\autoref{clm:G-int-U-is-planar-and-dag}, we have shown that \GintU is planar. The graph \GedgeU is obtained from \GintU by applying the splitting operation (\autoref{def:splitting-operation-edge-U}) on every (black) grid vertex, i.e., every vertex from the set $\bigcup_{1\leq i,j\leq k} V(U_{i,j})$. By~\autoref{def:lefty-right-topy-bottomy-U-edge}, every vertex of \GintU that is split has four neighbors in \GedgeU. Hence, one can observe (\autoref{fig:split-edge-not-U} and~\autoref{fig:split-edge-yes-U}) that the splitting operation (\autoref{def:splitting-operation-edge-U}) preserves planarity when we construct \GedgeU from \GintU.
\end{proof}

\medskip

\begin{claim}
    \normalfont
    The number of vertices in \GedgeU is $O(N^{2}k^{2})$.
    \label{clm:size-of-G-edge-U}
\end{claim}
\begin{proof}
    By~\autoref{clm:size-of-G-int-U}, the graph \GintU has $O(N^2 k^2)$ vertices. The only change when obtaining \GedgeU from \GintU is the splitting operation (\autoref{def:splitting-operation-edge-U}) adds at most three extra vertices for each black vertex of \GintU. Hence, the number of vertices of \GedgeU is $O(N^{2}k^{2})$.
\end{proof}
\medskip

\begin{definition}
    \label{def:lefty-right-topy-bottomy-U-edge}
    Recall \autoref{def:lefty-right-topy-bottomy-U}, where we defined the four neighbours of any grid vertex in \GintU. We maintain these definitions of the neighbours for each (black) grid vertex here in \GedgeU.
\end{definition}
\medskip

\subsection{Characterizing shortest paths in \GedgeU}
\label{sec:characterizing-shortest-in-G-edge-U}

The goal of this section is to characterize the structure of shortest paths between terminal pairs in \GedgeU. Recall (\autoref{def:G-edge-U}) that the set of terminal pairs is given by $\mathcal{T}:= \big\{(a_i, b_i) : i\in [k] \big\}\cup \big\{(c_j, d_j) : j\in [k] \big\}$.
As in \autoref{sec:setting-up-the-U}, the length of a path is the sum of the vertex costs.

We now define canonical paths in \GedgeU by adapting the definition of canonical paths (\autoref{def:hori-canonical-G-int-U} and~\autoref{def:vert-canonical-G-int-U}) in \GintU in accordance with the changes in going from \GintU to \GedgeU.

\medskip
\begin{definition}
    \textbf{(horizontal canonical paths in \GedgeU)}
    Fix a $j\in [k]$. For each $r\in [N]$, we define $\CanEdgeU(r\ ;\ c_j - d_j)$ to be the $c_j - d_j$ path in \GedgeU obtained from the path $\CanInterU(r\ ;\ c_j - d_j)$ in \GintU (recall~\autoref{def:hori-canonical-G-int-U}) in the following way:
    \begin{itemize}
        \item The first and last \magenta{magenta} edges are unchanged.
        \item If a black grid vertex $\w$ from $\CanInterU(r\ ;\ c_j - d_j)$ is \onesplit (\autoref{fig:split-edge-not-U}), then
        \begin{itemize}
            \item The unique edge $\lefty(\w) - \w$ is replaced with the edge $\lefty(\w) - \w_{\LB}$;
            \item The unique edge $\w - \righty(\w)$ is replaced with the edge $\w_{\TR} - \righty(\w)$;
            \item The path $\w_{\LB} - \w_{\Mid} - \w_{\TR}$ is added.
        \end{itemize}
        \item If a black grid vertex $\w$ from $\CanInterU(r\ ;\ c_j - d_j)$ is \twosplit (\autoref{fig:split-edge-yes-U}), then
        \begin{itemize}
            \item The unique edge $\lefty(\w) - \w$ is replaced with the edge $\lefty(\w) - \w_{\LB}$;
            \item The unique edge $\w - \righty(\w)$ is replaced with the edge $\w_{\TR} - \righty(\w)$;
            \item The path $\w_{\LB} - \w_{\Hor} - \w_{\TR}$ is added.
        \end{itemize}
    \end{itemize}

    \label{def:hori-canonical-Gedge-U}
\end{definition}
\medskip

\begin{definition}
    \textbf{(vertical canonical paths in \GedgeU)}
    Fix a $i\in [k]$. For each $r\in [N]$, we define $\CanEdgeU(r\ ;\ a_i - b_i)$ to be the $a_i - b_i$ path in \GedgeU obtained from the path $\CanInterU(r\ ;\ a_i - b_i)$ in \GintU (recall \autoref{def:vert-canonical-G-int-U}) in the following way.
    \begin{itemize}
        \item The first and last \magenta{magenta} edges are unchanged.
        \item If a black grid vertex $\w$ from $\CanInterU(r\ ;\ a_i - b_i)$ is \onesplit (\autoref{fig:split-edge-not-U}), then
        \begin{itemize}
            \item The unique edge $\topy(\w) - \w$ is replaced with the edge $\topy(\w) - \w_{\LB}$;
            \item The unique edge $\w - \bottomy(\w)$ is replaced with the edge $\w_{\TR} - \bottomy(\w)$;
            \item The path $\w_{\LB} - \w_{\Mid} - \w_{\TR}$ is added.
        \end{itemize}
        \item If a black grid vertex $\w$ from $\CanInter(r\ ;\ a_i - b_i)$ is \twosplit (\autoref{fig:split-edge-yes-U}), then
        \begin{itemize}
            \item The unique edge $\topy(\w) - \w$ is replaced with the edge $\topy(\w) - \w_{\LB}$;
            \item The unique edge $\w - \bottomy(\w)$ is replaced with the edge $\w_{\TR} - \bottomy(\w)$;
            \item The path $\w_{\LB} - \w_{\Ver} - \w_{\TR}$ is added.
        \end{itemize}
    \end{itemize}

    \label{def:verti-canonical-Gedge-U}
\end{definition}
\medskip

\begin{definition}
    \textbf{(Image of a horizontal canonical path from \GintU in \GedgeU)}
    Fix a $j\in [k]$ and $r\in [N]$. For each $\CanInterU(r\ ;\ c_j - d_j)$ path $R$ in \GintU, we define an image of R as follows
    \begin{itemize}
        \item The first and last \magenta{magenta} edges are unchanged.
        \item If a black grid vertex $\w$ from $\CanInterU(r\ ;\ c_j - d_j)$ is \onesplit (\autoref{fig:split-edge-not-U}), then
        \begin{itemize}
            \item The unique edge $\lefty(\w) - \w$ is replaced with the edge $\lefty(\w) - \w_{\LB}$;
            \item The unique edge $\w - \righty(\w)$ is replaced with the edge $\w_{\TR} - \righty(\w)$;
            \item The path $\w_{\LB} - \w_{\Mid} - \w_{\TR}$ is added.
        \end{itemize}
        \item If a black grid vertex $\w$ from $\CanInterU(r\ ;\ c_j - d_j)$ is \twosplit (\autoref{fig:split-edge-yes-U}), then
        \begin{itemize}
            \item The unique edge $\lefty(\w) - \w$ is replaced with the edge $\lefty(\w) - \w_{\LB}$;
            \item The unique edge $\w - \righty(\w)$ is replaced with the edge $\w_{\TR} - \righty(\w)$;
            \item Either the edges $\w_{\LB} - \w_{\Hor} - \w_{\TR}$ or $\w_{\LB} - \w_{\Ver} - \w_{\TR}$ are added.
        \end{itemize}
    \end{itemize}
    \label{def:hor-image-of-a-path-G-edge-U}
\end{definition}
\medskip

\begin{definition}
    \textbf{(Image of a vertical canonical path from \GintU in \GedgeU)}
    Fix a $i\in [k]$ and $r\in [N]$. For each $\CanInterU(r\ ;\ a_i - b_i)$ path $R$ in \GintU, we define an image of R as follows
    \begin{itemize}
        \item The first and last \magenta{magenta} edges are unchanged.
        \item If a black grid vertex $\w$ from $\CanInterU(r\ ;\ a_i - b_i)$ is \onesplit (\autoref{fig:split-edge-not-U}), then
        \begin{itemize}
            \item The unique edge $\lefty(\w) - \w$ is replaced with the edge $\lefty(\w) - \w_{\LB}$;
            \item The unique edge $\w - \righty(\w)$ is replaced with the edge $\w_{\TR} - \righty(\w)$;
            \item The path $\w_{\LB} - \w_{\Mid} - \w_{\TR}$ is added.
        \end{itemize}
        \item If a black grid vertex $\w$ from $\CanInterU(r\ ;\ a_i - b_i)$ is \twosplit (\autoref{fig:split-edge-yes-U}), then
        \begin{itemize}
            \item The unique edge $\lefty(\w) - \w$ is replaced with the edge $\lefty(\w) - \w_{\LB}$;
            \item The unique edge $\w - \righty(\w)$ is replaced with the edge $\w_{\TR} - \righty(\w)$;
            \item Either the edges $\w_{\LB} - \w_{\Hor} - \w_{\TR}$ or $\w_{\LB} - \w_{\Ver} - \w_{\TR}$ are added.
        \end{itemize}
    \end{itemize}
    \label{def:vert-image-of-a-path-G-edge-U}
\end{definition}
\medskip

Note that a single path, $R$, in \GintU can have several images in \GedgeU. This is because for every black vertex on $R$ that is \twosplit there are two choices of sub-path to add: either the path $\w_{\LB} - \w_{\Hor} - \w_{\TR}$ or the path $\w_{\LB} - \w_{\Ver} - \w_{\TR}$.

\medskip
\begin{remark}
    \label{rmk:vert-cost-1-edge-undir}
    \textbf{(Reducing the cost of vertices in \GedgeU)}
    Here we outline why the reduction of costs as described in \autoref{rmk:vert-cost-1-undir} can also be applied to \GedgeU whilst still preserving the properties of its own canonical paths (\autoref{def:hori-canonical-Gedge-U} and \autoref{def:verti-canonical-Gedge-U}) and its images (\autoref{def:hor-image-of-a-path-G-edge-U} and \autoref{def:vert-canonical-G-int-U}). Observe, also, that this process ensures that every vertex in \GedgeU has maximum degree of $4$.

    The splitting operation applied to \GintU in order to obtain \GedgeU (\autoref{def:splitting-operation-edge-U}) modifies only the non-terminal vertices of \GintU and thus \GedgeU can only differ from \GintU in its non-terminal vertices. The cost reduction in \autoref{rmk:vert-cost-1-undir} on the other hand only modifies terminal vertices, so we see the same constant increase of $2kn$ in the cost of every canonical path (or image thereof) for every set of vertices in $\{A,B,C,D\}$.
\end{remark}
\medskip

The following two lemmas (\autoref{lem:horizontal-canonical-is-shortest-G-edge-U} and~\autoref{lem:vertical-canonical-is-shortest-G-edge-U}) analyze the structure of shortest paths between terminal pairs in \GedgeU. First, we define the \emph{image} of a path from \GintU in the graph \GedgeU.
\medskip

\begin{lemma}
    \normalfont
    Let $j\in [k]$. The shortest paths in \GedgeU satisfy the following two properties:
    \begin{itemize}
        \item[(i)] For each $r\in [N]$, the horizontal canonical path $\CanEdgeU(r\ ;\ c_j - d_j)$ is a shortest $c_j - d_j$ path in \GedgeU.
        \item[(ii)] If $P$ is a shortest $c_j - d_j$ path in \GedgeU, then $P$ must be an image (\autoref{def:hor-image-of-a-path-G-edge-U}) of the path $\CanInterU(\ell\ ;\ c_j - d_j)$ for some $\ell\in [N]$.

    \end{itemize}
    \label{lem:horizontal-canonical-is-shortest-G-edge-U}
\end{lemma}
\begin{proof}
    The proof of this lemma is similar to that of \GintU in \autoref{lem:horizontal-canonical-is-shortest-G-int-U}, with some minor observational changes. Note that every path in \GintU contains only \green{green} and black vertices. The splitting operation (\autoref{def:splitting-operation-edge-U}) applied to each black vertex of \GintU has the following property: if a path $Q$ contains a black vertex $\w$ in \GintU, then in the corresponding path in \GedgeU this vertex $\w$ is \textbf{always replaced by three black vertices}, each with a cost to visit of~$1$, viz.
    \begin{itemize}
        \item If $\w$ is \onesplit (\autoref{fig:split-edge-not-U}), then it is replaced in $Q$ the three vertices $\w_{\LB}, \w_{\Mid}, \w_{\TR}$.
        \item If $\w$ is \twosplit (\autoref{fig:split-edge-yes-U}), then it is replaced in $Q$ either by the three vertices $\w_{\LB}, \w_{\Hor}, \w_{\TR}$ or the three vertices $\w_{\LB}, \w_{\Ver}, \w_{\TR}$.
    \end{itemize}
    Therefore, if a path $Q$ incurs a cost of $\alpha$ from visiting \green{green} vertices and a cost of $\beta$ from visiting black vertices in \GintU, then the corresponding path in \GedgeU incurs a cost of  $\alpha$ from visiting \green{green} vertices and $3\beta$ from black vertices. The proof of the first part of the lemma now follows from~\autoref{lem:horizontal-canonical-is-shortest-G-int-U}(i),~\autoref{def:splitting-operation-edge-U} and~\autoref{def:hori-canonical-Gedge-U}. The proof of the second part of the lemma follows from~\autoref{lem:horizontal-canonical-is-shortest-G-int-U}(ii)'s argument that it cannot take an edge that modifies the $y$-coordinate, along with~\autoref{def:splitting-operation-edge-U} and~\autoref{def:hor-image-of-a-path-G-edge-U}.
\end{proof}
\medskip
The proof of the next lemma is very similar to that of \autoref{lem:horizontal-canonical-is-shortest-G-edge-U}, and we skip repeating the details.

\medskip
\begin{lemma}
    \normalfont
    Let $i\in [k]$. The shortest paths in \GedgeU satisfy the following two properties:
    \begin{itemize}
        \item[(i)] For each $r\in [N]$, the vertical canonical path $\CanEdgeU(r\ ;\ a_i - b_i)$ is a shortest $a_i - b_i$ path in \GedgeU.
        \item[(ii)] If $P$ is a shortest $a_i - b_i$ path in \GedgeU, then $P$ must be an image (\autoref{def:vert-image-of-a-path-G-edge-U}) of the path $\CanInterU(\ell\ ;\ a_i - b_i)$ for some $\ell\in [N]$.
    \end{itemize}
    \label{lem:vertical-canonical-is-shortest-G-edge-U}
\end{lemma}
\medskip

\subsection{
    \texorpdfstring{\underline{Completeness}: $G$ has a $k$-clique $\Rightarrow$ All pairs in the instance $(U_{\edge}, \mathcal{T})$ of Undirected-$2k$-\EDSP can be satisfied}
    {Completeness: G has a k-clique -> All pairs in the instance (Uedge,T) of Undirected-2k-\EDSP can be satisfied}}
\label{sec:clique-to-2kedsp-U}

In this section, we show that if the instance $G$ of \kclique has a solution then the instance $(U_{\edge}, \mathcal{T})$ of $2k$-\EDSP also has a solution.

Suppose the instance $G=(V,E)$ of \kclique has a clique $X=\{v_{\gamma_1}, v_{\gamma_2}, \ldots, v_{\gamma_k} \}$ of size $k$. Let $Y=\{\gamma_1, \gamma_2, \ldots, \gamma_k\}\in [N]$. Now for each $i\in [k]$ we choose the path as follows:
\begin{itemize}
    \item The path $R_i$ to satisfy $a_i - b_i$ is chosen to be the horizontal canonical path  $\CanEdgeU(\gamma_i \ ;\ a_i - b_i)$ described in~\autoref{def:hori-canonical-Gedge-U}.
    \item The path $T_i$ to satisfy $c_i - d_i$ is chosen to be vertical canonical path  $\CanEdgeU(\gamma_i \ ;\ c_i - d_i)$ described in~\autoref{def:verti-canonical-Gedge-U}.
\end{itemize}

Now we show that the collection of paths given by $\mathcal{Q}:=\{R_1, R_2, \ldots, R_k, T_1, T_2, \ldots, T_K\}$ forms a solution for the instance  $(U_{\edge}, \mathcal{T})$ of Undirected-$2k$-\EDSP via the following two lemmas which argue being shortest for each terminal pair and pairwise edge-disjointness respectively:

\medskip
\begin{lemma}
    \normalfont
    For each $i \in [k]$, the path $R_i$ (resp. $T_i$) is a shortest $a_i - b_i$ (resp. $c_i - d_i$) path in \Gedge.
    \label{lem:completeness-edsp-shortest-U}
\end{lemma}
\begin{proof}
    Fix any $i\in [k]$.~\autoref{lem:horizontal-canonical-is-shortest-G-edge-U}(i) implies that $T_i$ is shortest $c_i - d_i$ path in \GedgeU. \autoref{lem:vertical-canonical-is-shortest-G-edge-U}(i) implies that $R_i$ is shortest $a_i - b_i$ path in \GedgeU.

\end{proof}
\medskip

Before proving~\autoref{lem:completeness-edsp-disjoint-U}, we first set up notation for some special sets of vertices in \GedgeU which helps to streamline some of the subsequent proofs.

\medskip
\begin{definition}
    \label{def:horizontal-vertical-sets-in-Gedge-U}
    \textbf{(horizontal \& vertical levels in \GedgeU)}
    For each $(i,j)\in [k]\times [k]$, let $U_{i,j}^{\edgesplitt}$ to be the graph obtained by applying the splitting operation (\autoref{def:splitting-operation-edge-U}) to each vertex of $U_{i,j}$. For each $j\in [k]$, we define the following set of vertices:
    \begin{equation}
        \begin{aligned}
            \HorizontalEdge(j) &= \{ c_j, d_j \} \cup \left( \bigcup_{i=1}^{k} V(U_{i,j}^{\edgesplitt})\right)\ \quad\\
            \quad \VerticalEdge(j) &= \{ a_j, b_j \} \cup \left( \bigcup_{i=1}^{k} V(U_{j,i}^{\edgesplitt})\right)
        \end{aligned}\label{eqn:eqn:hori-verti-sets-Gedge-U}
    \end{equation}

\end{definition}
\medskip

The next lemma shows that any two paths from $\mathcal{Q}$ are edge-disjoint.
\medskip

\begin{lemma}
    \normalfont
    Let $P\neq P'$ be any pair of paths from the collection  $\mathcal{Q}=\{R_1, R_2, \ldots, R_k, T_1, T_2, \ldots, T_K\}$. Then $P$ and $P'$ are edge-disjoint.
    \label{lem:completeness-edsp-disjoint-U}
\end{lemma}
\begin{proof}

    By~\autoref{def:horizontal-vertical-sets-in-Gedge-U}, it follows that every edge of the path $R_i$ has both endpoints in $\VerticalEdge(i)$ for every $i\in [k]$. Since $\VerticalEdge(i) \cap \VerticalEdge(i')=\emptyset$ for every $1\leq i\neq i'\leq k$, it follows that the collection of paths $\{R_1, R_2, \ldots, R_k\}$ are pairwise edge-disjoint.

    By~\autoref{def:horizontal-vertical-sets-in-Gedge-U}, it follows that every edge of the path $T_j$ has both endpoints in $\HorizontalEdge(j)$ for every $j\in [k]$. Since  $\HorizontalEdge(j) \cap \HorizontalEdge(j')=\emptyset$ for every $1\leq j\neq j'\leq k$, it follows that the collection of paths $\{T_1, T_2, \ldots, T_k\}$ are pairwise edge-disjoint.

    It remains to show that every pair of paths which contains one path from $\{R_1, R_2, \ldots, R_k\}$ and other path from $\{T_1, T_2, \ldots, T_k\}$ are edge-disjoint.
    \begin{claim}
        \label{clm:reduction-edge-disjoint-U}
        \normalfont
        For each $(i,j)\in [k]\times [k]$, the paths $R_i$ and $T_j$ are edge-disjoint in $U_{\edge}$.
    \end{claim}

    \begin{proof}
        Fix any $(i,j)\in [k]\times [k]$. First we argue that the vertex $\w_{i,j}^{\gamma_i, \gamma_j}$ is \twosplit, i.e., $(\gamma_i, \gamma_j)\in S_{i,j}$:
        \begin{itemize}
            \item If $i=j$ then $\gamma_i = \gamma_j$ and hence by~\autoref{eqn:clique-to-gt-reduction-U} we have $(\gamma_i, \gamma_j)\in S_{i,j}$
            \item If $i\neq j$, then $v_{\gamma_i} - v_{\gamma_j}\in E(G)$ since $X$ is a clique. Again, by~\autoref{eqn:clique-to-gt-reduction-U} we have $(\gamma_i, \gamma_j)\in S_{i,j}$.
        \end{itemize}
        Hence, by~\autoref{def:splitting-operation-edge-U}, it follows that the vertex $\w_{i,j}^{\gamma_i, \gamma_j}$ is \twosplit.

        By the construction of \GintU (\autoref{fig:mainUndir}) and definitions of canonical paths (\autoref{def:hori-canonical-G-int-U} and~\autoref{def:vert-canonical-G-int-U}), it is easy to verify that any pair of horizontal canonical path and vertical canonical path in \GintU are edge-disjoint and have only one vertex in common.

        By the splitting operation (\autoref{def:splitting-operation-edge-U}) and definitions of the paths $R_i$ (\autoref{def:verti-canonical-Gedge-U}) and $T_j$ (\autoref{def:hori-canonical-Gedge-U}), it follows that the only common edges between $R_i$ and $T_j$ must be from paths in \GedgeU that start at $\w_{i,j,\LB}^{\gamma_i,\gamma_j}$ and end at $\w_{i,j,\TR}^{\gamma_i,\gamma_j}$. Since $\w_{i,j}^{\gamma_i, \gamma_j}$ is \twosplit, we have
        \begin{itemize}
            \item By~\autoref{def:verti-canonical-Gedge-U}, the unique $\w_{i,j,\LB}^{\gamma_i,\gamma_j} - \w_{i,j,\TR}^{\gamma_i,\gamma_j}$ sub-path of $R_i$ is $\w_{i,j,\LB}^{\gamma_i,\gamma_j} - \w_{i,j,\Ver}^{\gamma_i,\gamma_j} - \w_{i,j,\TR}^{\gamma_i,\gamma_j}$.
            \item By~\autoref{def:hori-canonical-Gedge-U}, the unique $\w_{i,j,\LB}^{\gamma_i,\gamma_j} - \w_{i,j,\TR}^{\gamma_i,\gamma_j}$ sub-path of $T_i$ is $\w_{i,j,\LB}^{\gamma_i,\gamma_j} - \w_{i,j,\Hor}^{\gamma_i,\gamma_j} - \w_{i,j,\TR}^{\gamma_i,\gamma_j}$.
        \end{itemize}
        Hence, it follows that $R_i$ and $T_j$ are edge-disjoint.
    \end{proof}
    This concludes the proof of~\autoref{lem:completeness-edsp-disjoint-U}.
\end{proof}
\medskip

\noindent From~\autoref{lem:completeness-edsp-shortest-U} and~\autoref{lem:completeness-edsp-disjoint-U}, it follows that the collection of paths given by $\mathcal{Q}=\{R_1, R_2, \ldots, R_k,$ $T_1, T_2, \ldots, T_K\}$ forms a solution for the instance $(U_{\edge}, \mathcal{T})$ of Undirected-$2k$-\EDSP.

\subsection{
    \texorpdfstring{\underline{Soundness}: $(\frac{1}{2} +\epsilon)$-fraction of the pairs in the instance $(U_{\edge}, \mathcal{T})$ of Undirected-$2k$-\EDSP can be satisfied $\Rightarrow$ $G$ has a clique of size $\geq 2\epsilon \cdot k$}
    {Soundness: (1/2 + theta)-fraction of the pairs in the instance (Uedge,T) of Undirected-2k-\EDSP can be satisfied -> G has a clique of size >= 2 x theta x k}
}
\label{sec:2kedsp-to-clique-U}

In this section we show that if at least $(\frac{1}{2} +\epsilon)$-fraction of the $2k$ pairs from the instance $(U_{\edge}, \mathcal{T})$ of Undirected-$2k$-\EDSP can be satisfied then the graph $G$ has a clique of size $2\epsilon \cdot k$.

Let $\mathcal{P}$ be a collection of paths in $U_{\edge}$ which satisfies at least $(\frac{1}{2} +\epsilon)$-fraction of the $2k$ terminal pairs from the instance $(U_{\edge}, \mathcal{T})$ of Undirected-$2k$-\EDSP.

\medskip
\begin{definition}
    An index $i \in [k]$ is called \emph{good} if both the terminal pairs $a_i \leadsto b_i$ and $c_i \leadsto d_i$ are satisfied by $\mathcal{P}$.
    \label{def:good-for-edge-U}
\end{definition}
\medskip

The next lemma gives a lower bound on the number of good indices.

\medskip
\begin{lemma}
    \label{lem:good_size-U}
    \normalfont
    Let $Y \subseteq [k]$ be the set of good indices. Then $|Y| \ge 2\epsilon \cdot k$.
\end{lemma}
\begin{proof}
    If $i\in [k]$ is good then both the pairs $a_i - b_i$ and $c_i - d_i$ are satisfied by $\mathcal{P}$. Otherwise, at most one of these pairs $a_i - b_i$ and $c_i - d_i$ is satisfied. Hence, the total number of satisfied pairs is at
    most $2\cdot |Y|+1\cdot (k-|Y|)= k+|Y|$. However, we know that $\mathcal{P}$ satisfies at least
    $(\frac{1}{2} +\epsilon)\cdot |\mathcal{T}|=\left(\frac{1}{2}+\epsilon\right)\cdot 2k=k+2\epsilon\cdot k$ pairs. Hence, it follows that~$|Y| \geq
    2\epsilon \cdot k$.
\end{proof}
\medskip

\begin{lemma}
    \label{lem:good-equal-edge-U}
    \normalfont
    If $i\in [k]$ is good, then there exists $\delta_i \in [N]$ such that the two paths in $\mathcal{P}$ satisfying $a_i - b_i$ and $c_i - d_i$ in \GedgeU are images of the paths $\CanInterU(\delta_i\ ;\ a_i - b_i)$ and $\CanInterU(\delta_i\ ;\ c_i - d_i)$ from \GintU respectively.
\end{lemma}
\begin{proof}
    If $i$ is good, then by~\autoref{def:good-for-edge-U} both the pairs $a_i - b_i$ and $c_i - d_i$ are satisfied by $\mathcal{P}$. Let $P_1, P_2\in \mathcal{P}$ be the paths that satisfy the terminal pairs $(a_i, b_i)$ and $(c_i, d_i)$ respectively.
    Since $P_1$ is a shortest $a_i - b_i$ path in \GedgeU, by~\autoref{lem:vertical-canonical-is-shortest-G-edge-U}(ii) it follows that $P_1$ is an image of the vertical canonical path $\CanInterU(\alpha\ ;\ a_i - b_i)$ from \GintU for some $\alpha\in [N]$. Since $P_2$ is a shortest $c_i - d_i$ path in \GedgeU, by~\autoref{lem:horizontal-canonical-is-shortest-G-edge-U}(ii) it follows that $P_2$ is an image of the horizontal canonical path $\CanInterU(\beta\ ;\ c_i - d_i)$ from \GintU for some $\beta\in [N]$.

    Using the fact that $P_1$ and $P_2$ are edge-disjoint in \GedgeU, we now claim that $\w_{i,i}^{\alpha,\beta}$ is \twosplit:
    \begin{claim}
        \normalfont
        The vertex $\w_{i,i}^{\alpha,\beta}$ is \twosplit by the splitting operation of~\autoref{def:splitting-operation-edge-U}.
        \label{clm:must-be-two-split-edge-i-i-U}
    \end{claim}
    \begin{proof}
        By~\autoref{def:splitting-operation-edge-U}, every black vertex of \GintU is either \onesplit or \twosplit. If $\w_{i,i}^{\alpha,\beta}$ was \onesplit (\autoref{fig:split-edge-not-U}), then by~\autoref{def:hor-image-of-a-path-G-edge-U} and \autoref{def:vert-image-of-a-path-G-edge-U} the path $\w_{i,i,\LB}^{\alpha,\beta} - w_{i,i,\Mid}^{\alpha,\beta} - w_{i,i,\TR}^{\alpha,\beta}$ belongs to both the paths $P_1$ and $P_2$ contradicting the fact that they are edge-disjoint.
    \end{proof}
    By~\autoref{clm:must-be-two-split-edge-i-i-U}, we know that the vertex $\w_{i,i}^{\alpha,\beta}$ is \twosplit. Hence, from~\autoref{eqn:clique-to-gt-reduction-U} and~\autoref{def:splitting-operation-edge-U}, it follows that $\alpha=\beta$ which concludes the proof of the lemma.

\end{proof}
\medskip

\begin{lemma}
    \label{lem:good_edges-U}
    \normalfont
    If both $i,j \in [k]$ are good and $i \neq j$, then $v_{\delta_i}-v_{\delta_j} \in E(G)$.
\end{lemma}
\begin{proof}
    Since $i$ and $j$ are good, by~\autoref{def:good-for-edge-U}, there are paths $Q_1, Q_2 \in \mathcal{P}$ satisfying the pairs $(a_i, b_i), (c_j, d_j)$ respectively. By~\autoref{lem:good-equal-edge-U}, it follows that
    \begin{itemize}
        \item $Q_1$ is an image of the path $\CanInterU(\delta_i\ ;\ a_i - b_i)$ from \GintU.
        \item $Q_2$ is an image of the path $\CanInterU(\delta_j\ ;\ c_j - d_j)$ from \GintU.
    \end{itemize}

    Using the fact that $Q_1$ and $Q_2$ are edge-disjoint in \GedgeU, we now claim that $\w_{i,j}^{\delta_i,\delta_j}$ is \twosplit:
    \begin{claim}
        \normalfont
        The vertex $\w_{i,j}^{\delta_i,\delta_j}$ is \twosplit by the splitting operation of~\autoref{def:splitting-operation-edge-U}.
        \label{clm:must-be-two-split-edge-i-j-U}
    \end{claim}
    \begin{proof}
        By~\autoref{def:splitting-operation-edge-U}, every black vertex of \GintU is either \onesplit or \twosplit. If $\w_{i,j}^{\delta_j,\delta_j}$ was \onesplit (\autoref{fig:split-edge-not-U}), then by~\autoref{def:hor-image-of-a-path-G-edge-U} and \autoref{def:vert-image-of-a-path-G-edge-U} the path $\w_{i,j,\LB}^{\delta_i,\delta_j} - w_{i,j,\Mid}^{\delta_i,\delta_j} - w_{i,j,\TR}^{\delta_i,\delta_j}$ belongs to both the paths $Q_1$ and $Q_2$ contradicting the fact that they are edge-disjoint.
    \end{proof}
    By~\autoref{clm:must-be-two-split-edge-i-j-U}, we know that the vertex $\w_{i,j}^{\delta_i,\delta_j}$ is \twosplit. Since $i\neq j$, from~\autoref{eqn:clique-to-gt-reduction-U} and~\autoref{def:splitting-operation-edge-U}, it follows that $v_{\delta_i}-v_{\delta_j}\in E(G)$ which concludes the proof of the lemma.

\end{proof}
\medskip

From~\autoref{lem:good_size-U} and~\autoref{lem:good_edges-U}, it follows that the set $X:=\{v_{\delta_i}\ : i\in Y\}$ is a clique of size $\geq (2\epsilon)k$ in $G$.

\subsection{Proof of \autoref{thm:inapprox-edge-result-U} and \autoref{thm:hardness-edge-result-U}}
\label{sec:proof-of-main-theorem-edge-U}

Finally we are ready to prove \autoref{thm:inapprox-edge-result-U} and \autoref{thm:hardness-edge-result-U}, which are restated below.
\medskip

\apxundirectededgethm*
\exactundirectededgethm*
\begin{proof}{\textbf{\autoref{thm:hardness-edge-result-U}}}

    Given an instance $G$ of \kclique, we can use the construction from \autoref{sec:construction-of-Gedge-U} to build an instance $(U_{\edge}, \mathcal{T})$ of Undirected-$2k$-\EDSP such that $U_{\edge}$ is planar (\autoref{clm:GedgeU-is-planar}). The graph \GedgeU has $n=O(N^2 k^2)$ vertices (\autoref{clm:size-of-G-edge-U}), and it is easy to observe that it can be constructed from $G$ (via first constructing \GintU) in $\poly(N,k)$ time.

    It is known that \kclique is W[1]-hard parameterized by $k$, and under ETH cannot be solved in $f(k)\cdot N^{o(k)}$ time for any computable function $f$~\cite{chen-hardness}. Combining the two directions from~\autoref{sec:2kedsp-to-clique-U} (with $\epsilon = 0.5$) and~\autoref{sec:clique-to-2kedsp-U}  we obtain a parameterized reduction from an instance $(G,k)$ of \kclique with $N$ vertices to an instance $(U_{\edge}, \mathcal{T})$ of Undirected-$2k$-\EDSP where \GedgeU is a planar DAG (\autoref{clm:GedgeU-is-planar}) and has $O(N^2 k^2)$ vertices (\autoref{clm:size-of-G-edge-U}). As a result, it follows that Undirected-$k$-\EDSP on planar graphs is W[1]-hard parameterized by number $k$ of terminal pairs, and under ETH cannot be solved in $f(k)\cdot n^{o(k)}$ time where $f$ is any computable function and $n$ is the number of vertices.
\end{proof}

\begin{proof}{\textbf{\autoref{thm:inapprox-edge-result-U}}}

    Let $\delta$ and $r_0$ be the constants from~\autoref{thm:cli_inapprox}. Fix any constant $\epsilon\in (0,1/2]$. Set $\zeta = \dfrac{\delta \epsilon}{2}$ and $k=\max \Big\{\dfrac{1}{2\zeta} , \dfrac{r_0}{2\epsilon}\Big\}$.

    Suppose to the contrary that there exists an algorithm $\mathbb{A}_{\EDSP}$ running in $f(k)\cdot n^{\zeta k}$ time (for some computable function $f$) which given an instance of Undirected-$k$-\EDSP with $n$ vertices can distinguish between the following two cases:
    \begin{itemize}
        \item[(1)] All $k$ pairs of the Undirected-$k$-\EDSP instance can be satisfied
        \item[(2)] The max number of pairs of the Undirected-$k$-\EDSP instance that can be satisfied is less than $(\frac{1}{2}+\epsilon)\cdot k$
    \end{itemize}
    We now design an algorithm $\mathbb{A}_{\Clique}$ that contradicts~\autoref{thm:cli_inapprox} for the values $q=k$ and $r=(2\epsilon)k$. Given an instance of $(G,k)$ of \kclique with $N$ vertices, we apply the reduction from~\autoref{sec:construction-of-Gedge-U} to construct an instance $(U_{\edge}, \mathcal{T})$ of Undirected-$2k$-\EDSP where \GedgeU has $n=O(N^2 k^2)$ vertices (\autoref{clm:size-of-G-edge-U}). It is easy to see that this reduction takes $O(N^2 k^2)$ time as well. We now show that the number of pairs which can be satisfied from the Undirected-$2k$-\EDSP instance is related to the size of the max clique in $G$:
    \begin{itemize}
        \item If $G$ has a clique of size $q=k$, then by~\autoref{sec:clique-to-2kedsp-U} it follows that all $2k$ pairs of the instance $(U_{\edge},\mathcal{T})$ of Undirected-$2k$-\EDSP can be satisfied.
        \item If $G$ does not have a clique of size $r=2\epsilon k$, then we claim that the max number of pairs in $\mathcal{T}$ that can be satisfied is less than $(\frac{1}{2}+\epsilon)\cdot 2k$. This is because if at least $(\frac{1}{2}+\epsilon)$-fraction of pairs in $\mathcal{T}$ could be satisfied then by~\autoref{sec:2kedsp-to-clique-U} the graph $G$ would have a clique of size $\geq (2\epsilon) k=r$.
    \end{itemize}
    Since the algorithm $\mathbb{A}_{\EDSP}$ can distinguish between the two cases of all $2k$-pairs of the instance $(U_{\edge}, \mathcal{T})$ can be satisfied or only less than $(\frac{1}{2}+\epsilon)\cdot2k$ pairs can be satisfied, it follows that $\mathbb{A}_{\Clique}$ can distinguish between the cases $\Clique(G)\geq q$ and $\Clique(G)<r$.

    The running time of the algorithm $\mathbb{A}_{\Clique}$ is the time taken for the reduction from~\autoref{sec:construction-of-Gedge-U} (which is $O(N^2 k^2)$) plus the running time of the algorithm $\mathbb{A}_{\EDSP}$ which is $f(2k)\cdot n^{\zeta\cdot 2k}$. It remains to show that this can be upper bounded by $g(q,r)\cdot N^{\delta r}$ for some computable function $g$:
    \begin{align*}
        & O(N^2 k^2) + f(2k)\cdot n^{\zeta\cdot 2k}\\
        &\leq  c\cdot N^2 k^2 + f(2k)\cdot d^{\zeta\cdot 2k}\cdot (N^2 k^2)^{\zeta\cdot 2k} \tag{for some constants $c,d\geq 1$: this follows since $n=O(N^2 k^2)$}\\
        &\leq c\cdot N^2 k^2 + f'(k)\cdot N^{2\zeta\cdot 2k} \tag{where $f'(k)=f(2k)\cdot d^{\zeta\cdot 2k}\cdot k^{2\zeta\cdot 2k}$}\\
        &\leq 2c\cdot f'(k)\cdot N^{2\zeta\cdot 2k} \tag{since $4\zeta k\geq 2$ implies $f'(k)\geq k^2$ and $N^{2\zeta\cdot 2k}\geq N^2$}\\
        &= 2c\cdot f'(k) \cdot N^{\delta r} \tag{since $\zeta=\frac{\delta \epsilon}{2}$ and $r=(2\epsilon)k$}
    \end{align*}
    Hence, we obtain a contradiction to~\autoref{thm:cli_inapprox} with $q=k, r=(2\epsilon)k$ and $g(k)=2c\cdot f'(k)=2c\cdot f(2k)\cdot d^{\zeta\cdot 2k}\cdot k^{2\zeta\cdot 2k}$.
\end{proof}

	\section{Lower bounds for  \texorpdfstring{$k$}{k}-\VDSP on \texorpdfstring{$1$}{1}-planar graphs}
\label{sec:fpt-inapprox-vdsp-undir}

The goal of this section is to prove lower bounds on the running time of exact (\autoref{thm:hardness-vertex-result-U}) and approximate (\autoref{thm:inapprox-vertex-result-U}) algorithms for the Undirected-$k$-\VDSP problem. We have already seen the first part of the reduction (\autoref{sec:construction-of-GintU}) from \kclique resulting in the construction of the intermediate graph \GintU.~\autoref{sec:construction-of-Gvert-U} describes the next part of the reduction which edits the intermediate \GintU to obtain the final graph \GvertU. This corresponds to the ancestry of the fourth leaf in~\autoref{fig:flowchart}. The characterization of shortest paths between terminal pairs in \GvertU is given in~\autoref{sec:characterizing-shortest-in-G-vertex-U}. The completeness and soundness of the reduction from \kclique to Undirected-$2k$-\VDSP are proven
in~\autoref{sec:clique-to-2kvdsp-U} and~\autoref{sec:2kvdsp-to-clique-U}, respectively. Finally, we state our final results in~\autoref{sec:proof-of-main-theorem-vertex-U} allowing us to prove~\autoref{thm:hardness-vertex-result-U} and~\autoref{thm:inapprox-vertex-result-U}.

\subsection{Obtaining the graph \GvertU from \GintU via the splitting operation}
\label{sec:construction-of-Gvert-U}

Observe in \autoref{fig:mainUndir} that every black grid vertex in \GintU has degree exactly four, and these four neighbors are named as per~\autoref{def:lefty-right-topy-bottomy-U}.
We now define the splitting operation which allows us to obtain the graph \GvertU from the graph \GintU constructed in~\autoref{sec:construction-of-GintU}.

\medskip
\begin{definition}
    \normalfont
    \textbf{(splitting operation to obtain \GvertU from \GintU)} For each $i,j\in [k]$ and each $q,\ell\in [N]$
    \begin{itemize}
        \item If $(q,\ell)\in S_{i,j}$ then we \vertsplit (\autoref{fig:split-vertex-yes-U}) the vertex $\w_{i,j}^{q,\ell}$ into \textbf{two distinct} vertices $\w_{i,j,\Hor}^{q,\ell}$, and $\w_{i,j,\Ver}^{q,\ell}$.

        \item Otherwise, if $(q,\ell)\notin S_{i,j}$, then the vertex $\w_{i,j}^{q,\ell}$ is \notsplit (\autoref{fig:split-vertex-not-U}) and we define $\w_{i,j,\Hor}^{q,\ell} = \w_{i,j,\Ver}^{q,\ell}$.

    \end{itemize}
    In either case, the four edges (\autoref{def:lefty-right-topy-bottomy-U-vert}) incident on $\w_{i,j}^{q,\ell}$ are modified as follows:
    \begin{itemize}
        \item Replace the edge $\lefty(\w_{i,j}^{q,\ell})- \w_{i,j}^{q,\ell}$ by the edge $\lefty(\w_{i,j}^{q,\ell})- \w_{i,j,\Hor}^{q,\ell}$
        \item Replace the edge $\bottomy(\w_{i,j}^{q,\ell})- \w_{i,j}^{q,\ell}$ by the edge $\bottomy(\w_{i,j}^{q,\ell})- \w_{i,j,\Ver}^{q,\ell}$
        \item Replace the edge $\w_{i,j}^{q,\ell}- \righty(\w_{i,j}^{q,\ell})$ by the edge $\w_{i,j,\Hor}^{q,\ell}- \righty(\w_{i,j}^{q,\ell})$
        \item Replace the edge $\w_{i,j}^{q,\ell}- \topy(\w_{i,j}^{q,\ell})$ by the edge $\w_{i,j,\Ver}^{q,\ell}- \topy(\w_{i,j}^{q,\ell})$
    \end{itemize}
    \label{def:splitting-operation-vertex-U}
\end{definition}
\medskip

\begin{figure}[hbt!]
\centering
\begin{tikzpicture}[
vertex/.style={circle, draw=black, fill=black, text width=1.5mm, inner sep=0pt},
scale=0.8]
\node[vertex, label=above right:\footnotesize{$\w_{i,j}^{q,\ell}$}] (v) at (0,0) {} ;
\node[vertex, label=above:\footnotesize{$\lefty(\w_{i,j}^{q,\ell})$}] (l) at (-2,0) {};
\node[vertex, label=below:\footnotesize{$\righty(\w_{i,j}^{q,\ell})$}] (r) at (2,0) {};
\node[vertex, label=below:\footnotesize{$\bottomy(\w_{i,j}^{q,\ell})$}] (b) at (0,-2) {};
\node[vertex, label=above:\footnotesize{$\topy(\w_{i,j}^{q,\ell})$}] (t) at (0,2) {};
\draw[ultra thick] (v) -- (t);
\draw[ultra thick] (v) -- (r);
\draw[ultra thick] (l) -- (v);
\draw[ultra thick] (b) -- (v);

\draw[orange,double, ultra thick,->] (3,0) -- node[above=3mm, draw=none, fill=none, rectangle] {\vertsplit} (5,0);

\node[vertex, label=above:\footnotesize{$\w_{i,j,\Hor}^{q,\ell}$}] (vh) at (7.5,0) {} ;
\node[vertex, label=right:\footnotesize{$\w_{i,j,\Ver}^{q,\ell}$}] (vv) at (9.5,1) {} ;

\node[vertex, label=below:\footnotesize{$\lefty(\w_{i,j}^{q,\ell})$}] (l) at (6,0) {};
\node[vertex, label=below:\footnotesize{$\righty(\w_{i,j}^{q,\ell})$}] (r) at (10,0) {};
\node[vertex, label=below:\footnotesize{$\bottomy(\w_{i,j}^{q,\ell})$}] (b) at (8,-2) {};
\node[vertex, label=above:\footnotesize{$\topy(\w_{i,j}^{q,\ell})$}] (t) at (8,2) {};
\draw[dotted, ultra thick] (vv) -- (t);
\draw[dotted, ultra thick] (vh) -- (r);
\draw[dotted, ultra thick] (l) -- (vh);
\draw[dotted, ultra thick] (b) -- (vv);
\end{tikzpicture}

\caption{The \vertsplit operation for the vertex $\w_{i,j}^{q,\ell}$ when
$(q,\ell) \in S_{i,j}$.  The intent is that the horizontal path $\lefty(\w_{i,j}^{q,\ell})- \w_{i,j}^{q,\ell} - \righty(\w_{i,j}^{q,\ell})$ and the vertical path $\bottomy(\w_{i,j}^{q,\ell})- \w_{i,j}^{q,\ell} - \topy(\w_{i,j}^{q,\ell})$ are now actually vertex-disjoint after the \vertsplit operation (but were not vertex-disjoint before since they shared the vertex $\w_{i,j}^{q,\ell}$)}
\label{fig:split-vertex-yes-U}
\end{figure}
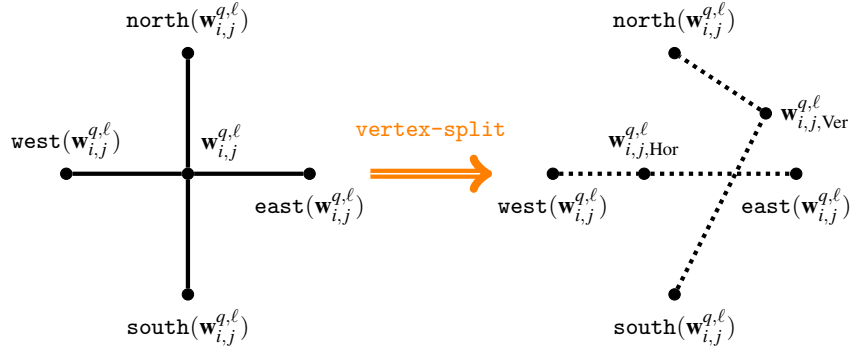
\begin{figure}[hbt!]
\centering
\begin{tikzpicture}[
vertex/.style={circle, draw=black, fill=black, text width=1.5mm, inner sep=0pt},
scale=0.85]
\node[vertex, label=above right:\footnotesize{$\w_{i,j}^{q,\ell}$}] (v) at (0,0) {} ;
\node[vertex, label=above:\footnotesize{$\lefty(\w_{i,j}^{q,\ell})$}] (l) at (-2,0) {};
\node[vertex, label=below:\footnotesize{$\righty(\w_{i,j}^{q,\ell})$}] (r) at (2,0) {};
\node[vertex, label=below:\footnotesize{$\bottomy(\w_{i,j}^{q,\ell})$}] (b) at (0,-2) {};
\node[vertex, label=above:\footnotesize{$\topy(\w_{i,j}^{q,\ell})$}] (t) at (0,2) {};
\draw[ultra thick] (v) -- (t);
\draw[ultra thick] (v) -- (r);
\draw[ultra thick] (l) -- (v);
\draw[ultra thick] (b) -- (v);

\draw[orange,double, ultra thick,->] (3,0) -- node[above=3mm, draw=none, fill=none, rectangle] {\notsplit} (5,0);

\node[vertex, label={[rotate=35,xshift=14mm,yshift=-4mm] \footnotesize{$\w_{i,j,\Hor}^{q,\ell}=\w_{i,j,\Ver}^{q,\ell}$}}] (mid) at (8,0) {} ;

{(0,0)}

\node[vertex, label=below:\footnotesize{$\lefty(\w_{i,j}^{q,\ell})$}] (l) at (6,0) {};
\node[vertex, label=below:\footnotesize{\quad $\righty(\w_{i,j}^{q,\ell})$}] (r) at (10,0) {};
\node[vertex, label=below:\footnotesize{$\bottomy(\w_{i,j}^{q,\ell})$}] (b) at (8,-2) {};
\node[vertex, label=above:\footnotesize{$\topy(\w_{i,j}^{q,\ell})$}] (t) at (8,2) {};
\draw[dotted, ultra thick] (mid) -- (t);
\draw[dotted, ultra thick] (mid) -- (r);
\draw[dotted, ultra thick] (l) -- (mid);
\draw[dotted, ultra thick] (b) -- (mid);
\end{tikzpicture}

\caption{The \notsplit operation for the vertex $\w_{i,j}^{q,\ell}$ when
$(q,\ell) \notin S_{i,j}$. The intent is that the horizontal path $\lefty(\w_{i,j}^{q,\ell})- \w_{i,j}^{q,\ell} - \righty(\w_{i,j}^{q,\ell})$ and the vertical path $\bottomy(\w_{i,j}^{q,\ell})- \w_{i,j}^{q,\ell} - \topy(\w_{i,j}^{q,\ell})$ are still not vertex-disjoint after the \notsplit operation since they share the vertex $\w_{i,j,\Hor}^{q,\ell}=\w_{i,j,\Ver}^{q,\ell}$.}
\label{fig:split-vertex-not-U}
\end{figure}

\medskip
\begin{definition}
    \normalfont
    \textbf{(defining the Undirected-$2k$-\VDSP instance)} The instance $(U_{\vertex}, \mathcal{T})$~of $2k$-\VDSP is defined as follows:
    \begin{itemize}
        \item The graph \GvertU is obtained by applying the splitting operation (\autoref{def:splitting-operation-vertex-U}) to each (black) grid vertex of $U_{\inter}$, i.e., the set of vertices given by $\bigcup_{1\leq i,j\leq k} V(U_{i,j})$.
        \item No \green{green} vertex is split in~\autoref{def:splitting-operation-vertex-U}, and hence the set of terminal pairs remains the same as defined in~\autoref{eqn:definition-of-mathcal-T-undir} and is given by $\mathcal{T}:= \big\{(a_i, b_i) : i\in [k] \big\}\cup \big\{(c_j, d_j) : j\in [k] \big\}$.
        \item We assign a cost of $1$ to each vertex present after the splitting operation (\autoref{def:splitting-operation-vertex-U}). Since each vertex of \GintU has a cost of $1$, it follows that each vertex of \GvertU also has a visit cost of $1$.
    \end{itemize}
    \label{def:G-vertex-U}
\end{definition}
\medskip

Note that the construction of \GvertU from \GintU differs from the construction of \GedgeU from \autoref{sec:construction-of-Gedge-U} only in its splitting operation.

\medskip
\begin{claim}
    \label{clm:GvertexU-is-1-planar}
    \normalfont
    \GvertU is $1$-planar\footnote{A $1$-planar graph is a graph that can be drawn in the Euclidean plane in such a way that each edge has at most one crossing point, where it crosses a single additional edge.}.
\end{claim}
\begin{proof}
    In~\autoref{clm:G-int-U-is-planar-and-dag}, we have shown that \GintU is planar. The graph \GvertU is obtained from \GintU by applying the splitting operation (\autoref{def:splitting-operation-vertex-U}) on every (black) grid vertex, i.e., every vertex from the set $\bigcup_{1\leq i,j\leq k} V(U_{i,j})$. By~\autoref{def:lefty-right-topy-bottomy-U-vert}, every vertex of \GintU that is split has at most 4 neighbours in \GintU.~\autoref{fig:split-vertex-not-U} maintains the planarity, but in~\autoref{fig:split-vertex-yes-U} we have two edges $\bottomy(\w_{i,j}^{q,\ell}) - \w_{i,j,\Ver}^{q,\ell}$ and $\w_{i,j,\Hor}^{q,\ell} - \righty(\w_{i,j}^{q,\ell})$ that cross each other at exactly one point. Since these are the only type of edges that can cross, we can draw \GvertU in the Euclidean plane in such a way that each edge has at most one crossing point, where it crosses a single additional edge. Therefore, the entire \GvertU is $1$-planar.

\end{proof}
\medskip

\begin{claim}
    \normalfont
    The number of vertices in \GvertU is $O(N^{2}k^{2})$.
    \label{clm:size-of-G-vertex-U}
\end{claim}
\begin{proof}
    The only change in going from \GintU to \GvertU is the splitting operation (\autoref{def:splitting-operation-vertex-U}). If a black grid vertex $\w$ in \GintU is \notsplit (\autoref{fig:split-vertex-not-U}) then we replace it by \textbf{one} vertex $\w_{\Ver}=\w_{\Hor}$ in \GvertU. If a black grid vertex $\w$ in \GintU is \vertsplit (\autoref{fig:split-vertex-yes-U}) then we replace it by the \textbf{two} vertices $\w_{\Hor}$ and $\w_{\Ver}$ in \GvertU. In both cases, the increase in number of vertices is only by a constant factor. The number of vertices in \GintU is $O(N^2 k^2)$ from~\autoref{clm:size-of-G-int-U}, and hence it follows that the number of vertices in \GvertU is $O(N^2 k^2)$.

\end{proof}
\medskip

\begin{definition}
    \label{def:lefty-right-topy-bottomy-U-vert}
    Recall \autoref{def:lefty-right-topy-bottomy-U}, where we defined the four neighbours of any grid vertex in \GintU. We maintain these definitions of the neighbours for each (black) grid vertex here in \GvertU.
\end{definition}
\medskip

\subsection{Characterizing shortest paths in \GvertU}
\label{sec:characterizing-shortest-in-G-vertex-U}

The goal of this section is to characterize the structure of shortest paths between terminal pairs in \GvertU. Recall (\autoref{def:splitting-operation-vertex-U}) that the set of terminal pairs is given by $\mathcal{T}:= \big\{(a_i, b_i) : i\in [k] \big\}\cup \big\{(c_j, d_j) : j\in [k] \big\}$. As in \autoref{sec:setting-up-the-U}, the length of a path is the sum of the vertex costs.

We now define canonical paths in \GvertU by adapting the definition of canonical paths (\autoref{def:hori-canonical-G-int-U} and~\autoref{def:vert-canonical-G-int-U}) in \GintU in accordance with the changes in going from \GintU to \GvertU.

\medskip
\begin{definition}
    \textbf{(horizontal canonical paths in \GvertU}
    Fix some $j\in [k]$. For each $r\in [N]$, we define $\CanVertexU(r\ ;\ c_j - d_j)$ to be the $c_j - d_j$ path in \GvertU obtained from the path $\CanInterU(r\ ;\ c_j - d_j)$ in \GintU (recall~\autoref{def:hori-canonical-G-int-U}) in the following way:
    \begin{itemize}
        \item The first and last \magenta{magenta} edges are unchanged;
        \item If a black grid vertex $\w$ from $\CanInterU(r\ ;\ c_j - d_j)$ is \notsplit (\autoref{fig:split-vertex-not-U}), then
        \begin{itemize}
            \item The unique edge $\lefty(\w) - \w$ is replaced with the edge $\lefty(\w) - \w_{\Hor}=\w_{\Ver}$;
            \item The unique edge $\w - \righty(\w)$ is replaced with the edge $\w_{\Hor}=\w_{\Ver} - \righty(\w)$;
        \end{itemize}
        \item If a black grid vertex $\w$ from $\CanInterU(r\ ;\ c_j - d_j)$ is \vertsplit (\autoref{fig:split-vertex-yes-U}), then
        \begin{itemize}
            \item The unique edge $\lefty(\w) - \w$ is replaced with the edge $\lefty(\w) - \w_{\Hor}$;
            \item The unique edge $\w - \righty(\w)$ is replaced with the edge $\w_{\Hor} - \righty(\w)$;
        \end{itemize}
    \end{itemize}

    \label{def:hori-canonical-Gvertex-U}
\end{definition}
\medskip

\begin{definition}
    \textbf{(vertical canonical paths in \GvertU)}
    Fix a $j\in [k]$. For each $r\in [N]$, we define $\CanVertexU(r\ ;\ a_j - b_j)$ to be the $a_j - b_j$ path in \GvertU obtained from the path $\CanInterU(r\ ;\ a_j - b_j)$ in \GintU (recall~\autoref{def:vert-canonical-G-int-U}) in the following way:
    \begin{itemize}
        \item The first and last \magenta{magenta} edges are unchanged.
        \item If a black grid vertex $\w$ from $\CanInterU(r\ ;\ a_j - b_j)$ is \notsplit (\autoref{fig:split-vertex-not-U}), then
        \begin{itemize}
            \item The unique edge $\topy(\w) - \w$ is replaced with the edge $\topy(\w) - \w_{\Hor}=\w_{\Ver}$;
            \item The unique edge $\w - \bottomy(\w)$ is replaced with the edge $\w_{\Hor}=\w_{\Ver} - \bottomy(\w)$;
        \end{itemize}
        \item If a black grid vertex $\w$ from $\CanInterU(r\ ;\ a_j - b_j)$ is \vertsplit (\autoref{fig:split-vertex-yes-U}), then
        \begin{itemize}
            \item The unique edge $\topy(\w) - \w$ is replaced with the edge $\topy(\w) - \w_{\Ver}$;
            \item The unique edge $\w - \bottomy(\w)$ is replaced with the edge $\w_{\Ver} - \bottomy(\w)$;
        \end{itemize}
    \end{itemize}

    \label{def:vert-canonical-Gvertex-U}
\end{definition}
\medskip

\begin{definition}
    \textbf{(Image of a horizontal canonical path from \GintU in \GvertU)}
    Fix a $j\in [k]$ and $r\in [N]$. For each $\CanInterU(r\ ;\ c_j - d_j)$ path $R$ in \GintU, we define an image of R as follows
    \begin{itemize}
        \item The first and last \magenta{magenta} edges are unchanged.
        \item If a black grid vertex $\w$ from $\CanInterU(r\ ;\ c_j - d_j)$ is \notsplit (\autoref{fig:split-vertex-not-U}), then
        \begin{itemize}
            \item The unique edge $\lefty(\w) - \w$ is replaced with the edge $\lefty(\w) - \w_{\Hor}=\w_{\Ver}$;
            \item The unique edge $\w - \righty(\w)$ is replaced with the edge $\w_{\Hor}=\w_{\Ver} - \righty(\w)$;
        \end{itemize}
        \item If a black grid vertex $\w$ from $\CanInterU(r\ ;\ c_j - d_j)$ is \vertsplit (\autoref{fig:split-vertex-yes-U}), then
        \begin{itemize}
            \item The series of edges $\lefty(\w) - \w - \righty(\w)$ is replaced with either the path $\lefty(\w) - \w_{\Ver} - \righty(\w)$ or $\lefty(\w) - \w_{\Hor} - \righty(\w)$;
        \end{itemize}
    \end{itemize}
    \label{def:hor-image-of-a-path-G-vert-U}
\end{definition}
\medskip

\begin{definition}
    \textbf{(Image of a vertical canonical path from \GintU in \GvertU)}
    Fix a $i\in [k]$ and $r\in [N]$. For each $\CanInterU(r\ ;\ a_i - b_i)$ path $R$ in \GintU, we define an image of R as follows
    \begin{itemize}
        \item The first and last \magenta{magenta} edges are unchanged.
        \item If a black grid vertex $\w$ from $\CanInterU(r\ ;\ a_i - b_i)$ is \notsplit (\autoref{fig:split-vertex-not-U}), then
        \begin{itemize}
            \item The unique edge $\topy(\w) - \w$ is replaced with the edge $\topy(\w) - \w_{\Hor}=\w_{\Ver}$;
            \item The unique edge $\w - \bottomy(\w)$ is replaced with the edge $\w_{\Hor}=\w_{\Ver} - \bottomy(\w)$;
        \end{itemize}
        \item If a black grid vertex $\w$ from $\CanInterU(r\ ;\ a_i - b_i)$ is \vertsplit (\autoref{fig:split-vertex-yes-U}), then
        \begin{itemize}
            \item The series of edges $\topy(\w) - \w - \bottomy(\w)$ is replaced with either the path  $\topy(\w) - \w_{\Ver} - \bottomy(\w)$ or $\topy(\w) - \w_{\Hor} - \bottomy(\w)$;
        \end{itemize}
    \end{itemize}
    \label{def:vert-image-of-a-path-G-vert-U}
\end{definition}
\medskip

Note that a single path, $R$, in \GintU can have several images in \GvertU. This is because for every black vertex on $R$ that is \twosplit there are two choices of sub-path to add: either the path $\w_{\LB} - \w_{\Hor} - \w_{\TR}$ or the path $\w_{\LB} - \w_{\Ver} - \w_{\TR}$.

\medskip

\begin{remark}
    \label{rmk:vert-cost-1-vert-undir}
    \textbf{(Reducing the cost of vertices in \GvertU)}
    Here we outline why the reduction of costs as described in \autoref{rmk:vert-cost-1-undir} can also be applied to \GvertU whilst still preserving the properties of its own canonical paths (\autoref{def:hori-canonical-Gvertex-U} and \autoref{def:vert-canonical-Gvertex-U}) and its images (\autoref{def:hor-image-of-a-path-G-vert-U} and \autoref{def:vert-image-of-a-path-G-vert-U}). Observe, also, that this process ensures that every vertex in \GvertU has maximum degree of $4$.

    The splitting operation applied to \GintU in order to obtain \GvertU (\autoref{def:splitting-operation-vertex-U}) modifies only the non-terminal vertices of \GintU and thus \GvertU can only differ from \GintU in its non-terminal vertices. The cost reduction in \autoref{rmk:vert-cost-1-undir} on the other hand only modifies terminal vertices, so we see the same constant increase of $2kn$ in the cost of every canonical path (or image thereof) for every set of vertices in $\{A,B,C,D\}$.
\end{remark}

The following two lemmas (\autoref{lem:horizontal-canonical-is-shortest-G-vertex-U} and~\autoref{lem:vertical-canonical-is-shortest-G-vertex-U}) analyze the structure of shortest paths between terminal pairs in \GedgeU. First, we define the \emph{image} of a path from \GintU in the graph \GedgeU.
\medskip

\begin{lemma}
    \normalfont
    Let $j\in [k]$. The shortest paths in \GvertU satisfy the following two properties:
    \begin{itemize}
        \item[(i)] For each $r\in [N]$, the path $\CanVertexU(r\ ;\ c_j - d_j)$ is a shortest $c_j - d_j$ path in \GvertU.
        \item[(ii)] If $P$ is a shortest $c_j - d_j$ path in \GvertU, then $P$ must be an image (\autoref{def:hor-image-of-a-path-G-vert-U}) of the path $\CanInterU(\ell\ ;\ c_j - d_j)$ for some $\ell\in [N]$.
    \end{itemize}
    \label{lem:horizontal-canonical-is-shortest-G-vertex-U}
\end{lemma}

\begin{proof}
    The proof of this lemma is similar to that of \GintU in \autoref{lem:horizontal-canonical-is-shortest-G-int-U}, with some minor observational changes. Note that every path in \GintU contains only \green{green} and black vertices. The splitting operation (\autoref{def:splitting-operation-vertex-U}) applied to each black vertex of \GintU has the following property: if a path $Q$ contains a black vertex $\w$ in \GvertU, then in the corresponding path in \GvertU this vertex $\w$ is \textbf{always replaced by one other vertex} with a cost to visit of~$1$:
    \begin{itemize}
        \item If $\w$ is \notsplit (\autoref{fig:split-vertex-not-U}), then it is replaced in $Q$ the vertex $\w_{\Hor} = \w_{\Ver}$.
        \item If $\w$ is \vertsplit (\autoref{fig:split-vertex-yes-U}), then it is replaced in $Q$ either by the vertex $\w_{\Ver}$ or the vertex $\w_{\Hor}$.
    \end{itemize}
    Therefore, if a path $Q$ incurs a cost of $\alpha$ from visiting \green{green} vertices and a cost of $\beta$ from visiting black vertices in \GintU, then the corresponding path in \GvertU incurs a cost of  $\alpha$ from visiting \green{green} vertices and $\beta$ from black vertices. The proof of the first part of the lemma now follows from~\autoref{lem:horizontal-canonical-is-shortest-G-int-U}(i),~\autoref{def:splitting-operation-vertex-U} and~\autoref{def:hori-canonical-Gvertex-U}. The proof of the second part of the lemma follows from~\autoref{lem:horizontal-canonical-is-shortest-G-int-U}(ii)'s argument that it cannot take an edge that modifies the $y$-coordinate, along with~\autoref{def:splitting-operation-vertex-U} and~\autoref{def:hor-image-of-a-path-G-vert-U}.
\end{proof}
\medskip

The proof of the next lemma is very similar to that of \autoref{lem:horizontal-canonical-is-shortest-G-vertex-U}, and omit the details.
\medskip

\begin{lemma}
    \normalfont
    Let $i\in [k]$. The shortest paths in \GvertU satisfy the following two properties:
    \begin{itemize}
        \item[(i)] For each $r\in [N]$, the path $\CanVertex(r\ ;\ a_i - b_i)$ is a shortest $a_i - b_i$ path in $U_{\vertex}$.
        \item[(ii)] If $P$ is a shortest $a_i - b_i$ path in \GvertU, then $P$ must be an image (\autoref{def:vert-image-of-a-path-G-vert-U}) of the path $\CanInterU(\ell\ ;\ a_i - b_i)$ for some $\ell\in [N]$.
    \end{itemize}
    \label{lem:vertical-canonical-is-shortest-G-vertex-U}
\end{lemma}
\medskip

\subsection{
    \texorpdfstring{\underline{Completeness}: $G$ has a $k$-clique $\Rightarrow$ All pairs in the instance $(U_{\vertex}, \mathcal{T})$ of Undirected-$2k$-VDSP can be satisfied}
    {Completeness: G has a k-clique -> All pairs in the instance (Uvert,T) of Undirected-2k-\VDSP can be satisfied}
}
\label{sec:clique-to-2kvdsp-U}

In this section, we show that if the instance $G$ of \kclique has a solution then the instance $(U_{\vertex}, \mathcal{T})$ of Undirected-$2k$-\VDSP also has a solution. The proofs are very similar to those of
Suppose the instance $G=(V,E)$ of \kclique has a clique $X=\{v_{\gamma_1}, v_{\gamma_2}, \ldots, v_{\gamma_k} \}$ of size $k$. Let $Y=\{\gamma_1, \gamma_2, \ldots, \gamma_k\}\in [N]$. Now for each $i\in [k]$ we choose the path as follows:
\begin{itemize}
    \item The path $R_i$ to satisfy $a_i - b_i$ is chosen to be the horizontal canonical path $\CanVertexU(\gamma_i \ ;\ a_i - b_i)$ described in~\autoref{def:vert-canonical-Gvertex-U}.
    \item The path $T_i$ to satisfy $c_i - d_i$ is chosen to be vertical canonical path  $\CanVertexU(\gamma_i \ ;\ c_i - d_i)$ described in~\autoref{def:hori-canonical-Gvertex-U}.
\end{itemize}

Now we show that the collection of paths given by $\mathcal{Q}:=\{R_1, R_2, \ldots, R_k, T_1, T_2, \ldots, T_K\}$ forms a solution for the instance  $(U_{\vertex}, \mathcal{T})$ of Undirected-$2k$-\VDSP via the following two lemmas which argue being shortest for each terminal pair and pairwise vertex-disjointness respectively:

\medskip
\begin{lemma}
    \normalfont
    For each $i \in [k]$, the path $R_i$ (resp. $T_i$) is a shortest $a_i - b_i$ (resp. $c_i - d_i$) path in $U_{\vertex}$.
    \label{lem:completeness-vdsp-shortest-U}
\end{lemma}
\begin{proof}
    Fix any $i\in [k]$.~\autoref{lem:horizontal-canonical-is-shortest-G-vertex-U}(i) implies that $T_i$ is shortest $c_i - d_i$ path in \Gvert.~\autoref{lem:vertical-canonical-is-shortest-G-vertex-U}(i) implies that $R_i$ is shortest $a_i - b_i$ path in \GvertU.
\end{proof}
\medskip

Before proving~\autoref{lem:completeness-vdsp-disjoint-U}, we first set up notation for some special sets of vertices in \GvertU which helps to streamline some of the subsequent proofs.

\medskip

\begin{definition}
    \label{def:horizontal-vertical-sets-in-Gvertex-U}
    \textbf{(horizontal \& vertical levels in \GvertU)}
    For each $(i,j)\in [k]\times [k]$, let $U_{i,j}^{\vertsplitt}$ to be the graph obtained by applying the splitting operation (\autoref{def:splitting-operation-vertex-U}) to each vertex of $U_{i,j}$. For each $j\in [k]$, we define the following set of vertices:
    \begin{equation}
        \begin{aligned}
            \HorizontalVertex(j) &= \{ c_j, d_j \} \cup \left( \bigcup_{i=1}^{k} V(U_{i,j}^{\vertsplitt})\right)\ \quad\\
            \quad \VerticalVertex(j) &= \{ a_j, b_j \} \cup \left( \bigcup_{i=1}^{k} V(U_{j,i}^{\vertsplitt})\right)
        \end{aligned}\label{eqn:hori-verti-sets-Gvert-u}
    \end{equation}
\end{definition}
\medskip

The next lemma shows that any two paths from $\mathcal{Q}$ are vertex-disjoint.
\medskip

\begin{lemma}
    \normalfont
    Let $P\neq P'$ be any pair of paths from the collection  $\mathcal{Q}=\{R_1, R_2, \ldots, R_k, T_1, T_2, \ldots, T_K\}$. Then $P$ and $P'$ are vertex-disjoint.
    \label{lem:completeness-vdsp-disjoint-U}
\end{lemma}
\begin{proof}

    By~\autoref{def:horizontal-vertical-sets-in-Gvertex-U}, it follows that every edge of the path $R_i$ has both endpoints in $\VerticalVertex(i)$ for every $i\in [k]$. Since $\VerticalVertex(i) \cap \VerticalVertex(i')=\emptyset$ for every $1\leq i\neq i'\leq k$, it follows that the collection of paths $\{R_1, R_2, \ldots, R_k\}$ are pairwise vertex-disjoint.

    By~\autoref{def:horizontal-vertical-sets-in-Gvertex-U}, it follows that every edge of the path $T_j$ has both endpoints in $\HorizontalVertex(j)$ for every $j\in [k]$. Since  $\HorizontalVertex(j) \cap \HorizontalVertex(j')=\emptyset$ for every $1\leq j\neq j'\leq k$, it follows that the collection of paths $\{T_1, T_2, \ldots, T_k\}$ are pairwise vertex-disjoint.

    It remains to show that every pair of paths which contains one path from $\{R_1, R_2, \ldots, R_k\}$ and other path from $\{T_1, T_2, \ldots, T_k\}$ are vertex-disjoint.
    \begin{claim}
        \label{clm:reduction-vertex-disjoint-U}
        \normalfont
        For each $(i,j)\in [k]\times [k]$, the paths $R_i$ and $T_j$ are vertex-disjoint in $U_{\vertex}$.
    \end{claim}

    \begin{proof}
        Fix any $(i,j)\in [k]\times [k]$. First we argue that the vertex $\w_{i,j}^{\gamma_i, \gamma_j}$ is \vertsplit, i.e., $(\gamma_i, \gamma_j)\in S_{i,j}$:
        \begin{itemize}
            \item If $i=j$ then $\gamma_i = \gamma_j$ and hence by~\autoref{eqn:clique-to-gt-reduction-U} we have $(\gamma_i, \gamma_j)\in S_{i,j}$
            \item If $i\neq j$, then $v_{\gamma_i} - v_{\gamma_j}\in E(G)$ since $X$ is a clique. Again, by~\autoref{eqn:clique-to-gt-reduction-U} we have $(\gamma_i, \gamma_j)\in S_{i,j}$.
        \end{itemize}
        Hence, by~\autoref{def:splitting-operation-vertex-U}, it follows that the vertex $\w_{i,j}^{\gamma_i, \gamma_j}$ is \vertsplit, i.e., $\w_{i,j,\Hor}^{\gamma_i, \gamma_j}\neq \w_{i,j,\Ver}^{\gamma_i, \gamma_j}$.

        By the construction of \GintU (\autoref{fig:mainUndir}) and definitions of canonical paths (\autoref{def:hori-canonical-G-int-U} and~\autoref{def:vert-canonical-G-int-U}), it is easy to verify that any pair of horizontal canonical path and vertical canonical path in \GintU have only one vertex in common.

        By the splitting operation (\autoref{def:splitting-operation-vertex-U}) and definitions of the paths $R_i$ (\autoref{def:vert-canonical-Gvertex-U}) and $T_j$ (\autoref{def:hori-canonical-Gvertex-U}), it follows that
        \begin{itemize}
            \item $R_i$ contains $\w_{i,j,\Ver}^{\gamma_i, \gamma_j}$ but does not contain $\w_{i,j,\Hor}^{\gamma_i, \gamma_j}$
            \item $T_j$ contains $\w_{i,j,\Hor}^{\gamma_i, \gamma_j}$ but does not contain $\w_{i,j,\Ver}^{\gamma_i, \gamma_j}$
        \end{itemize}

    \end{proof}
    This concludes the proof of~\autoref{lem:completeness-vdsp-disjoint-U}.
\end{proof}
\medskip

\noindent From~\autoref{lem:completeness-vdsp-shortest-U} and~\autoref{lem:completeness-vdsp-disjoint-U}, it follows that the collection of paths given by $\mathcal{Q}=\{R_1, R_2, \ldots, R_k,$ $T_1, T_2, \ldots, T_K\}$ forms a solution for the instance $(U_{\vertex}, \mathcal{T})$ of Undirected-$2k$-\VDSP.

\subsection{
    \texorpdfstring{\underline{Soundness}: $(\frac{1}{2} +\epsilon)$-fraction of the pairs in the instance $(U_{\vertex}, \mathcal{T})$ of $2k$-VDSP can be satisfied $\Rightarrow$ $G$ has a clique of size $\geq 2\epsilon \cdot k$}
    {Soundness: (1/2 + theta)-fraction of the pairs in the instance (Uvert,T) of 2k-\VDSP can be satisfied -> G has a clique of size >= 2 x theta x k}}
\label{sec:2kvdsp-to-clique-U}

In this section we show that if at least $(\frac{1}{2} +\epsilon)$-fraction of the $2k$ pairs from the instance $(U_{\vertex}, \mathcal{T})$ of $2k$-\VDSP can be satisfied then the graph $G$ has a clique of size $2\epsilon \cdot k$.

Let $\mathcal{P}$ be a collection of paths in $U_{\vertex}$ which satisfies at least $(\frac{1}{2} +\epsilon)$-fraction of the $2k$ terminal pairs from the instance $(U_{\vertex}, \mathcal{T})$ of Undirected-$2k$-\VDSP.

\medskip
\begin{definition}
    An index $i \in [k]$ is called \emph{good} if both the terminal pairs $a_i - b_i$ and $c_i - d_i$ are satisfied by $\mathcal{P}$.
    \label{def:good-for-vertex-U}
\end{definition}
\medskip

The proof of the next lemma, which gives a lower bound on the number of good indices, is exactly the same as that of~\autoref{lem:good_size-U} and we do not repeat it here.

\medskip
\begin{lemma}
    \label{lem:good_size-vertex-U}
    Let $Y \subseteq [k]$ be the set of good indices. Then $|Y| \ge 2\epsilon \cdot k$.
\end{lemma}
\medskip

\begin{lemma}
    \label{lem:good-equal-vertex-U}
    If $i\in [k]$ is good, then there exists $\delta_i \in [N]$ such that the two paths in $\mathcal{P}$ satisfying $a_i - b_i$ and $c_i - d_i$ in \GedgeU are the vertical canonical path $\CanVertexU(\delta_i\ ;\ a_i - b_i)$ and the horizontal canonical path $\CanVertexU(\delta_i\ ;\ c_i - d_i)$ respectively.
\end{lemma}
\begin{proof}
    If $i$ is good, then by~\autoref{def:good-for-vertex-U} both the pairs $a_i - b_i$ and $c_i - d_i$ are satisfied by $\mathcal{P}$. Let $P_1, P_2\in \mathcal{P}$ be the paths that satisfy the terminal pairs $(a_i, b_i)$ and $(c_i, d_i)$ respectively.
    Since $P_1$ is a shortest $a_i - b_i$ path in \GvertU, by~\autoref{lem:vertical-canonical-is-shortest-G-vertex-U}(ii) it follows that $P_1$ is the vertical canonical path $\CanVertexU(\alpha\ ;\ a_i - b_i)$ for some $\alpha\in [N]$. Since $P_2$ is a shortest $c_i - d_i$ path in \GvertU, by~\autoref{lem:horizontal-canonical-is-shortest-G-vertex-U}(ii) it follows that $P_2$ is the horizontal canonical path $\CanVertexU(\beta\ ;\ c_i - d_i)$ for some $\beta\in [N]$.

    Using the fact that $P_1$ and $P_2$ are vertex-disjoint in \GvertU, we now claim that $\w_{i,i}^{\alpha,\beta}$ is \vertsplit:
    \begin{claim}
        The vertex $\w_{i,i}^{\alpha,\beta}$ is \vertsplit by the splitting operation of~\autoref{def:splitting-operation-vertex-U}.
        \label{clm:must-be-vert-split-vertex-i-i-U}
    \end{claim}
    \begin{proof}
        By~\autoref{def:splitting-operation-vertex-U}, every black vertex of \GintU is either \vertsplit or \notsplit. If $\w_{i,i}^{\alpha,\beta}$ was \notsplit (\autoref{fig:split-vertex-not-U}), then by~\autoref{def:hori-canonical-Gvertex-U} and~\autoref{def:vert-canonical-Gvertex-U}, the vertex $\w_{i,i,\Hor}^{\alpha,\beta}=\w_{i,i,\Ver}^{\alpha,\beta}$ belongs to both $P_1$ and $P_2$ contradicting the fact that they are vertex-disjoint.
    \end{proof}
    By~\autoref{clm:must-be-vert-split-vertex-i-i-U}, we know that the vertex $\w_{i,i}^{\alpha,\beta}$ is \vertsplit. Hence, from~\autoref{eqn:clique-to-gt-reduction-U} and~\autoref{def:splitting-operation-vertex-U}, it follows that $\alpha=\beta$ which concludes the proof of the lemma.
\end{proof}
\medskip

\begin{lemma}
    \label{lem:good_vertices-U}
    If both $i,j \in [k]$ are good and $i \neq j$, then $v_{\delta_i}-v_{\delta_j} \in E(G)$.
\end{lemma}
\begin{proof}
    Since $i$ and $j$ are good, by~\autoref{def:good-for-vertex-U}, there are paths $Q_1, Q_2 \in \mathcal{P}$ satisfying the pairs $(a_i, b_i), (c_j, d_j)$ respectively. By~\autoref{lem:good-equal-vertex-U}, it follows that
    \begin{itemize}
        \item $Q_1$ is the vertical canonical path $\CanVertexU(\delta_i\ ;\ a_i - b_i)$.
        \item $Q_2$ is the horizontal canonical path $\CanVertexU(\delta_j\ ;\ c_j - d_j)$.
    \end{itemize}

    Using the fact that $Q_1$ and $Q_2$ are vertex-disjoint in \GvertU, we now claim that $\w_{i,j}^{\delta_i,\delta_j}$ is \vertsplit:
    \begin{claim}
        The vertex $\w_{i,j}^{\delta_i,\delta_j}$ is \vertsplit by the splitting operation of~\autoref{def:splitting-operation-vertex-U}.
        \label{clm:must-be-vert-split-edge-i-j-U}
    \end{claim}
    \begin{proof}
        By~\autoref{def:splitting-operation-vertex-U}, every black vertex of \GintU is either \vertsplit or \notsplit. If $\w_{i,j}^{\delta_j,\delta_j}$ was \notsplit (\autoref{fig:split-vertex-not-U}), then by~\autoref{def:hori-canonical-Gvertex-U} and~\autoref{def:vert-canonical-Gvertex-U}, the vertex $\w_{i,j,\Hor}^{\delta_i,\delta_j}=\w_{i,j,\Ver}^{\delta_i,\delta_j}$ belongs to both $Q_1$ and $Q_2$ contradicting the fact that they are vertex-disjoint
    \end{proof}
    By~\autoref{clm:must-be-vert-split-edge-i-j-U}, we know that the vertex $\w_{i,j}^{\delta_i,\delta_j}$ is \vertsplit. Since $i\neq j$, from~\autoref{eqn:clique-to-gt-reduction-U} and~\autoref{def:splitting-operation-vertex-U}, it follows that $v_{\delta_i}-v_{\delta_j}\in E(G)$ which concludes the proof of the lemma.

\end{proof}
\medskip

From~\autoref{lem:good_size-vertex-U} and~\autoref{lem:good_vertices-U}, it follows that the set $X:=\{v_{\delta_i}\ : i\in Y\}$ is a clique of size $\geq (2\epsilon)k$ in $G$.

\subsection{Proof of \autoref{thm:inapprox-vertex-result-U} and \autoref{thm:hardness-vertex-result-U}}
\label{sec:proof-of-main-theorem-vertex-U}

Finally we are ready to prove \autoref{thm:inapprox-vertex-result-U} and \autoref{thm:hardness-vertex-result-U}, which are restated below.
\medskip

\apxundirectedvertexthm*
\exactundirectedvertexthm*
\begin{proof}{\textbf{\autoref{thm:hardness-vertex-result-U}}}

    Given an instance $G$ of \kclique, we can use the construction from \autoref{sec:construction-of-Gvert-U} to build an instance $(U_{\vertex}, \mathcal{T})$ of Undirected-$2k$-\VDSP such that $U_{\vertex}$ is a $1$-planar graph (\autoref{clm:GvertexU-is-1-planar}). The graph \GvertU has $n=O(N^2 k^2)$ vertices (\autoref{clm:size-of-G-vertex-U}), and it is easy to observe that it can be constructed from $G$ (via first constructing \GintU) in $\poly(N,k)$ time.

    It is known that \kclique is W[1]-hard parameterized by $k$, and under ETH cannot be solved in $f(k)\cdot N^{o(k)}$ time for any computable function $f$~\cite{chen-hardness}. Combining the two directions from~\autoref{sec:2kvdsp-to-clique-U} (with $\epsilon = 0.5$) and~\autoref{sec:clique-to-2kvdsp-U}  we obtain a parameterized reduction from an instance $(G,k)$ of \kclique with $N$ vertices to an instance $(U_{\vertex}, \mathcal{T})$ of Undirected-$2k$-\VDSP where \GvertU is a $1$-planar graph (\autoref{clm:GvertexU-is-1-planar}) and has $O(N^2 k^2)$ vertices (\autoref{clm:size-of-G-vertex-U}). As a result, it follows that Undirected-$k$-\VDSP on $1$-planar graphs is W[1]-hard parameterized by number $k$ of terminal pairs, and under ETH cannot be solved in $f(k)\cdot n^{o(k)}$ time where $f$ is any computable function and $n$ is the number of vertices.
\end{proof}

\medskip

\begin{proof}{\textbf{\autoref{thm:inapprox-vertex-result-U}}}

    Let $\delta$ and $r_0$ be the constants from~\autoref{thm:cli_inapprox}. Fix any constant $\epsilon\in (0,1/2]$. Set $\zeta = \dfrac{\delta \epsilon}{2}$ and $k=\max \Big\{\dfrac{1}{2\zeta} , \dfrac{r_0}{2\epsilon}\Big\}$.

    Suppose to the contrary that there exists an algorithm $\mathbb{A}_{\VDSP}$ running in $f(k)\cdot n^{\zeta k}$ time (for some computable function $f$) which given an instance of Undirected-$k$-\VDSP with $n$ vertices can distinguish between the following two cases:
    \begin{itemize}
        \item[(1)] All $k$ pairs of the Undirected-$k$-\VDSP instance can be satisfied
        \item[(2)] The max number of pairs of the Undirected-$k$-\VDSP instance that can be satisfied is less than $(\frac{1}{2}+\epsilon)\cdot k$
    \end{itemize}
    We now design an algorithm $\mathbb{A}_{\Clique}$ that contradicts~\autoref{thm:cli_inapprox} for the values $q=k$ and $r=(2\epsilon)k$. Given an instance of $(G,k)$ of \kclique with $N$ vertices, we apply the reduction from~\autoref{sec:construction-of-Gvert-U} to construct an instance $(U_{\vertex}, \mathcal{T})$ of Undirected-$2k$-\VDSP where \GvertU has $n=O(N^2 k^2)$ vertices (\autoref{clm:size-of-G-vertex-U}). It is easy to see that this reduction takes $O(N^2 k^2)$ time as well. We now show that the number of pairs which can be satisfied from the Undirected-$2k$-\VDSP instance is related to the size of the max clique in $G$:
    \begin{itemize}
        \item If $G$ has a clique of size $q=k$, then by~\autoref{sec:clique-to-2kvdsp-U} it follows that all $2k$ pairs of the instance $(U_{\vert},\mathcal{T})$ of Undirected-$2k$-\VDSP can be satisfied.
        \item If $G$ does not have a clique of size $r=2\epsilon k$, then we claim that the max number of pairs in $\mathcal{T}$ that can be satisfied is less than $(\frac{1}{2}+\epsilon)\cdot 2k$. This is because if at least $(\frac{1}{2}+\epsilon)$-fraction of pairs in $\mathcal{T}$ could be satisfied then by~\autoref{sec:2kvdsp-to-clique-U} the graph $G$ would have a clique of size $\geq (2\epsilon) k=r$.
    \end{itemize}
    Since the algorithm $\mathbb{A}_{\VDSP}$ can distinguish between the two cases of all $2k$-pairs of the instance $(U_{\vertex}, \mathcal{T})$ can be satisfied or only less than $(\frac{1}{2}+\epsilon)\cdot2k$ pairs can be satisfied, it follows that $\mathbb{A}_{\Clique}$ can distinguish between the cases $\Clique(G)\geq q$ and $\Clique(G)<r$.

    The running time of the algorithm $\mathbb{A}_{\Clique}$ is the time taken for the reduction from~\autoref{sec:construction-of-Gvert-U} (which is $O(N^2 k^2)$) plus the running time of the algorithm $\mathbb{A}_{\VDSP}$ which is $f(2k)\cdot n^{\zeta\cdot 2k}$. It remains to show that this can be upper bounded by $g(q,r)\cdot N^{\delta r}$ for some computable function $g$:
    \begin{align*}
        & O(N^2 k^2) + f(2k)\cdot n^{\zeta\cdot 2k}\\
        &\leq  c\cdot N^2 k^2 + f(2k)\cdot d^{\zeta\cdot 2k}\cdot (N^2 k^2)^{\zeta\cdot 2k} \tag{for some constants $c,d\geq 1$: this follows since $n=O(N^2 k^2)$}\\
        &\leq c\cdot N^2 k^2 + f'(k)\cdot N^{2\zeta\cdot 2k} \tag{where $f'(k)=f(2k)\cdot d^{\zeta\cdot 2k}\cdot k^{2\zeta\cdot 2k}$}\\
        &\leq 2c\cdot f'(k)\cdot N^{2\zeta\cdot 2k} \tag{since $4\zeta k\geq 2$ implies $f'(k)\geq k^2$ and $N^{2\zeta\cdot 2k}\geq N^2$}\\
        &= 2c\cdot f'(k) \cdot N^{\delta r} \tag{since $\zeta=\frac{\delta \epsilon}{2}$ and $r=(2\epsilon)k$}
    \end{align*}
    Hence, we obtain a contradiction to~\autoref{thm:cli_inapprox} with $q=k, r=(2\epsilon)k$ and $g(k)=2c\cdot f'(k)=2c\cdot f(2k)\cdot d^{\zeta\cdot 2k}\cdot k^{2\zeta\cdot 2k}$.
\end{proof}

	\section{Conclusion \& Open Problems}
	\label{sec:app-open-problems}
	In this paper, we obtained approximate and exact lower bounds for all four variants of the $k$-\dsp problem.
	We leave open the following natural questions:
	\begin{itemize}
		\item Can we improve on the $(\frac{1}{2}+\epsilon)$ factor of the FPT inapproximability results for \edsp and \vdsp on planar and $1$-planar graphs respectively? Perhaps this could be achieved by modifying the reduction for the $o(k)$ factor lower bound on general graphs by Bentert et al.~\cite{bentertTightApproximationKernelization2024}.
		\item Can we obtain inapproximability lower bounds also for the $k$-\disjp problem on directed graphs, possibly on graph classes such as DAGs or planar graphs for which FPT or XP algorithms are known~\cite{DBLP:conf/focs/CyganMPP13,DBLP:journals/tcs/FortuneHW80}?
		\item \autoref{thm:hardness-vertex-result} gives W[1] hardness and, under ETH, an $f(k)\cdot n^{o(k)}$ lower bound for \vdsp on directed $1$-planar graphs.
		Can we get equivalent lower bounds for planar graphs, which would show Bérczi and Kobayashi's $n^{O(k)}$-time algorithm~\cite{berczi-kobayashi} to be tight?
		Alternatively, is an FPT algorithm for the problem possible by either adapting Cygan et al.'s $2^{2^{O(k^2)}}\cdot n^{O(1)}$-time algorithm for planar \vdp~\cite{DBLP:conf/focs/CyganMPP13}, or through an entirely new technique.
	\end{itemize}

	\bibliographystyle{splncs04}
	\bibliography{papers}

\end{document}